\documentclass[journal]{IEEEtran}

\usepackage[cmex10]{amsmath}
\usepackage{epsfig,subfigure,amssymb,color,multirow,array}
\usepackage{algorithm,algorithmic}
\usepackage[all,2cell]{xy} \UseAllTwocells \SilentMatrices
\DeclareGraphicsExtensions{.pdf}
\newtheorem{defn}{Definition}
\newtheorem{thm}{Theorem}
\newtheorem{cor}{Corollary}
\newtheorem{prop}{Proposition}
\newtheorem{lemma}{Lemma}

\newtheorem{ex}{Example}
\usepackage{cite}

\usepackage{verbatim}
\newenvironment{proof}{\begin{IEEEproof}}{\end{IEEEproof}}

\newcommand{\qed}              {\hfill \IEEEQED}
\newcommand  {\ab}        	{{\boldsymbol a}}
\newcommand  {\bb}        	{{\boldsymbol b}}
\newcommand  {\cb}        	{{\boldsymbol c}}
\newcommand  {\db}        	{{\boldsymbol d}}

\newcommand  {\vb}        	{{\boldsymbol v}}
\newcommand  {\xb}        	{\underline{x}}
\newcommand  {\yb}        	{\underline{y}}
\newcommand  {\kb}        	{{\boldsymbol k}}

\newcommand  {\ib}                {{\boldsymbol i}}
\newcommand  {\jb}                {{\boldsymbol j}}

\newcommand  {\wb}        	{{\boldsymbol w}}
\newcommand  {\ub}        	{{\boldsymbol u}}
\newcommand  {\eb}        	{\underline {e}}

\newcommand  {\rb}        	{{\boldsymbol r}}
\newcommand {\C}		{\mathbb C}

\newcommand{\mI}		{{\tt{I}}}

\newcommand{\Dm}		{{\tt{D}}}
\newcommand{\Am}		{{\tt{A}}}

\newcommand{\Hm}		{{\tt{H}}}

\newcommand{\Sf}		{{\frak{S}}}
\newcommand{\Cf}		{{\frak{C}}}

\newcommand{\rank}{\textnormal{rank}}
\newcommand{\adj}{\textnormal{adj}}
\renewcommand{\det}{\textnormal{det}}
\newcommand{\abs}[1]{\left| #1 \right|}
\def\beq{\begin{equation}}
\def\eeq{\end{equation}}
\def\bea{\begin{IEEEeqnarray}{rCl}} 
\def\eea{\end{IEEEeqnarray}}
\def\bean{\begin{IEEEeqnarray*}{rCl}} 
\def\eean{\end{IEEEeqnarray*}} 
\def\ben{\begin{enumerate}}
\def\een{\end{enumerate}}
\def\no{\IEEEnonumber} 
\def\yes{\IEEEyesnumber} 
\def\R{\mathbb R}

\newcommand{\F}{{\mathbb F}}

\begin{document}
\bibliographystyle{ieeetran}
\allowdisplaybreaks

\title{Analysis and Practice of Uniquely Decodable One-to-One Code}

\author{Chin-Fu Liu,
     Hsiao-feng (Francis) Lu, and
     Po-Ning Chen
\thanks{The authors are with the Department of 
Electrical and Computer Engineering, 
National Chiao-Tung University (NCTU), Hsinchu 30010, Taiwan (e-mails: hubert.liu.1031@gmail.com, francis@mail.nctu.edu.tw, poning@faculty.nctu.edu.tw). They are also with the Center of Information and Communications Technology of NCTU, Taiwan.}
}

\maketitle

\begin{abstract}
In this paper, we consider the so-called uniquely decodable one-to-one code (UDOOC) 
that is formed by inserting a ``comma" indicator, termed the {\em unique word} (UW), between consecutive one-to-one codewords for separation. Along this research direction,
we first investigate several general combinatorial properties of UDOOCs,
in particular the enumeration of the number of UDOOC codewords for any (finite) codeword length.
Based on the obtained formula on the number of length-$n$ codewords for a given UW, the per-letter average codeword length of UDOOC for the optimal compression of a given source statistics can be computed.
Several upper bounds on the average codeword length of such UDOOCs are next established.
The analysis on the bounds of average codeword length
then leads to two asymptotic bounds for sources having infinitely many alphabets,
one of which is achievable and hence tight for a certain source statistics and UW,
and the other of which proves the achievability of source entropy rate of UDOOCs
when both the block size of source letters for UDOOC compression and
UW length go to infinity.
Efficient encoding and decoding algorithms
for UDOOCs are also given in this paper.
Numerical results show that when grouping three English letters as a block,
the UDOOCs with $\text{UW}=0001$, $0000$, $000001$ and $000000$ 
can respectively reach the compression rates of $3.531$, $4.089$, $4.115$, $4.709$ 
bits per English letter (with the lengths of UWs included), where 
the source stream to be compressed is the book titled \emph{Alice's Adventures in Wonderland}. In comparison
with the first-order Huffman code, the second-order Huffman code, the third-order Huffman code  and the Lempel-Ziv code, which respectively achieve the compression rates of $3.940$, $3.585$, $3.226$ and $6.028$ bits per single English letter, the proposed UDOOCs can potentially result in comparable compression rate to the Huffman code under similar decoding complexity and yield
a smaller average codeword length than that of the Lempel-Ziv code, thereby confirming 
the practicability of UDOOCs.
\end{abstract}

\section{Introduction} 

The investigation of lossless source coding can be roughly classified into two categories,
one for the compression of a sequence of source letters and the other for a single ``one shot" source symbol \cite{Cover}.
A well-known representative for the former is the Huffman code, while the latter
is usually referred to as the one-to-one code (OOC).

The Huffman code is an optimal entropy code that can achieve the minimum average codeword length for a given statistics of source letters. 
It obeys the rule of unique decodability and hence the concatenation of Huffman codewords
can be uniquely recovered by the decoder. 
Although optimal in principle, it may encounter several obstacles in implementation.
For example, the rare codewords are exceedingly long in length, thereby hampering
the efficiency of decoding.
Other practical obstacles include 
\begin{enumerate}
\item[i)] the codebook needs to be pre-stored for encoding and decoding, which might demand a large memory space for sources with moderately large alphabet size, 
\item[ii)] the decoding of a sequence of codewords must be done in sequential, not in parallel, and 
\item[iii)] erroneous decoding of one codeword could affect the decoding of subsequent codewords, i.e., error propagation.
 \end{enumerate}
In contrast to unique decodability,
the OOC only requires an assignment of distinct codewords 
to the source symbols.
It has been studied since 1970s \cite{LEN72} and is shown
to achieve an average codeword length 
smaller than the source entropy
minus a nontrivial amount of quantity called \emph{anti-redundancy} \cite{LEN08O}.
Various research works over the years have shown that the anti-redundancy can be as large as the logarithm of the source entropy 
\cite{LEN94,LEN96,LEN07,LEN80,LEN78,LEN82,LEN08J,LEN04,LEN08O,LEN09,LEN86}
In comparison with an entropy coding like Huffman code, 
the codewords of an OOC can be sequenced alphabetically 
and hence the practice of an OOC is generally considered to be more computationally convenient.

A question that may arise from the above discussion is whether
we could add a ``comma'' indicator, termed \emph{Unique Word} (UW) in this paper, 
in-between consecutive OOC codewords, 
and use the OOC for the lossless compression of a sequence of source letters.
A direct merit of such a structure is that the alphabetically sequenced OOC codewords can be manipulated without {\em a priori} stored codebook at both the encoding and decoding ends.
This is however achieved at a price of an additional constraint that
the ``comma" indicator must not appear as an internal subword\footnote{
We say $\ab=a_1 \ldots a_m$ is not an \emph{internal subword} of $\bb=b_1 \ldots b_n$ if there does not exist $i$ such that $b_i \ldots b_{i+m-1}=\ab$ for all $1 < i < n-m+1$.
When the same condition holds for all $1 \leq i \leq n-m+1$ (i.e., with two equalities), we say $\ab$ is not a \emph{subword} of $\bb$.
} in 
the concatenation of either an OOC codeword with a comma indicator, or a comma indicator with an OOC codeword.

On the one hand, this additional constraint facilitates the fast identification of OOC codewords in a coded bit-stream 
and makes feasible the subsequent parallel decoding of them.
On the other hand, the achievable average codeword length of a UW-forbidden OOC
may increase significantly for a bad choice of UWs.
Therefore, it is of theoretical importance 
to investigate the minimum average codeword length of a UW-forbidden OOC, 
in particular the selection of a proper UW that could minimize this quantity.
Since the resultant UW-forbidden OOC coding system satisfies unique decodability (UD), we will refer to it conveniently as the UDOOC in the sequel.

We would like to point out that the conception of inserting UWs between consecutive words might not be new in existing applications. 
For example, in the IEEE 802.11 standard for wireless local area networks \cite{std}, an entity similar to the UW in a bit-stream has been specified as a boundary indicator for a frame, or as a synchronization 
support, or as a part of error control mechanism.
In written English, punctuation marks and spacing are essential to disambiguate the meaning of sentences. However, a complete theoretical study of the UDOOC conception remains undone.
This is therefore the main target of this paper.
We now give a formal definition of binary UDOOCs.

\medskip

\begin{defn} \label{defn:udooc}
Given UW $\kb=k_1 k_2 \ldots k_L \in \F\times\cdots\times \F=\F^L$, 
where $\F=\{0,1\}$, we say ${\cal C}_\kb(n)$ is a UDOOC of length $n\geq 1$ associated with $\kb$ if it contains all binary length-$n$ tuples ${\boldsymbol b}=b_1 \ldots b_n$ such that $\kb$ is not an internal subword of the concatenated bit-stream $\kb {\boldsymbol b} \kb$. As a special case, we set ${\cal C}_\kb(0) :=\{\text{null}\}$.\footnote{In our binary UDOOC, 
it is allowed to place two UWs side-by-side with 
nothing in-between in order to produce a \emph{null} codeword.
} 
The overall UDOOC associated with $\kb$, denoted by ${\cal C}_\kb$, is given by
\[
{\cal C}_\kb \ := \ \bigcup_{n \geq 0} {\cal C}_\kb(n).
\]
\end{defn}

\medskip

For a better comprehension of Definition \ref{defn:udooc}, we next give an example to illustrate 
how a UDOOC is generated and how it is used in encoding and decoding.
The way to count the number of length-$n$ UDOOC codewords will follow.

\medskip

\begin{ex}\label{ex1}
Suppose the UW $\kb=00$ is chosen. In order to prohibit the concatenation of any UDOOC codeword and UW, regardless of the ordering, from containing 00 as an internal subword, the following constraints must be satisfied. 
\begin{itemize}
\item Type-I constraints: The UW cannot be a subword of any UDOOC codeword. This means that within any codeword of length $n\geq1$:
\begin{itemize}
\item[] (C1) ``0'' can only be followed by ``1''.
\item[] (C2)  ``1'' can be followed by either ``0'' or ``1''.
\end{itemize}
\item Type-II constraints: Besides the type-I constraints, the UW cannot appear as an internal subword, containing the boundary of any UDOOC codeword and UW, regardless of the ordering. This implies that except for the ``null" codeword:
\begin{itemize}
\item[] (C3) The first bit of a codeword cannot be ``0''.
\item[] (C4) The last bit of a codeword cannot be ``0''.
\end{itemize}
\end{itemize}

\begin{figure}
\begin{center}
\includegraphics[width=0.5\textwidth]{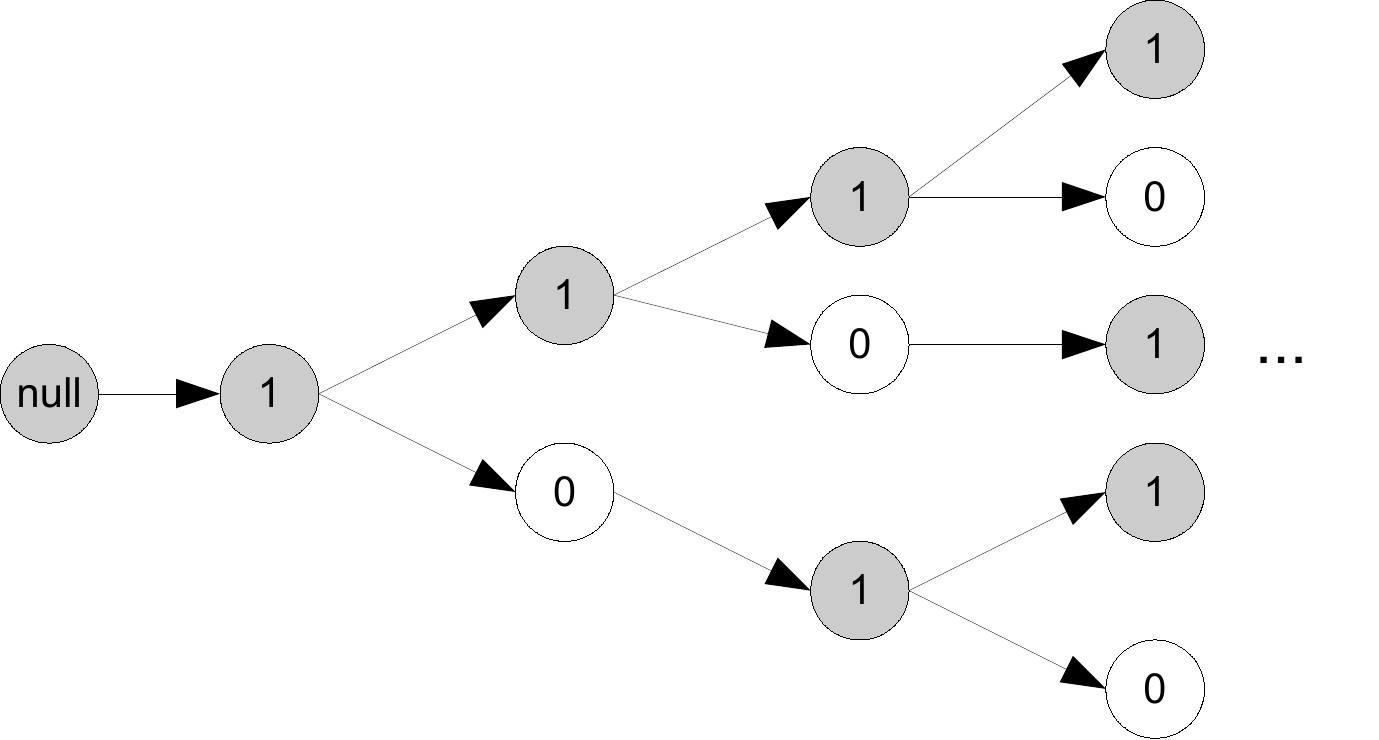}
\caption{UDOOC code tree for $\text{UW}=00$.}\label{00pctree}
\end{center}
\end{figure}

By Constraints (C1)-(C4), we can place the UDOOC codewords
on a code tree as shown in Fig. \ref{00pctree},
in which each path starting from the root node and ending
at a gray-shaded node corresponds to a codeword.
Thus, the codewords for $\text{UW}=00$
include null, 1, 11, 101, 111, 1011, 1101, 1111, etc.
It should be noted that we only show the codewords of length up to four, while the code tree actually can grow indefinitely in depth.

At the decoding stage, suppose the received bit-stream is 00100110010100111100,
where we add UWs at both the left and the right ends to indicate the margins of the bit-stream.
This may facilitate, for example, noncoherent bit-stream transmission.  
Then, the decoder first locates
UWs and parses the bit-stream into separate codewords as 1, 11, 101 and 1111,
after which the four codewords can be decoded separately (possibly in parallel) to their respective source symbols.

With the code tree representation, the number of length-$n$ codewords in a UDOOC code tree can be straightforwardly calculated. Let the ``null"-node be placed at level $0$.
For $n\geq 1$, denote by $a_n$ and $b_n$ the numbers of ``1"-nodes and ``0"-nodes at the $n$th level of the code tree, respectively.
By the two type-I constraints, the following recursions 
hold:
$$
\left\{\begin{array}{lcl}
a_n &=& a_{n-1} + b_{n-1}\\
b_n &=& a_{n-1}
\end{array}\right.\text{ for }n\geq 2.
$$
With the initial values of $a_1=1$ and $a_2=1$,
it  follows that 
$\{a_n\}_{n=1}^\infty$ is the renowned Fibonacci sequence \cite{AMM59}, i.e., $a_n = a_{n-1}+a_{n-2}$
for $n\geq 3$. 
This result, together with the two type-II constraints, implies
that the number of length-$n$ codewords is $|{\cal C}_{00}(n)|=a_n$ for $n \geq 1$, which
according to the Fibonacci recursion is given by:
\[
 a_n \ = \  \frac{\varphi^n - \bar{\varphi}^n }{\sqrt{5}},
\]
where $\varphi = \frac{1+\sqrt{5}}{2}$ is the Golden ratio and $\bar{\varphi} = \frac{1-\sqrt{5}}{2}$ is the Galois conjugate of $\varphi$ in number field ${\mathbb Q}(\sqrt{5})$. 
Thus, $|{\cal C}_{00}(n)|$ grows exponentially in $n$ with base $\varphi\approx 1.618$.
\qed
\end{ex}

\medskip

We can similarly examine the choice of $\text{UW}=01$ and draw the respective code tree in Fig.~\ref{01pctree}, where its type-I constraints become:
\begin{itemize}
\item[] (C1) ``0'' can only be followed by ``0''.
\item[] (C2)  ``1'' can be followed by either ``0'' or ``1''.
\end{itemize}
and no type-II constraints are required.
We then obtain
$$
\left\{\begin{array}{lcl}
a_n &=& a_{n-1}\\
b_n &=& a_{n-1} + b_{n-1}
\end{array}\right.\text{ for }n\geq 2,
$$
and $|{\cal C}_{01}(n)|=a_n+b_n=n+1$.
Although from Figs.~\ref{00pctree} and \ref{01pctree}, taking $\text{UW}=01$ seems to provide more codewords than taking $\text{UW}=00$ at small $n$, the linear growth of $|{\cal C}_{01}(n)|$ with respect to codeword length $n$ suggests that such choice is not as good as the choice of $\text{UW}=00$ when $n$ is  moderately large.

The above two exemplified UWs point to an important fact that
the best UW, which minimizes the average codeword length,
depends on the code size required.
Thus, the investigation of the efficiency of a UW 
may need to consider the transient superiority in addition to claiming the asymptotic 
winner. 

In this paper, we provide efficient encoding and decoding algorithms
for UDOOCs, and investigate their general combinatorial properties, in particular
the enumeration of the number of codewords for any (finite) codeword length.
Based on the obtained formula for $|{\cal C}_\kb(n)|$, i.e., the number of length-$n$ codewords for a given UW $\kb$,
the average codeword length of the optimal compression of a given source statistics using UDOOC can be computed.
Classifications of UWs are followed, where two types of equivalences are specified,
which are (exact) \emph{equivalence} and \emph{asymptotic equivalence}.
UWs that are equivalent in the former sense are required to yield exactly the same minimum average codeword length for every source statistics, while asymptotic equivalence
only dictates the UWs to result in the same asymptotic growth rate as codeword length approaches infinity. Enumeration of the number of asymptotic equivalent UW classes are then studied with the help of methodologies in \cite{MATH81} and \cite{MATH03}.
Furthermore, three upper bounds on the average codeword length of UDOOCs are established. The first one is a general upper bound when only the largest probability of source symbols is given.
The second upper bound refines the first one under the premise that
the source entropy is additionally known. When both the largest and second largest probabilities of source symbols are present apart from the source entropy,
the third upper bound can be used. Since these bounds are derived in terms of different techniques, actually none of the three bounds dominates 
the other two for all statistics. Comparison of these bounds for an English text 
with statistics from \cite{oxfdic}
and that with statistics from the book \emph{Alice's Adventures in Wonderland} will be accordingly provided.
The analysis on bounds of the average codeword length
gives rise to two asymptotic bounds on ultimate per-letter average codeword length,
one of which is tight for a certain choice of source statistics and UW,
and the other of which  leads to 
the achievability of the ultimate per-letter average codeword length 
to the source entropy rate when both the source block length for compression  
and UW length tend to infinity.

It may be of interest to note that the enumeration of the number of codewords, i.e., $c_{\kb,n}=|{\cal C}_\kb(n)|$,
is actually obtained indirectly via the determination of an auxiliary quantity $s_{\kb,n}$, which is the number of words satisfying the type-I constraints but not necessarily the type-II constraints.
By utilizing the Goulden-Jackson cluster method 
\cite{LNMATH,DIFF10,SP98,DIFF99,DIFF06},
an explicit formula for $s_{\kb,n}$ can be established. The desired enumeration formula for the number of length-$n$ UDOOC codewords
is then obtained by proving that both
the so-called linear constant coefficient difference equation (LCCDE)
and the asymptotic growth rate of $s_{\kb,n}$ and $c_{\kb,n}$ 
are identical.
We next show based on the obtained formula that  the all-zero UW has the largest asymptotic growth rate among all UWs of the same length, while the UW with the smallest growth rate is $00\ldots 01$. Interestingly, the all-zero UW is often the one that yields the smallest $c_{\kb,n}$ for small $n$, in contrast to UW $00\ldots 01$, whose $c_{\kb,n}$ tops all other UWs 
when $n$ is small.
We afterwards demonstrate by using these two special UWs that the general encoding and decoding algorithms can be considerably simplified when further taking into consideration 
the structure of particular UWs.
A side result from the enumeration of $c_{\kb,n}$ is that for all UWs, 
the codeword growth rate of UDOOCs will tend to $|\F|=2$ as the length of the UW goes to infinity. 

With regard to the compression performance of the proposed UDOOCs, numerical results show that when grouping three English letters as a block and separating
the consecutive blocks by UWs,
the UDOOCs with $\text{UW}=0001$, $0000$, $000001$ and $000000$ 
can respectively reach the compression rates of $3.531$, $4.089$, $4.115$, $4.709$ 
bits per English letter (with the length of UWs included), where 
the source stream to be compressed is the book titled \emph{Alice's Adventures in Wonderland}. In comparison
with the first-order Huffman code, the second-order Huffman code, the third-order Huffman code\footnote{A
$k$th-order Huffman code maps a block of $k$ source letters onto a variable-length codeword. 
}  and the Lempel-Ziv code, which respectively achieve the compression rates of $3.940$, $3.585$, $3.226$ and $6.028$ bits per English letter, the proposed UDOOCs can potentially result in comparable compression rate to the Huffman code under similar decoding complexity and yield
a smaller average codeword length than that of the Lempel-Ziv code, thereby confirming 
the practicability of the scheme of separating OOC codewords by UWs.

In the literature, there are a number of publications 
on enumeration of words in a set that forbids the appearance of a specific pattern  
\cite{MATH293,MATH195,MATH98,MATH99,MATH00}.
For example, Doroslova investigated the number of binary length-$n$ words,
in which a specific subword like $1010\ldots 10$ is not allowed \cite{MATH98}.
He then extended the result to non-binary alphabet and forbidden subwords of length $3$ \cite{MATH195,MATH00}, and forbidden subwords of length $4$ \cite{MATH99},
as well as the so-called ``good" forbidden subwords \cite{MATH293}.
The analyses in 
\cite{MATH293,MATH195,MATH98,MATH99,MATH00}
however depend 
on the specific structure of forbidden subwords considered, and no asymptotic 
examination is performed.
On the other hand, algorithmic approaches have been devoted
to a problem of similar (but not the same) kind, one of which
is called the Goulden-Jackson clustering method 
\cite{LNMATH,SP98,DIFF99,DIFF10,DIFF06}.

Instead of enumerating the number of words internally without a forbidden pattern,
some researchers investigate the inherent characteristic of such patterns.
In this literature, Rivals and Rahmann \cite{MATH03} provide an algorithm to account for the number of {\em overlaps}\footnote{
In \cite{MATH81} and \cite{MATH03}, the authors actually use a different name ``autocorrelation" for ``overlap"  originated from \cite{LNMATH}.
Specifically, they define the \emph{autocorrelation} $\vb=v_1\cdots v_L$ of a binary length-$L$ string $\ub=u_1\cdots u_L$ as a binary zero-one bit-stream of length $L$ such that $v_i=1$ if $i$ is a period of $\ub$, where
$i$ is said to be a period of $\ub$ when $u_j=u_{i+j}$ for every $1\leq j\leq L-i$.
Since the term \emph{autocorrelation} is extensively used in other literature ilke digital communications to illustrate similar but different conception, we adopt the name of ``overlap" in this paper.}
for a given set of patterns, for which 
the definition will be later given in this paper for completeness (cf. Definition \ref{definition-ao}).
Different from the algorithmic approach in \cite{MATH03},
Guibas and Odlyzko established 
upper and lower bounds
for the number of overlaps when the length of the concerned pattern goes to infinity
\cite{MATH81}.

\begin{figure}
\begin{center}
\includegraphics[width=0.5\textwidth]{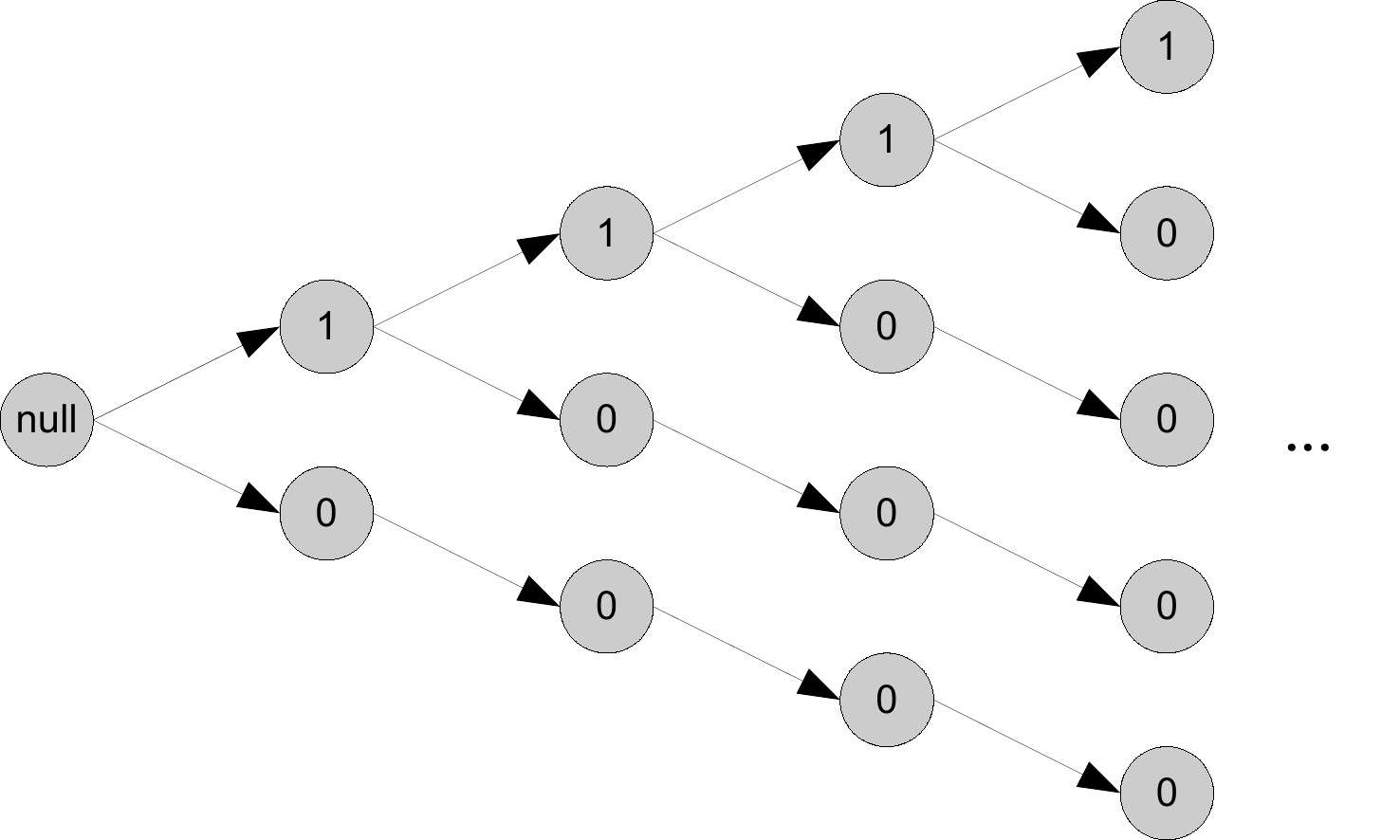}
\caption{UDOOC code tree for $\text{UW}=01$.}\label{01pctree}
\end{center}
\end{figure}

The rest of the paper is organized as follows. In Section \ref{sec:2},
construction of general UDOOCs is introduced. In Section \ref{sec:3},
combinatorial properties of UDOOCs, including the enumeration of the number of codewords,  are derived. In Section \ref{sec:4},
the encoding and decoding algorithms as well as bounds on average codeword length for general UDOOCs  are provided and discussed. 
In Section \ref{sec:5}, numerical results on the compression performance of UDOOCs are presented. 
Conclusion is drawn in Section \ref{sec:6}.

\section{Construction of UDOOCs} \label{sec:2}

In the previous section, we have seen that the code tree of a UDOOC
with $\text{UW}=00$ (or $\text{UW}=01$) is by far a useful tool for devising its properties.
Along this line, we will provide a systematic construction of code tree for general UDOOC
in this section. Specifically, a digraph \cite{DIGRAPH}  whose directional edges meet the type-I and type II constraints\footnote{For clarity of its explanation,
we introduce the so-called type-I and type-II constraints in Example \ref{ex1}. Listing
these constraints for a general UW however may be tedious and less comprehensive. 
As will be seen from this section, these constraints can actually  be 
absorbed into the construction of the digraph (See specifically Eq. \eqref{digraph});
hence, explicitly listing of constraints becomes of secondary necessity.
} from the UW will be first introduced. By the digraph,
the construction of a general UDOOC code tree as well as the determination 
of the growth rate of UDOOC codewords with respect to the codeword length will follow.

\subsection{Digraphs for UDOOCs}

Let $\kb=k_1\ldots k_L$ be the chosen UW of length $L$. 
Denote by $G_\kb=(V,E_\kb)$ the digraph for the UDOOC with $\text{UW}=\kb$, 
where $V=\F^{L-1}$ 
is the set of all binary length-$(L-1)$ tuples,
and $E_\kb$ is the set of directional edges given by
\begin{equation}\label{digraph}
E_\kb:= \left\{ ({\boldsymbol i}, {\boldsymbol j}) \in V^2: 
i_2^{L-1}=j_1^{L-2}\text{ and } i_1 {\boldsymbol j} \neq \kb 
\right\}.
\end{equation}
Here, we use the conventional shorthand $i_s^t=i_s i_{s+1} \ldots i_{t}$ 
to denote a binary string from index $s$ to index $t$,
and the elements in $V$ are interchangeably 
denoted by either ${\boldsymbol i}=i_1 \ldots i_{L-1}$ or $i_1^{L-1}$, depending on 
whichever is more convenient.

Define the $2^{L-1}$-by-$2^{L-1}$ adjacency matrix $\Am_\kb$ for the digraph $G_\kb$ by
putting its $(i+1,j+1)$th entry as
\begin{equation}\label{adjMatrix}
 \left( \Am_\kb \right)_{i+1,j+1}= 
	 			\begin{cases}
	                       1, & \mbox{if } ({\boldsymbol i}, {\boldsymbol j}) \in E_\kb, \\
						   0, & \textnormal{otherwise,} \\
				\end{cases}
\end{equation}
where we abuse the notation by using $i$ (resp.~$j$) to be
the integer corresponding to binary representation of ${\boldsymbol i}=i_1 \ldots i_{L-1}$ (resp.~${\boldsymbol j}=j_1 \ldots j_{L-1}$) with the leftmost bit being the most significant bit.
As an example, for $\kb=010$, we have $V=\F^2=\{00,01,10,11\}$, 
\begin{multline*}
E_{010}=\{(00,00),(00,01),(01,11),\\
(10,00),(10,01),(11,10),(11,11)\},
\end{multline*}
 $G_{010}=(V,E_{010})$ in Fig.~\ref{010digraph}, and 
\[\Am_{010}=\begin{bmatrix} 
	 					1&1&0&0\\
	 					0&0&0&1\\
	 					1&1&0&0\\
	 					0&0&1&1\\
	                  \end{bmatrix} .
\]
We remark that the adjacency matrix $\Am_\kb$ will be used 
for enumerating the number of UDOOC codewords in next section.
                                    
\subsection{Code Trees for UDOOCs}

	\begin{figure}
		\begin{center}
		\includegraphics[width=0.8\columnwidth]{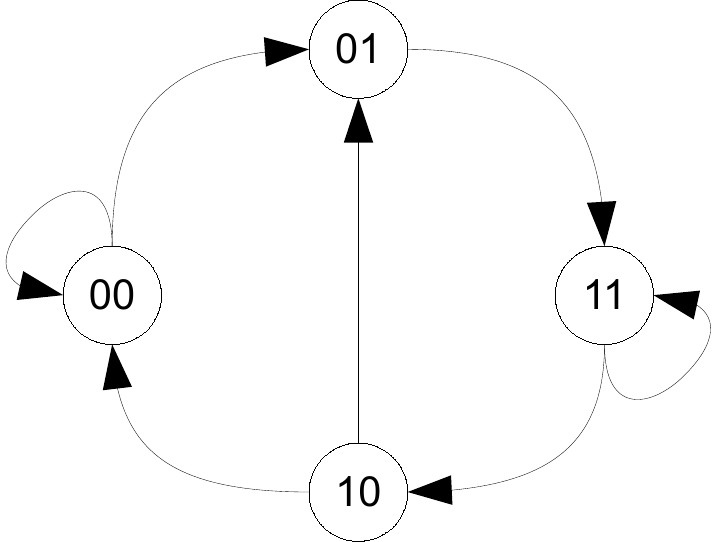}
		\caption{Digraph $G_{010}$ for UW $\kb=010$.}\label{010digraph}
		\end{center}
	\end{figure}
	
Equipped with digraph $G_\kb$,  
constructing the code tree for the UDOOC with $\text{UW}=\kb$
becomes straightforward. Recall that a UDOOC codeword of length $n$ is a binary $n$-tuple ${\boldsymbol b}=b_1 \ldots b_n$, satisfying that $\kb$ is not an internal subword of the concatenated bit-stream $\kb {\boldsymbol b} \kb$. As such, the traversal of the digraph for constructing a UDOOC code tree
should start from the vertex $k_2^{L}\in V$, which corresponds to the initial ``null"-node in the code tree. 
Next, a ``$0$"-node at level $1$ is generated if both $(k_2^{L},j_1^{L-1})\in E_\kb$ and $j_{L-1}=0$ are satisfied. 
By the same rule, the ``null"-node is followed by a ``$1$"-node at level $1$ if $(k_2^{L},j_1^{L-1})\in E_\kb$ and $j_{L-1}=1$.
We then move the current vertex to $j_1^{L-1}$ and draw a branch 
from ``$j_{L-1}$"-node at level $1$ 
to a followup
``$0$"-node (resp.~``$1$"-node) at level $2$ in the code tree if $(j_1^{L-1},\ell_1^{L-1})\in E_\kb$ and $\ell_{L-1}=0$ (resp.~$\ell_{L-1}=1$). We move the current vertex again to $\ell_1^{L-1}$ 
and re-do the above procedure to generate the nodes in the next level.
Repeating this process will 
complete the exploration of the nodes in the entire code tree.

Determination of the gray-shaded nodes that end a codeword 
can be done as follows. Since 
$\kb$ cannot be an internal subword of $\kb {\boldsymbol b} \kb$,
a node should be gray-shaded if 
it is immediately followed by a sequence of offspring nodes with 
their binary marks equal to $k_1\ldots k_{L-1}$.
The construction of the UDOOC code tree is accordingly finished.

As an example, we continue from the exemplified UW $\kb=010$ with digraph $G_\kb$ in Fig.~\ref{010digraph} and explore its respective UDOOC code tree in Fig.~\ref{010pctree}
by following the previously mentioned procedure.
By starting from the vertex $k_{2}^{3}=10$ that corresponds to the ``null"-node,
two succeeding nodes are generated since
both $(10,00)$ and $(10,01)$ are in $E_{010}$ (cf.~Fig.~\ref{010pctree}).
Now from vertex $00$ that corresponds to the ``$0$"-node at level $1$, we can reach either vertex $00$ or vertex $01$ in one transition; hence, both ``$0$"-node and ``$1$"-node are 
the succeeding nodes to the ``$0$"-node at level $1$.
However, since vertex $01$ can only walk to vertex $11$ in one transition,
the ``$1$"-node at level $1$ has only one succeeding node with mark ``$1$."
Continuing this process then exhausts all the nodes in the code tree in Fig.~\ref{010pctree}.
Next, all nodes that are followed by $k_1k_2=01$ in sequence in the code tree are gray-shaded. The construction of the code tree for the UDOOC with UW $\kb=010$ is then completed.

We end this section by giving the type-I and type-II constraints for the exemplified code tree as follows.
\begin{itemize}
\item Type-I constraints:
\begin{itemize}
\item[] (C1) ``0'' can be followed by either ``0'' or ``1''.
\item[] (C2)  ``1'' can be followed by ``0'' only when the node prior to this ``1''-node is not a 
``0''-node. 
\end{itemize}
\item  Type-II constraints:
\begin{itemize}
\item[] (C3) The first two bits of a UDOOC codeword cannot be ``10.''
\item[] (C4) The last two bits of a UDOOC codeword cannot be ``01.''
\end{itemize}
\end{itemize}
Note that with these constraints (in particular (C4)), one can also perform the node-shading step by first gray-shading all the nodes in the code tree, and then unshade those that end with ``01" 
(in addition to the ``1"-node at level 1 for this specific UW).
Nevertheless, it may be tedious to perform the node-unshading  
for a general UW. For example, 
when UW $\kb=01001$, all nodes that end a codeword $b_1^n$, satisfying either
$b_{n-3}b_{n-2}b_{n-1}b_nk_1=\kb$ or $b_{n-2}b_{n-1}b_nk_1k_2=\kb$, should be unshaded.
This confirms the superiority of constructing the UDOOC code tree in terms of the digraph
over analyzing the explicit listing of constraints from the adopted UW that are perhaps convenient only for some special UWs.

	\begin{figure}
		\begin{center}
		\includegraphics[width=0.8\columnwidth]{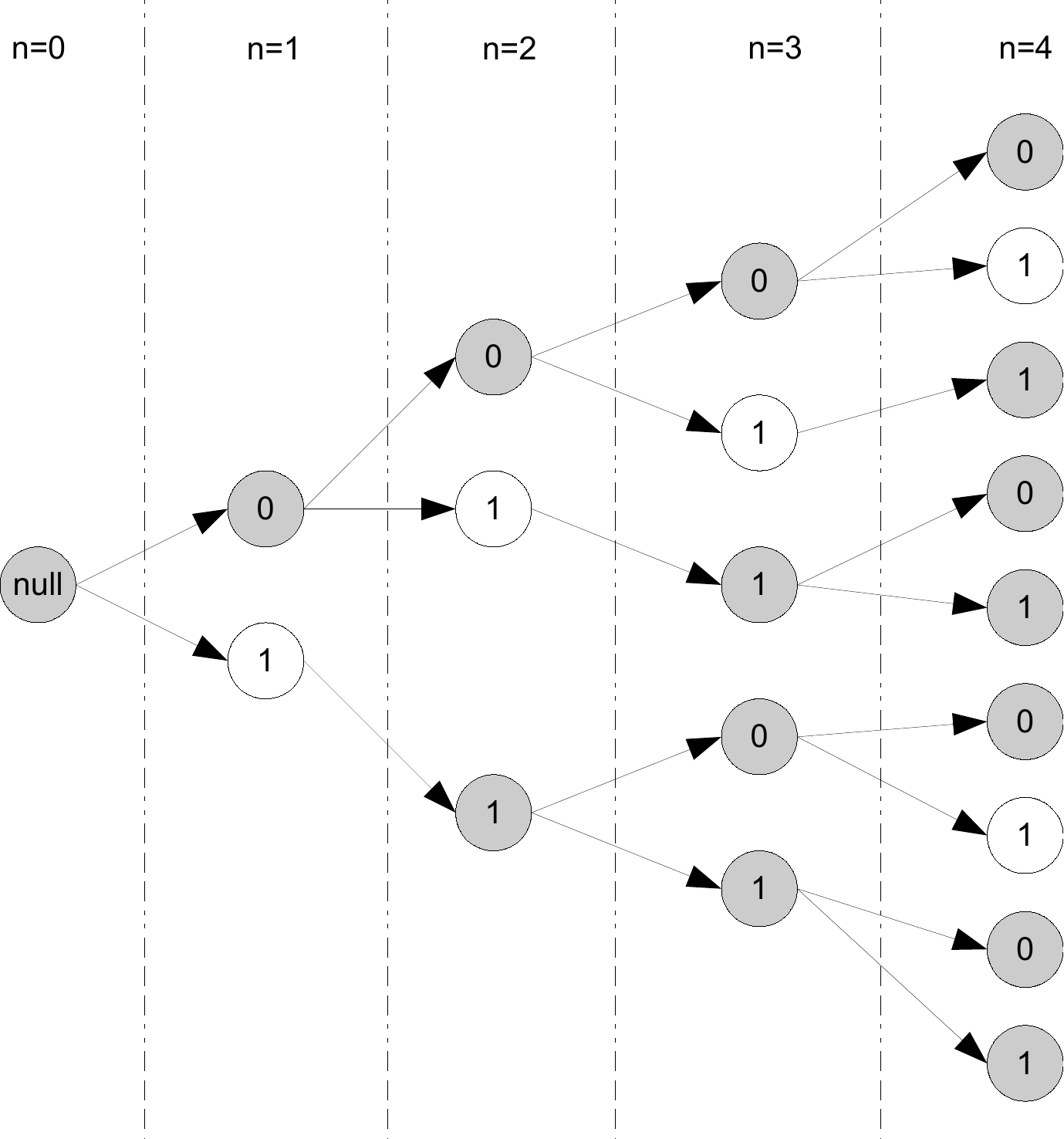}
		\caption{Code tree for the UDOOC with UW $\kb=010$.}\label{010pctree}
		\end{center}
	\end{figure}

\section{Combinatorial Properties of UDOOCs} \label{sec:3}

\subsection{The Determination of  $\abs{{\cal C}_\kb(n)}$} \label{sec:3a}

In this subsection, we will see that the conception of digraph $G_\kb$, in particular
its respective adjacency matrix $\Am_\kb$, can lead to a formula for the number of length-$n$ codewords, i.e., $c_{\kb,n}=|{\cal C}_\kb(n)|$.

In accordance with the fact that the traversal of the digraph for constructing a UDOOC code tree should start from vertex $k_2^{L}\in V$, we define a length-$2^{L-1}$ initial vector
as $(\xb_\kb)_{j+1}=1$ if integer $j$ has the binary representation $k_2^{L}$, and $(\xb_\kb)_{j+1}=0$, otherwise, for $0\leq j< 2^{L-1}$.
It then follows that the $(\ell+1)$th entry of row vector $\xb_\kb^\top \Am_\kb^n$ gives the number of length-$n$ walks that end at vertex ${\boldsymbol \ell}$ on digraph $G_\kb$, 
where ``$^\top$" denotes the vector/matrix transpose operation, and ${\boldsymbol \ell}=\ell_1\ldots\ell_{L-1}$ is the binary representation of integer index $\ell$.

However, not every length-$n$ walk produces a codeword. Notably, some nodes on the code tree will be gray-shaded and some will not.
Recall that $\kb$ cannot be an internal subword of $\kb{\boldsymbol b}\kb$ 
if ${\boldsymbol b}=b_1 b_2 \ldots b_n$ is a codeword.
This implies that ${\boldsymbol b}$ is a length-$n$ codeword if, and only if, the vertex sequence
$k_2^L$, $k_3^Lb_1$, $k_4^Lb_1^2$, $\ldots$, $b_nk_1^{L-2}$, $k_1^{L-1}$
is a valid walk of length $n+L-1$ on digraph $G_\kb$. 
As a result,  the number of length-$n$ codewords equals the number of length-$(n+L-1)$ walks  from vertex $k_2^{L}$ to vertex $k_1^{L-1}$  on digraph $G_\kb$. 
Following the above discussion, 
we define the length-$2^{L-1}$ \emph{ending vector} $\yb_\kb$
as $(\yb_\kb)_{j+1}=1$ if integer $j$ has the binary representation $k_1^{L-1}$, and $(\yb_\kb)_{j+1}=0$, otherwise, for $0\leq j<2^{L-1}$. Then, 
the number of length-$n$ codewords is given by
\begin{equation}
	c_{\kb,n} : = \abs{{\cal C}_\kb (n)} \ = \ \xb_\kb^\top\Am_\kb^{n+L-1}\yb_\kb. \label{eq:ckbn}
\end{equation} 

\subsection{Equivalence among UWs}

Two UWs that result in the same minimum average codeword length
for every source statistics
should be considered \emph{equivalent}. 
This leads to the following definition.

\medskip

\begin{defn} \label{defn:eqrln}
Two UWs $\kb$ and $\kb'$ are said to be \emph{equivalent}, denoted by $\kb \equiv \kb'$, if  the numbers of their length-$n$ codewords in the corresponding UDOOCs are the same for all $n$, i.e.,
\begin{equation}
c_{\kb,n} \ = \ c_{\kb',n}\ \textnormal{ for all $n \geq 0$}. 
\end{equation}
\end{defn}

\medskip

By this definition, UDOOCs associated with equivalent UWs have the same number of codewords in every code tree level; hence they achieve the same minimum average codeword length in the lossless compression of a sequence of source letters. This equivalence relation allows us to focus only on 
one UW in every equivalent class.
It is however hard to exhaust and identify all equivalent classes of UWs of arbitrary length.
Instead, we will introduce a less restrictive notion of \emph{asymptotic equivalence}
when the asymptotic compression rate of UDOOCs is concerned, and derive the number of all asymptotically equivalent classes of UWs in Section \ref{sec:asymp}.

Some properties about the (exact) equivalence of UWs are given below.

\medskip

\begin{prop}[Equivalence in order reversing] \label{prop:rev}
UW $\kb'=k_L \ldots k_1$ is equivalent to UW $\kb=k_1 k_2 \ldots k_L$.
\end{prop}
\begin{proof}
It follows simply from that $\bb=b_1b_2\ldots b_n \in {\cal C}_\kb$ if, and only if, 
$\bb'=b_nb_{n-1}\ldots b_1\in {\cal C}_{\kb'}$. 
\end{proof}

\medskip

\begin{prop}[Equivalence in binary complement] \label{prop:comp}
If $\bar\kb$ is the bit-wise binary complement of $\kb$,
then $\bar{\kb}$ and $\kb$ are equivalent. 
\end{prop}
\begin{proof}
It is a consequence of the fact
that the concatenated bit-stream $\kb \bb \kb$ contains $\kb$ as an internal subword if, and only if, the binary complement $\overline{\kb \bb \kb}$ of $\kb \bb \kb$ contains $\bar{\kb}$ as an internal subword. 
\end{proof}

\medskip

From Propositions \ref{prop:rev} and \ref{prop:comp}, it can be verified that 
there are at most four equivalent classes for UWs of length $L=4$.
Representative UWs for these four equivalent classes
are $0000$, $0001$, $0100$ and $0101$, respectively.

\subsection{Growth Rates of UDOOCs}

In this subsection, we investigate the asymptotic growth rate of UDOOCs, of which the definition is given below.

\medskip

\begin{defn}
Given UW $\kb$, the asymptotic growth rate of the resulting UDOOC is defined as 
\begin{equation}
g_\kb \ := \ \lim_{n\to \infty} \frac{c_{\kb,n+1}}{c_{\kb,n}}. \label{eq:gkb}
\end{equation}
\end{defn}

\medskip

By its definition, the asymptotic growth rate of a UDOOC indicates 
how fast the number of codewords grows as $n$ increases. 

It is obvious that $g_\kb \leq 2$ for all UWs because the upper bound of $2$ is the growth rate for unconstrained binary sequences of length $n$. In addition, the limit in \eqref{eq:gkb} must exist since it can be inferred from enumerative combinatorics \cite{EC}, and also from algebraic graph theory \cite{GRAPH}, that $g_\kb$ is the largest eigenvalue of adjacency matrix $\Am_\kb$. In the next proposition, we show that the largest eigenvalue of adjacency matrix $\Am_\kb$ is unique for all UWs but $\kb=01$.

\medskip

\begin{prop}[Uniqueness of the largest eigenvalue of $\Am_\kb$] \label{prop:kUnique}
	For any UW $\kb$ of length $L\geq 2$ except $\kb=01$, the largest eigenvalue of adjacency matrix $\Am_\kb$ is unique and is real. 
\end{prop}
\begin{proof}  By Perron-Frobenius theorem \cite{GR}\cite{HJ}, the largest eigenvalue of adjacency matrix $\Am_\kb$ is unique and real with algebraic multiplicity equal to $1$ if $G_\kb$ is a strongly connected diagrph. Thus, we only need to show that $G_\kb$ is a strongly connected digraph except for $\kb=01$.
    
   We then argue that $G_\kb$ is a strongly connected diagrph when $L\geq 3$ as follows.
   According to the definition of $E_\kb$ in \eqref{digraph},    
   the only situation that a vertex may not be strongly connected to other vertex is when
   ${\boldsymbol j}=k_2k_3 \cdots k_L$. 
   This however cannot happen when $L\geq 3$ because vertex $\overline{k}_1k_2\cdots k_{L-1}$ will connect strongly to $k_2k_3 \cdots k_L$. 
The proof is completed after verifying the two cases for $L=2$, i.e., $G_{00}$ is strongly connected but $G_{01}$ is not.
\end{proof}

\medskip

The digraph for $\kb=01$ is plotted in Fig.~\ref{01digraph}. 
It clearly indicates that there is no directed path from vertex 0 to vertex 1.
In fact, the algebraic multiplicity of the largest eigenvalue $1$ of $\Am_{01}$ is two.

	\begin{figure}
		\begin{center}
		\includegraphics[width=0.8\columnwidth]{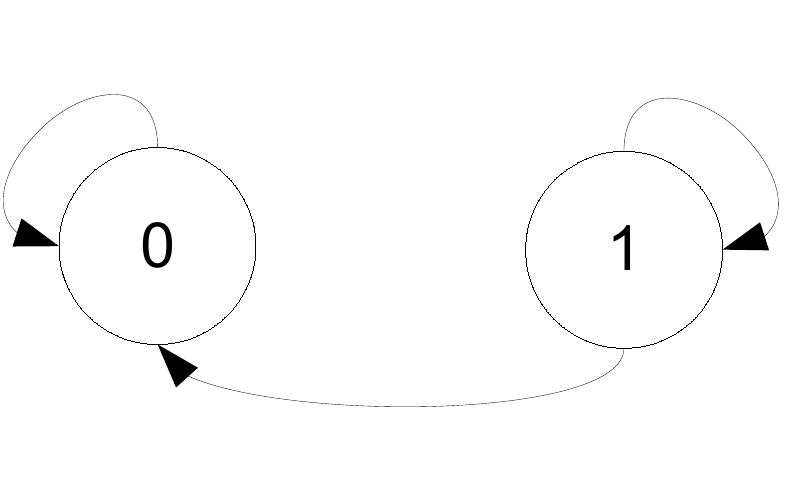}
		\caption{Digraph $G_\kb$ for UW $\kb=01$.}\label{01digraph}
		\end{center}
	\end{figure}

By the standard technique of using an indeterminate $z$ in enumerative combinatorics, we can enumerate the numbers $c_{\kb,n}$ as 
\begin{eqnarray}
\sum_{n=0}^\infty c_{\kb,n} z^n &=& \sum_{n=0}^\infty  \underline{x}_\kb^\top \Am_\kb^{n+L-1} \yb_\kb z^n \nonumber\\
&=& \underline{x}_\kb^\top  \left( \sum_{n=0}^\infty \Am_\kb^n z^n \right) \Am_\kb^{L-1} \yb_\kb \nonumber\\
&=&  \underline{x}_\kb^\top  \left( \mI - \Am_\kb z \right)^{-1} \Am_\kb^{L-1} \yb_\kb \nonumber\\
&=&  \frac{\underline{x}_\kb^\top \adj\left( \mI - \Am_\kb z \right) \Am_\kb^{L-1} \yb_\kb}{\det\left( \mI - \Am_\kb z \right) }, \label{eq:enumerate}
\end{eqnarray}
where the first equality follows from \eqref{eq:ckbn} and $\mI$ denotes the identity matrix of proper size. Equation~\eqref{eq:enumerate} then implies that
 $\det\left( \mI - \Am_\kb z \right)$ can give a linear recursion of $c_{\kb,n}$ in the form of a linear constant coefficient difference equation (LCCDE).

Now let $\lambda_1, \ldots, \lambda_m$ be distinct nonzero eigenvalues of adjacency matrix $\Am_\kb$ with algebraic multiplicities $e_1, \ldots, e_m$, respectively,
where we assume  with no loss of generality that $\abs{\lambda_1} \geq \cdots \geq \abs{\lambda_m}$.
In terms of the standard technique of partial fraction for rational functions, we can rewrite \eqref{eq:enumerate} as
\begin{equation}
\frac{\xb_\kb^\top  \adj\left( \mI - \Am_\kb z \right) \Am_\kb^{L-1} \underline{y}_\kb}{\det\left( \mI - \Am_\kb z \right) } \ = \ \sum_{i=1}^m \frac{p_i(z)}{(1-\lambda_i z)^{e_i}} \label{eq:z}
\end{equation}
for some polynomials $p_i(z)$. 
The next step is expectantly to rewrite the righ-hand-side (RHS) of \eqref{eq:z} as a power series of indeterminate $z$
in order to recover the actual values of $c_{\kb,n}$ for all $n$.
As an example, this can be done by
\[
\frac{1}{\left( 1 - \lambda_i z \right)^{e_i}} \ = \ \sum_{n=0}^\infty \binom{n+e_i-1}{n} \lambda_i^n z^n,
\]
which holds for all $\abs{z} < \min_{1\leq i\leq m}\frac{1}{\abs{\lambda_i}}$.

Although the asymptotic growth rate $g_\kb$ equals exactly the largest eigenvalue of 
adjacency matrix $\Am_\kb$, it is in general difficult to find a closed-form expression for this value without a proper reshaping of adjacency matrix $\Am_\kb$. 
Another approach is to consider the following set for $n\geq L$,
\begin{equation}\label{eq:skb}
{\cal S}_\kb(n)  :=  \left\{ \bb \in \F^n  : \textnormal{ $\kb$ is not a subword of $\bb$} \right\}, 
\end{equation} 
which, in a way, defines the set of distinct length-$n$ walks on digraph $G_\kb$. 
Denoting $s_{\kb,n} := \abs{{\cal S}_\kb(n)}$ and by an argument similar to \eqref{eq:enumerate}, one can easily show that
\beq
\sum_{n =0}^\infty s_{\kb,n} z^n \ = \ \sum_{n=0}^{L-1} 2^n z^n + z^L \frac{\underline{{\bf 1}}^\top \adj\left( \mI - \Am_\kb z \right) \underline{{\bf 1}} }{\det\left( \mI - \Am_\kb z \right) }, \label{eq:skbnz1}
\eeq
where $\underline{{\bf 1}}$ is the all-one column vector of appropriate length. 
Equation \eqref{eq:skbnz1} then implies that the enumeration of $s_{\kb,n}$ also depends upon the polynomial $\det\left( \mI - \Am_\kb z \right)$ as $c_{\kb,n}$ does. Based on this observation,
we can infer and prove that $s_{\kb,n}$ has the same asymptotic growth rate as $c_{\kb,n}$. We summarize this important inference in the proposition below,
while the proof will be relegated to the next subsection.

\medskip

\begin{prop} \label{prop:connect}
	For any UW $\kb$, sequences $\{c_{\kb,n}\}_{n=0}^\infty$ and $\{ s_{\kb,n} \}_{n=0}^\infty$ have the same asymptotic growth rate, i.e.,  $$g_\kb=\mathfrak{g}_\kb,$$
	where
	$$\mathfrak{g}_\kb\ :=\ \lim_{n\to \infty} \frac{s_{\kb,n+1}}{s_{\kb,n}}.$$

Notably, in order to distinguish the asymptotic growth rate of $s_{\kb,n}$ from that of $c_{\kb,n}$, a different font $\mathfrak{g}_\kb$ is used to denote the asymptotic growth rate of $s_{\kb,n}$.
\end{prop}

\subsection{Enumeration of $s_{\kb,n}$}

Enumerating $s_{\kb,n}$ turns out to be easier than enumerating $c_{\kb,n}$ due to that there is lesser number of constraints on the sequences in ${\cal S}_{\kb}(n)$.
It can be done by an approach similar to the {\em Goulden-Jackson clustering method} \cite{DIFF99}. 
Before delivering the main theorems, we define the {\em overlap function} and {\em overlap vector} of a binary stream $\kb$ as follows.

\medskip

\begin{defn}\label{definition-ao}
For a given $\kb$ of length $L$, its overlap function is defined as 
\begin{equation}
r_\kb(i) := 	   \begin{cases} 
	                       1 & \mbox{if } k_{i+1}^L = k_1^{L-i} \text{ and } 0 \leq i \leq L-1\\
						   0 & \textnormal{otherwise. } \\
				\end{cases} \label{eq:rkbi}
\end{equation}
Furthermore, we define its length-$L$ overlap vector as
$(\underline r_\kb)_j=r_\kb(j-1)$ for $j=1\ldots L$.
\end{defn}

\medskip

\begin{thm} \label{thm:1}
For a length-$L$ UW $\kb$ with overlap function $r_\kb(i)$, 
\beq
\sum_{n \geq 0} s_{\kb,n} z^n \ = \ \dfrac{1+\sum_{i=1}^{L-1} r_\kb(i)z^i}{(1-2z)\left(1+\sum_{i=1}^{L-1} r_\kb(i)z^i\right)+z^L}. \label{eq:skbz}
\eeq
Moreover, let $h_{\kb}(z)$ denote the denominator of \eqref{eq:skbz}, i.e.,
\beq
h_\kb(z) \ = \ (1-2z)\left(1+\sum_{i=1}^{L-1} r_\kb(i)z^i\right)+z^L.  \label{eq:hkbz}
\eeq
Then 
\begin{equation}\label{firstclaim}
h_\kb(z)
=\det(\mI-\Am_\kb z),
\end{equation}
where $\Am_\kb$ is the adjacency matrix associated with digraph $G_\kb$. 
\end{thm}
\begin{proof} The result \eqref{eq:skbz} follows from the Goulden-Jackson clustering method \cite{DIFF99}. For completeness, a simplified proof to this claim is provided in 
Appendix \ref{app:a}. To establish the second claim, i.e., \eqref{firstclaim}, we combine \eqref{eq:skbnz1} and \eqref{eq:skbz} to give
\[
\dfrac{1+\sum_{i=1}^{L-1} r_\kb(i)z^i}{(1-2z)(1+\sum_{i=1}^{L-1} r_\kb(i)z^i)+z^L}\  \ = \frac{f(z)}{\det(\mI-\Am_\kb z)}
\]
for some polynomial $f(z)$. Notice that the left-hand-side (LHS) is an irreducible rational function in $z$. Furthermore, Proposition \ref{prop:deg} in Appendix \ref{app:b} shows $\deg\det(\mI-\Am_\kb z)=L$. These  then imply that
\[
\det(\mI-\Am_\kb z) \ = \ h_\kb(z)
\]
and
\[f(z) \ = \ 1+\sum_{i=1}^{L-1} r_\kb(i)z^i.
\]
 \eqref{firstclaim} is thus established. 
\end{proof}

\medskip

The next example illustrates the usage of the above theorem to the target result
of Proposition \ref{prop:connect}.

\medskip

\begin{ex}
Consider the case of $\kb=000$. Then from \eqref{eq:rkbi}, the corresponding overlap function is 
\[
r_{000} (i) \ =\ 
\begin{cases}
1, &  i=0,1,2,\\
0, & \textnormal{otherwise.} \\
\end{cases}
\]
Substituting the above into \eqref{eq:skbz}, we obtain
\[
\sum_{n \geq 0} s_{000,n} z^n \ = \ \frac{1+z+z^2}{1-z-z^2-z^3}.
\]  
By regarding the above as an LCCDE, we conclude that the sequence $s_{000,n}$ satisfies the following recursion for all $n \geq 0$:
\[
s_{000,n} = s_{000,n-1} + s_{000,n-2} + s_{000,n-3}+ \delta_n + \delta_{n-1} + \delta_{n-2},
\]
where 
$$\delta_n=\begin{cases}1,&n=0\\
0,&\text{otherwise}\end{cases}$$
 is the Kronecker delta function. 
\qed
\end{ex}

\medskip

Equipped with Theorem \ref{thm:1}, we are now ready to prove Proposition \ref{prop:connect}. 

\medskip

\begin{proof}[Proof of Proposition \ref{prop:connect}]
From the proof of Theorem \ref{thm:1}, we have seen that the enumeration of $s_{\kb,n}$ is given by the following irreducible rational function
\[
\sum_{n =0}^\infty s_{\kb,n} z^n \ = \ \frac{1+\sum_{i=1}^{L-1} r_\kb(i)z^i}{h_{\kb}(z)}
\]
where $h_{\kb}(z)$ is the denominator of \eqref{eq:skbz} and is given by \eqref{eq:hkbz}. Hence, it follows from the standard partial fraction technique and Proposition \ref{prop:kUnique} that 
\[
\mathfrak{g}_\kb = \max \{ \abs{u}^{-1}  :  h_\kb(u) =0, u \in \C\},
\]
where $\C$ is the set of complex numbers.
Next, noticing that the function $h_{\kb}(z)$, i.e., $\det(\mI-\Am_\kb z)$, also appears as the denominator of the enumeration function for $c_{\kb,n}$ (cf. \eqref{eq:enumerate}), we get
\[
g_\kb \leq \max \{ \abs{u}^{-1} : h_\kb(u) =0, u \in \C\},
\]
since the rational function in  \eqref{eq:enumerate} could be reducible. This shows $ g_\kb \leq \mathfrak{g}_\kb$.

To prove $g_\kb\geq \mathfrak{g}_\kb$ (which then implies $g_\kb=\mathfrak{g}_\kb$),
it suffices to show that $c_{\kb,n+2} \geq s_{\kb,n}$ for $n\geq L$.
This can be done by substantiating that
for any $\bb \in {\cal S}_\kb (n)$,
 there exist a prefix bit $p$ and a suffix bit $q$, where $p,q \in \mathbb{F}$, such that $p \bb q \in {\cal C}_\kb (n+2)$. 
 
Using the prove-by-contradiction argument, we first assume that
$\kb$ is an internal subword of both $\kb p\bb$ and $\kb \bar{p} \bb$, where $\bar{p} = 1 - p$. 
This assumption, together with $\bb \in {\cal S}_\kb (n)$, implies the existence 
of indices $1<i<L+2$ and $1<j<L+2$ such that 
\begin{equation}\label{p4-1}
\underbrace{k_i \cdots k_L p b_1 \cdots b_{i-2}} _{=\, \ab} = \underbrace{k_j \cdots k_L \bar{p} b_1 \cdots b_{j-2}}_{=\, \tilde\ab} = \kb,
\end{equation}
where we abuse the notations to let 
\begin{equation}\label{p4-2}
\ab=\begin{cases}
k_i\cdots k_L p,&\text{if }i=2\\
k_i \cdots k_L p b_1 \cdots b_{i-2},&\text{if }2<i<L+1\\
pb_1\cdots b_{i-2},&\text{if }i=L+1
\end{cases}
\end{equation}
and similar notational abuse is applied to $\tilde\ab$ and $j$.
Assume without loss of generality that $i<j$.
Then, the sums of the last $(j-1)$ bits of $\ab$ and $\tilde\ab$ must equal, i.e.,
\[
k_{L-(j-i-1)} + \cdots + k_L+p+b_1+ \cdots + b_{i-2} = \bar{p} + b_{1} + \cdots + b_{j-2}.
\]
Canceling out common terms at both sides gives 
\begin{equation}\label{a1}
k_{L-(j-i-1)} + \cdots + k_L+p = \bar{p} + b_{i-1} + \cdots + b_{j-2}.
\end{equation}
Note again that $\tilde\ab =\kb$; hence, substituting $b_{(j-2)-\ell}$ by $k_{L-\ell}$ for $\ell=0,1, \cdots , j-i-1$ in \eqref{a1} gives $p=\bar{p}$, which contradicts
the assumption that $\bar p=1-p$.

For the suffix bit $q$, we again assume to the contrary that there exist indices $i$ and $j$, satisfying $n+1-L<i<j<n+2$, such that 
\begin{equation}\label{p4-3}
\underbrace{b_i \cdots b_n q k_1 \cdots k_{L-n+i-2}} _{=\,\db} = \underbrace{b_j \cdots b_n \bar{q} k_1 \cdots k_{L-n+j-2}} _{=\,\tilde\db} = \kb
\end{equation}
After canceling out common terms in the respective sums of the first $(n+2-i)$ bits of $\db$ and $\tilde\db$, we obtain 
\[
b_i + \cdots + b_{j-1} + q = \bar{q} + k_1+ \cdots + k_{j-i}.
\]
Since $\db=\kb$, the above implies $q=\bar{q}$, which again leads to a contradiction.
\end{proof}

\medskip

One application of the result $h_\kb(z)=\det(\mI-\Am_\kb z)$ in Theorem \ref{thm:1} is to obtain a recursion formula for $c_{\kb,n}$, i.e., an  LCCDE for $c_{\kb,n}$. This is provided in the next corollary.

\medskip

\begin{cor}\label{cor1}
For a length-$L$ UW $\kb$ with overlap function $r_\kb(i)$, let $c_{\kb,n}$ be the number of length-$n$ codewords in the UDOOC ${\cal C}_\kb$ defined as before. Then, for $n\geq L$,
\beq
c_{\kb,n}  =  \left[ \sum_{i=1}^{L-1} r_\kb(i) \left( 2 c_{\kb,n-i-1} - c_{\kb,n-i} \right) \right]+ 2 c_{\kb,n-1} - c_{\kb,n-L}. \label{eq:recdkbn}
\eeq
\end{cor}
\begin{proof}
To prove \eqref{eq:recdkbn}, we first note that the characteristic polynomial for $\Am_\kb$ is given by 
\begin{eqnarray*}
	\chi_{\Am_\kb}(z) & = & \det(z \mI-\Am_\kb) \\
	                  & = & z^{2^{L-1}} h_\kb(1/z) \\
	                  & = & z^{2^{L-1}-L}\left(z^L h_\kb(1/z)\right) 
\end{eqnarray*}
where $z^L h_\kb(1/z)$ is a polynomial with degree $$\deg h_\kb(z) = \deg\det(\mI-\Am_\kb z)=L.$$
Denote
\begin{equation} 
m \ = \ \min \{ p > 0 \ : \ \mbox{Nullity}(\Am_\kb^p) = 2^{L-1}-L \}, \label{eq:ch1}
\end{equation}
where Nullity$()$ indicates the dimension of the null space of the square matrix inside parentheses.
By Cayley-Hamilton Theorem \cite{HK},  
the following polynomial 
\begin{eqnarray}
\lefteqn{\mu_\kb(z) :=  z^m (z^L h_\kb(1/z))} \no \\ 
& = & z^m \left(z^L-2z^{L-1}+ \sum_{i=1}^{L-1} r_{\kb}(i)z^{L-i-1}(z-2)+1\right)\no \\ \label{eq:ch2}
\end{eqnarray}
is an annihilating polynomial for $\Am_\kb$. We shall remark that $\mu_\kb(z)$ needs not to be the minimal polynomial for $\Am_\kb$. Plugging \eqref{eq:ch2} into \eqref{eq:ckbn} yields that for $n-1 \geq \max \{ m ,L-1\}$, we have 
\bea
\lefteqn{c_{\kb,n}}\no\\
\quad & =& \xb_\kb^\top\Am_\kb^{n+L-1}\yb_\kb \no \\
&=& \xb_\kb^\top\Am_\kb^{n-1-m}\Am_\kb^{m+L}\yb_\kb \label{se} \\
&=& \xb_\kb^\top\Am_\kb^{n-1} \left[ 2\Am_\kb^{L-1} + \sum_{i=1}^{L-1} r_\kb(i)(2\Am_\kb^{L-i-1}-\Am_\kb^{L-i}) - \mI \right] \yb_\kb \no \\ \no
&=& \left[ \sum_{i=1}^{L-1} r_\kb(i) \left( 2 c_{\kb,n-i-1} - c_{\kb,n-i} \right) \right]+ 2 c_{\kb,n-1} - c_{\kb,n-L}, \\ \label{eq:ch3} 
\eea
where the condition of $n-1 \geq \max \{ m ,L-1\}$ 
 follows from i) $n-1-m \geq 0$ such that \eqref{se} holds, and ii) $n-1 \geq L-1$ such that the last term of the RHS of \eqref{eq:ch3} represents $c_{\kb,n-L}$. 
Finally, since $\rank(\Am_\kb^{p}) \leq L$ for $p = L-1$ (see Proposition \ref{prop:deg}),
we have $m\leq L-1$, which immediately gives $\max \{ m ,L-1\}=L-1$. The proof is thus completed. 
\end{proof}

\medskip

So far, we learn that $c_{\kb,n}$ and $s_{\kb,n}$ have the same asymptotic growth rate, and both of their enumerations depend on $\det(\mI-\Am_\kb z)$.
Below we will use $s_{\kb,n}$ to determine the asymptotic growth rates corresponding to two specific UWs, $\ab=0\ldots00$ and $\bb=0\ldots01$. 
We then proceed to show that $\ab$ has the largest growth rate among all UWs of the same length, while the smallest growth rate is resulted when $\text{UW}=\bb$. 

\medskip

\begin{thm} \label{thm:2}
Among all UWs of the same length, the all-zero UW has the largest growth rate, while  UW $0 \ldots 01$ achieves the smallest. 
\end{thm}
\begin{proof}
For notational convenience, we set $\ab=0 \ldots 0$ and $\bb=0 \ldots 01$. 
For $\text{UW}=\ab$, it can be verified from \eqref{eq:rkbi} and \eqref{eq:skbz} that 
\beq\label{ha}
h_\ab(z) = 1-\sum_{i=1}^L z^i,
\eeq
and hence the sequence of $\{s_{\ab,n}\}_{n=1}^\infty$ satisfies the following recursion:
\[
s_{\ab,n} \ = \ \sum_{i=1}^{L} s_{\ab,n-i}\quad\text{ for }n\geq L.
\]
Similarly, we have $h_\bb(z) =  1-2z+z^L$, and therefore,
\[
s_{\bb,n} = 2 s_{\bb,n-1} - s_{\bb,n-L}\quad\text{ for }n\geq L.
\]
For general UW $\kb$ of length $L$, \eqref{eq:skbz} gives the following recursion for $n \geq L$
\beq
s_{\kb,n} \ = \ \sum_{i=1}^{L-1} \left( 2 s_{\kb,n-i-1} - s_{\kb,n-i} \right) r_\kb(i) + 2 s_{\kb,n-1} - s_{\kb,n-L}. \label{eq:LCCDE3}
\eeq
Note that $r_\kb(i) \in \{0,1\}$ by definition, and $2 s_{\kb,m-1} \geq s_{\kb,m}$ for all $m$.
From \eqref{eq:LCCDE3}, the following bounds hold for any UW $\kb$ with $n \geq L$:
\beq
2 s_{\kb,n-1} - s_{\kb,n-L} \ \leq \ s_{\kb,n} \ \leq \ \sum_{i=1}^L s_{\kb,n-i}, \label{eq:rbounds}
\eeq
where the lower and upper bounds are respectively obtained by replacing all $r_\kb(i)$ in \eqref{eq:LCCDE3} by $0$ and $1$.
In particular, $s_{\kb,n}$ equals the upper bound in \eqref{eq:rbounds}
 when $\kb=\ab=00\cdots 0$, and the lower bound is achieved
 when $\kb$ is $\bb=00\cdots 01$. By dividing all terms in \eqref{eq:rbounds} by $s_{\kb,n-1}$ and taking $n \to \infty$, we obtain 
\beq
2 - g_{\kb}^{-(L-1)} \ \leq \ g_\kb \ \leq 1 + g_\kb^{-1} + \cdots + g_\kb^{-L+1}. \label{eq:LCCDE4}
\eeq
To prove our claim that $g_\ab$ is the largest and $g_\bb$ is the smallest among all $g_\kb$, we first assume to the contrary that there exists $\hat\kb$ with $g_{\hat\kb} > g_\ab$. Substituting this into \eqref{eq:LCCDE4} leads to the following contradiction
\[
g_{\hat\kb} \  \stackrel{\text{(i)}}{<} \  \sum_{i=1}^L g_\ab^{-i+1} \  \stackrel{\text{(ii)}}{=} g_\ab,
\]
where (i) holds because $g_{\hat\kb}^{-1} < g_\ab^{-1}$ by assumption and (ii) is valid because $g_\ab^{-1}$ is a zero of $h_\ab(z)$ given in \eqref{ha}. 

To show $g_\bb$ achieves the minimum, again assume to the contrary that there exists $\hat\kb$ such that $g_{\hat\kb} < g_\bb$. Note from \eqref{eq:LCCDE4} that
\beq
0 \leq g_{\hat\kb} - 2 + g_{\hat\kb}^{-(L-1)} = (1-g_{\hat\kb}) \left( g_{\hat\kb}^{-L+1} + \cdots + g_{\hat\kb}^{-1} - 1 \right). \label{eq:gkbineq}
\eeq
Although $g_\kb \geq 1$ in general, 
we claim in this case $g_{\hat\kb} > 1$. For otherwise, that $h_{\hat\kb}(z=g_{\hat\kb}^{-1}=1)=0$ according to \eqref{eq:hkbz} implies that $r_{\hat\kb}(i)=0$ for all $i$; hence, $h_{\hat\kb}(z) = h_\bb(z)$ and $g_{\hat\kb} = g_\bb$, a contradiction. Now with $1 < g_{\hat\kb} < g_\bb$, the following series of inequalities lead to the desired contradiction:
\[
g_\bb \stackrel{\text{(i)}}{=} g_\bb^{-L+2} + \cdots + 1  \stackrel{\text{(ii)}}{<} g_{\hat\kb}^{-L+2} + \cdots + 1  \stackrel{\text{(iii)}}{\leq} g_{\hat\kb},
\]
where (i) follows from $h_\bb(z=g_\bb^{-1})=0$ and $g_\bb>1$,  (ii) holds because $g_\bb^{-1} < g_{\hat\kb}^{-1}$, and (iii) is due to \eqref{eq:gkbineq} and $g_{\hat\kb}>1$.
\end{proof}

\medskip

Using a similar technique in the proof of Theorem \ref{thm:2}, we can further devise a general upper bound and a general lower bound
for $g_\kb$ that hold for any $\kb$.

\medskip

\begin{thm} \label{thm:bounds}
For any UW $\kb$ of length $L\geq 2$, the asymptotic growth rate $g_\kb$ satisfies
\beq
2-2^{-(L-2)} \leq  g_\kb \leq 2-2^{-L}. \label{eq:thm3bounds}
\eeq
\end{thm}
\begin{proof} It is straightforward to see $s_{\kb,n-1} \leq s_{\kb,n} \leq 2 s_{\kb,n-1}$ and hence $1\leq g_\kb \leq 2$. 

To prove the upper bound, we assume without loss of generality that $g_\kb>1$ since
the upper bound trivially holds when $g_\kb=1$. We then derive
\[
g_\kb-1 \ = \  g_\kb(1-g_\kb^{-1})\  \stackrel{\text{(i)}}{\leq} \  1-g_\kb^{-L} \ \stackrel{\text{(ii)}}{\leq} \ 1 - 2^{-L},
\]
where (i) follows from multiplying both sides of the second inequality in \eqref{eq:LCCDE4} by $(1-g_\kb^{-1})$ with the fact $g_\kb>1$, and (ii) holds since $g_\kb \leq 2$. 

To establish the  lower bound, we use the following series of inequalities:
\bea
 g_\kb(1-g_\kb^{-1}) &=& g_\kb-1 \no\\
& \stackrel{\text{(i)}}{\geq} &  1-g_\kb^{-(L-1)}\no\\
&=& (1-g_\kb^{-1}) \left( 1+g_\kb^{-1}+g_\kb^{-2}+\cdots +g_\kb^{-(L-2)} \right)\no\\
& \stackrel{\text{(ii)}}{\geq} & (1-g_\kb^{-1}) \left( 1+2^{-1}+2^{-2}+\cdots +2^{-(L-2)} \right)\no\\
&=& (1-g_\kb^{-1}) \left( 2-2^{-(L-2)} \right), \label{eq:gkblb}
\eea
where (i) is from the first inequality in \eqref{eq:LCCDE4},  and (ii) holds because $g_\kb \leq 2$. Equipped with \eqref{eq:gkblb}, we next distinguish two cases to complete the proof.
\begin{enumerate}
\item When $L=2$, the lower bound is trivially valid and is actually achieved
by taking $\text{UW}=01$ as $g_{01}=1$ is
the multiplicative inverse of the smallest zero of polynomial $h_{01}(z)=1-2z+z^2=(1-z)^2$. \medskip

\item For $L > 2$, it suffices to show $g_\kb>1$. Assume to the contrary that there exists $\kb$ of length $L > 2$ such that $g_\kb = 1$. By $h_\kb(z=g_\kb^{-1}=1) = 0$ and \eqref{eq:hkbz}, we have $r_\kb(i)=0$ for all $i$ and hence $h_\kb(z) = 1 - 2 z + z^L$. Since $g_\kb$ is the multiplicative inverse of the smallest zero of $h_\kb(z)$, the absolute values of all the remaining zeros of $h_\kb(z)$, say $\lambda_1, \ldots, \lambda_{L-1}$, must be strictly larger than $1$. It then follows from the splitting of $h_\kb(z)$, i.e.,  
\[
h_\kb(z) \ = \ (z-1) \prod_{i=1}^{L-1} (z- \lambda_i), 
\]
the constant term of $h_{\kb}(z)$ must have absolute value $\prod_{i=1}^{L-1} \abs{\lambda_i} > 1$, contradicting to the fact that the constant term in polynomial $h_\kb(z) = 1 - 2z + z^L$ is $1$. 
\end{enumerate}
\end{proof}

\medskip

Theorem \ref{thm:bounds} provides concrete explicit expressions for 
both upper and lower bounds on $g_\kb$. Although the bounds are asymptotically tight
and well approximate the true $g_\kb$ for moderately large $L$, they are not sharp in general. We can actually refine them using Theorem \ref{thm:2} and obtain that
$g_\bb\leq g_\kb\leq g_\ab,$
where from the proof of Theorem \ref{thm:2}, we have 
$$g_\ab = \max\left\{ \abs{t}  : h_\ab(z=t^{-1})=0, t \in \C\right\}$$ 
and 
$$g_\bb = \max\left\{ \abs{t}  : h_\bb(z=t^{-1})=0, t \in \C\right\}.$$ 
The determination of $g_\ab$ and $g_\bb$ can be done via
finding the largest $\abs{s}$ and $\abs{t}$, $0 \neq s, t \in \C$, such that 
$h_\ab(s^{-1})=0$ and $h_\bb(t^{-1})=0$, respectively.
By noting that 
\begin{eqnarray*}
(z-1) \left[ z^L h_\ab \left( z^{-1} \right) \right] &=& (z-1)(z^L-z^{L-1}- \cdots - 1) \\
&=& z^{L+1} - 2 z^L + 1
\end{eqnarray*}
and
$$z^L h_\bb \left( z^{-1} \right)=z^L - 2 z^{L-1} + 1,$$
we conclude the following corollary.

\medskip

\begin{cor}
Let $\ab=0 \ldots 0$ and $\bb=0 \ldots 01$ be binary streams of length $L$. Then
for any $\kb$ of the same length to $\ab$ and $\bb$,
$$g_\bb\leq g_\kb\leq g_\ab.$$
In addition, $g_\ab=\alpha_{L+1}$ and $g_\bb=\alpha_L$, where
\[
\alpha_L \ := \ \max\{ \abs{t} \ : \ t^L-2t^{L-1}+1=0, \, t \in \C \}.
\]
In particular, we have $\alpha_L\approx 2-2^{-L+1}$ for large $L$.
\end{cor}

\medskip

Based on Theorem \ref{thm:bounds}, the following corollary is immediate by taking $L$ to infinity.

\medskip

\begin{cor} \label{cor:largeL}
For any UW $\kb$ of length $L$, the asymptotic growth rate of the corresponding UDOOC approaches $2$ as $L \to \infty$, i.e., 
\[
\lim_{L \to \infty} g_{\kb} \ = \ 2. 
\]
\end{cor}

\medskip

In Table \ref{eigtab},  we illustrate the asymptotic growth rates of UDOOCs for UWs $\ab$ and $\bb$ with lengths up to $8$. Also shown are the bounds in Theorem \ref{thm:bounds}.
It is seen that for moderately large $L$, all UDOOCs have roughly the same asymptotic growth rate, and hence are about the same good in terms of compressing sources of large size.
Furthermore, having $g_\kb \to 2$ as $L \to \infty$ means that for very large $L$,
UDOOCs can have asymptotic growth rates comparable 
to the unconstrained OOC, whose asymptotic growth rate equals $2$.

\begin{table}[!h]
\begin{center}
\caption{The asymptotic growth rates for UWs $\ab$ and $\bb$ and the bounds in 
Theorem \ref{thm:bounds} with various $L$}\label{eigtab}
\scriptsize
\begin{tabular}{  c | c c c c c c c }
\hline $L$ & 2 & 3 & 4 & 5 & 6 & 7 & 8 \\
\hline $2-2^{-L}$&1.75&1.875&1.938&1.969&1.984&1.992&1.996\\
\hline $g_\ab$ & 1.618 & 1.839 & 1.928 & 1.966 & 1.984 & 1.992 & 1.996 \\
\hline $g_\bb$ & 1 & 1.618 & 1.839 & 1.928 & 1.966 & 1.984 & 1.992 \\
\hline $2-2^{-(L-2)}$&1&1.5&1.75&1.875&1.938&1.969&1.984\\
\hline
\end{tabular} 
\end{center}
\end{table}

\subsection{Asymptotic Equivalence} \label{sec:asymp}

After presenting the results on asymptotic growth rates, we proceed to define
the asymptotic equivalence for UWs and show that the number of asymptotic equivalent UW classes is upper bounded by the number of different overlap vectors in Definition \ref{definition-ao}.

\medskip

\begin{defn} \label{defn:eqrln2}
Two UWs $\kb$ and $\kb'$ are said to be \emph{asymptotically equivalent}, denoted by $\kb \stackrel{\text{a.e.}}{\equiv} \kb'$, if  they have the same growth rate, i.e.,
$g_{\kb}=g_{\kb'}$.
\end{defn}

\medskip

Following the definition, we have the next proposition.

\medskip

\begin{prop} \label{prop4}
Fix the length $L$ of UWs,
and denote by $N_L$ the number of all possible overlap vectors of length $L$, i.e., $N_L = \abs{\left\{ \underline{r}_\kb: \kb \in \F^L\right\}}$.
Then, the number of asymptotically equivalent UW classes is upper-bounded by 
$N_L$. 
\end{prop}
\begin{proof}
Since the growth rate of  $s_{\kb,n}$ is given by $\max\left\{ \abs{t} : h_{\kb} \left( z = t^{-1} \right)=0, \, t \in \C\right\}$, in which the polynomial $h_{\kb}(z)$, defined in \eqref{eq:hkbz}, is completely determined by the respective overlap vector ${\underline r}_\kb$. As two different polynomials $h_\kb(z)$ and $h_{{\kb}'}(z)$,
resulting respectively from two different overlap vectors ${\underline r}_\kb$ and ${\underline r}_{{\kb}'}$,
could yield the same growth rate, 
the number of distinct asymptotic growth rates of $s_{\kb,n}$ for various $\kb$
must be upper-bounded by $N_L$.
The proof is then completed after invoking the result from Proposition \ref{prop:connect} 
that $s_{\kb,n}$ and $c_{\kb,n}$ have the same growth rate.
\end{proof}

\medskip

One may find the number of asymptotically equivalent UW classes by a brutal force algorithm when $L$ is small.
 With the help of Proposition \ref{prop4}, an efficient algorithm for its upper bound $N_L$ 
 is available in \cite{MATH03}, in which $\underline\rb_\kb$ is regarded as \textit{(auto)correlations} of a string. Values of $N_L$ for various $L$ are accordingly listed in Table \ref{kappa}. This table shows the trend, as being pointed out in \cite{MATH81}, that $\ln N_L$ grows at the speed of $(\ln L)^2$, or specifically, 
 \begin{equation}\label{30}
\dfrac{1}{2 \ln 2}\leq \liminf_{L\rightarrow\infty}\dfrac{\ln N_L}{ { \left( \ln L \right) }^2} \leq\limsup_{L\rightarrow\infty}
\dfrac{\ln N_L}{ { \left( \ln L \right) }^2}\leq \dfrac{1}{2 \ln \frac{3}{2}}.
\end{equation}
 
\begin{table}[!h]
\begin{center}
\caption{$N_L$ values for various $L$.
It is stated in \cite{MATH03} that the lower asymptotic bound $1/(2\ln(2))\approx 0.72$
 in \eqref{30} only holds for very large $L$; hence, this lower bound is not valid for $L\leq 13$ in this table . }
\label{kappa}
\scriptsize
\begin{tabular}{  c | p{0.1cm} p{0.1cm} p{0.1cm} p{0.1cm} p{0.1cm} p{0.1cm} p{0.1cm} p{0.1cm} p{0.1cm} p{0.1cm} p{0.1cm} p{0.1cm} p{0.1cm} }
\hline $L$ & 1 & 2 & 3 & 4 & 5 & 6 & 7 & 8 & 9 & 10 & 11 & 12 & 13\\
\hline $N_L$ & 1 & 2 & 3 & 4 & 6 & 8 & 10 & 13 & 17 & 21 & 27 & 30 & 37 \\
\hline 
$\frac{\ln N_L}{(\ln L)^2}$&$-$ & 1.44 & .91 & .72 & .69 & .65 & .61 & .59 & 59 & .57 & .57 & .55 & .55 \\
\hline
\end{tabular} 
\end{center}
\end{table}

\section{Encoding and Decoding Algorithms of UDOOCs}\label{sec:4}

In this section, the encoding and decoding algorithms of UDOOCs are presented.
Also provided are upper bounds for the averaged codeword length 
of the resulting UDOOC.

Denote by ${\cal U} = \{ u_1, u_2, \cdots, u_M\}$ the 
source alphabet of size $M$ to be encoded.
Assume without loss of generality that $p_1\geq p_2\geq \cdots\geq p_M$, 
where $p_i$ is the probability of occurrence for source symbol $u_i$.

Then, an optimal lossless source coding scheme for UDOOCs associated with UW $\kb$ should
assign codewords of shorter lengths to messages with higher probabilities
and reserve longer codewords for less likely messages.
By following this principle, 
the encoding mapping $\phi_\kb$ from ${\cal U}$ to ${\cal C}_\kb$
should satisfy $\ell(\phi_\kb(u_i))\leq\ell(\phi_\kb(u_j))$ whenever $i\leq j$,
where $\ell(\phi_\kb(u_i))$ denotes the length of bit stream $\phi_\kb(u_i)$.
The coding system thus requires an ordering of the words
in ${\cal C}_\kb$ according to their lengths.
This can be achieved in terms of the recurrence equation for  
$c_{\kb,n}$ (for example, \eqref{eq:ch3}).
As such, $\phi_\kb(u_1)$ must be the null word, and the mapping $\phi_\kb$ must always form a bijection mapping between
$\{ u_i : F_{\kb,n-1}< i \leq  F_{\kb,n}\}$ and ${\cal C}_\kb(n)$ for every integer $n\geq 1$, where
\beq
F_{\kb,n} := 
	 			\begin{cases}
	                       \sum_{i=0}^n c_{\kb,i} , & \mbox{if } n \geq 0, \\
						   0, & \textnormal{otherwise.} \\
				\end{cases}
				\label{eq:Fkbn}
\eeq
This optimal assignment results in average codeword length:
\beq
L_\kb \ = \ \ell(\kb) + \sum_{i=1}^M p_i \cdot\ell(\phi_\kb(u_i)), \label{eq:L_kb}
\eeq
where the first term $\ell(\kb)$ accounts for the insertion of UW $\kb$ to separate adjacent codewords. 

\subsection{Upper Bounds on Average Codeword Length of UDOOCs}

The average codeword length $L_\kb$ is clearly a function of the source distributions and does not in general exhibit a closed-form formula. In order to understand the general compression performance of UDOOCs, 
three upper bounds on $L_\kb$ are established in this subsection. 
The first upper bound is applicable to the situation when
the largest probability $p_1$ of source symbols is given. 
Other than $p_1$, the second upper bound additionally requires
the knowledge of the source entropy.
When both the largest and second largest probabilities (i.e., $p_1$ and $p_2$) of source symbols are present apart from the source entropy, the third upper bound can be used. 
Note that the third upper bound holds for all UWs and requires no knowledge about $\kb$; therefore,  
one might predict  that the third upper bound 
could be looser than the other two.
Experiments using English text from \emph{Alice's Adventures in Wonderland} however
indicate that such an intuitive prediction 
is not always valid.
Nevertheless, the second upper bound is better than the first one in most cases we have examined.
Details are given below.

\medskip

\begin{prop}[The first upper bound on $L_\kb$] \label{prop:upper1}For UW $\kb$ of length $L$, the average codeword length $L_\kb$ is upper-bounded as follows:
\beq\label{firstbound}
L_{\kb} \leq L+(1-p_1)N_\kb
\eeq
where $N_\kb$ is the smallest integer such that $F_{\kb, N_{\kb}} \geq M$.
\end{prop}
\begin{proof} 
It can be derived from \eqref{eq:L_kb} and $\ell(\phi_\kb(u_1))=0$ that
\begin{eqnarray*}
L_{\kb} &=& L +  \sum_{i=2}^M p_i \ell(\phi_\kb(u_i))\\
&\leq &L + \sum_{i=2}^M p_i \ell(\phi_\kb(u_M))\\
&=&L+(1-p_1)N_\kb.
\end{eqnarray*}
\end{proof}

\medskip

\begin{prop}[The second upper bound on $L_\kb$] \label{prop:upk}\label{prop:upper2} Suppose $g_\kb>1$. Then
\begin{multline}
L_{\kb} \leq L+\dfrac{\text{H}({\cal U})+p_1\log_{2}(p_1)}{\log_2 (g_\kb)}\\
+(1-p_1)(1-\log_{g_\kb}(K_\kb))\label{secondbound}
\end{multline}
where $\text{H}({\cal U})=\sum_{i=1}^Mp_i\log_2(1/p_i)$ is the source entropy with units in bits,
$K_\kb$ is a constant given by
\beq
K_\kb = \min\left\{g_\kb^{1-n_i}F_{\kb, n_i-1} \, : \, i=2, \cdots, M\right\}, \label{eq:K}
\eeq
and $n_i$ is the smallest integer satisfying
$F_{\kb,n_i}\geq i$.
\end{prop}
\begin{proof}
From the definitions of $n_i$ and $K_\kb$ we have 
\[
K_\kb\ g_\kb^{n_i-1} \leq F_{\kb, n_i-1}  < i \leq \frac{1}{p_i}
\]
where the last inequality follows from that $p_i \leq \frac{1}{i}$ for $1\leq i\leq M$ as 
$p_1\geq p_2\geq\cdots\geq p_M$. By $g_\kb > 1$ the above implies
\[
n_i \leq 1-\log_{g_\kb}\left(K_\kb p_i\right)
=1-\log_{g_\kb}\left(p_i\right)-\log_{g_\kb}\left(K_\kb\right).
\]
Note that $\ell(\phi_\kb(u_i)) \leq n_i$ by the property of optimal lossless compression function $\phi_\kb$. Consequently, we have 
\bean
L_\kb &=& L +  \sum_{i=2}^M p_i \ell(\phi_\kb(u_i)) \\
&\leq & L + \sum_{i=2}^M p_in_i\\
&\leq & L + \sum_{i=2}^M p_i\left(1-\log_{g_\kb}(p_i)-\log_{g_\kb}(K_\kb)\right)\\
&=& L -\sum_{i=2}^M p_i \log_{g_\kb}\left(p_i\right)+\sum_{i=2}^M p_i \left(1-\log_{g_\kb}(K_\kb)\right)\\
&=& L+\dfrac{\text{H}({\cal U})+p_1\log_{g_\kb}(p_1)}{\log_2 (g_\kb)}
+(1-p_1)(1-\log_{g_\kb}K_\kb).
\eean
\end{proof}

\medskip

The previous two upper bounds require the computations of either $N_\kb$, or $g_\kb$ and $K_\kb$; hence, they are functions of UW $\kb$. Next we provide a simple third upper bound that holds universally for all UWs.

\medskip

\begin{prop}[The third upper bound on $L_\kb$]\label{prop:upper3} For UW $\kb$ of length $L>2$, 
\begin{multline}
L_\kb \leq L+\dfrac{\text{H}({\cal U})+p_1\log_2(p_1)+p_2\log_2(p_2)}{\log_2(2-2^{2-L})}\\
+(2-2p_1-p_2).\label{thirdbound}
\end{multline}
\end{prop}
\begin{proof}
First, we claim that 
\begin{equation}\label{claim1}
c_{\kb,n}\geq 2^{n-2}\quad\text{for }2\leq n\leq L+1.
\end{equation}
This claim can be established
by showing that for any binary sequence $\bb=b_1 \ldots b_m \in {\mathbb F}^m$, where $0\leq m=n-2 \leq L-1$, there exist a prefix bit $p$ and a suffix bit $q$, where $p$, $q \in {\mathbb F}$, such that $\kb$ is not an internal subword of $\kb p \bb q \kb$.
This can be  done in two steps: i) there exists $q \in \F$ such that $\kb$ is not a subword of $\ub:=\bb q k_1^{L-1}$, and
ii) there exists $p \in \F$ such that $\kb$ is not an internal subword of $\kb p\ub$.

Because the first step trivially holds when $m=0$, we only need to focus on the case of  $m>0$.
Utilizing the prove-by-contradiction argument, we suppose that $\kb$ is a subword of both $\bb q k_1^{L-1}$ and $\bb \bar{q} k_1^{L-1}$, where $\bar{q}=1-q$. This implies the existence of indices $1\leq i <j \leq m+1$ such that 
\[
\underbrace{b_i \cdots b_m q k_1 \cdots k_{L-m+i-2}} _{=\,\db} = \underbrace{b_j \cdots b_m \bar{q} k_1 \cdots k_{L-m+j-2}} _{=\,\tilde\db} = \kb
\]
where we abuse the notations to let
$$
\db=\begin{cases}
\bb q k_1\cdots k_{L-m-1},&\text{if }i=1\\
b_i \cdots b_m q k_1 \cdots k_{L-m+i-2},&\text{if }1<i<m+1\\
q k_1\cdots k_{L-1},&\text{if }i=m+1
\end{cases}
$$
and similar notational abuse is applied to $\tilde\db$ and $j$.
After canceling out common terms in the respective sums of the first $(m+2-i)$ bits of $\db$ and $\tilde\db$, we obtain 
\[
b_i + \cdots + b_{j-1} + q = \bar{q} + k_1+ \cdots + k_{j-i}.
\]
Since $\db=\kb$, the above then implies $q=\bar{q}$, which leads to a contradiction.
The validity of the first step is verified.

After verifying $\ub=\bb q k_1^{L-1}\in{\cal S}_\kb(m+L)$, we can follow the proof of 
Proposition \ref{prop:connect} to confirm the second step (See the paragraph regarding \eqref{p4-1} and \eqref{p4-2}). The claim in \eqref{claim1} is thus validated.
Note that the equality in \eqref{claim1} holds when 
$\kb$ is all-zero or all-one. 

Next, we note also from the proof of Proposition \ref{prop:connect} that
$c_{\kb,n+2} \geq s_{\kb,n}$ for $n\geq L$.
Since $s_{\kb,L}=2^L-1$, we immediately have $c_{\kb,L+2}\geq 2^L-1$.
On the other hand, we can obtain from \eqref{eq:rbounds} 
that\,\footnote{
We can prove \eqref{TBD1} by induction.
Extending the definition of ${\cal S}_\kb(n)$ in \eqref{eq:skb}, we obtain that $s_{\kb,n}=2^n$ for $0 \leq n < L$.
This implies
$$ \frac{s_{\kb,L}}{s_{\kb,L-1}}=\frac{2^L-1}{2^{L-1}}=2-2^{1-L} \geq 2-2^{2-L}$$
and 
$$\frac{s_{\kb,m}}{s_{\kb,m-1}}=\frac{2^m}{2^{m-1}}=2\geq 2-2^{2-L} 
\text{ for all }1\leq m<L.$$
Now we suppose that for some $n\geq L$ fixed, \eqref{TBD1} is true for all $1\leq m\leq n$, i.e.,
$$
\dfrac{s_{\kb,m}}{s_{\kb,m-1}} \geq 2-2^{2-L}\text{ for all }1\leq m\leq n.
$$
Then, we derive by \eqref{eq:rbounds} that
\begin{eqnarray*}
\frac{s_{\kb,n+1}}{s_{\kb,n}} &\geq& 
2-\frac{s_{\kb,n-L+1}}{s_{\kb,n}}\geq 2- \dfrac{s_{\kb,n-L+1}}{s_{\kb,n-L+1}(2-2^{2-L})^{L-1}}\\
&=&2-(2-2^{2-L})^{1-L}
\geq 2-2^{2-L}.
\end{eqnarray*}
This completes the proof of \eqref{TBD1}.
 }
\beq\label{TBD1}
\dfrac{s_{\kb,n}}{s_{\kb,n-1}} \geq 2-2^{2-L}\text{ for }n\geq L.
\eeq
This concludes:
\beq\label{cbound}
    c_{\kb,n} \geq \begin{cases}
	                       1, & \mbox{if }  0 \leq n \leq 1, \\
	                       2^{n-2}, & \mbox{if }  2 \leq n \leq L+1, \\
						   {(2-2^{2-L})}^{n-L-2} (2^L-1), & \mbox{if } n \geq L+2,  
		         \end{cases}
\eeq
where $c_{\kb,0}=1$ because ${\cal C}_{\kb}(0)$ contains only the null codeword,
and $c_{\kb,1}\geq 1$ can be verified again by that $\kb$ cannot be the internal subword
of both $\kb p\kb$ and $\kb\bar p\kb$.\,\footnote{If it were not true, then there exist indices $i$ and $j$, $2\leq i<j\leq L+1$, such that $\kb=k_i\cdots k_L p k_1\cdots k_{i-2}=k_j\cdots k_L \bar p k_1\cdots k_{j-2}$; hence, $p-k_i=\bar p-k_j$ with $k_i=k_j=k_1$. The desired contradiction is obtained.
} 
The lower bound  \eqref{cbound} then indicates that if $2^L-1 \geq (2-2^{2-L})^L$ for $L > 2$, we can immediately have the following
exponential lower bound for $c_{\kb,n}$, i.e., 
\beq\label{cbound2}
    c_{\kb,n} \geq \begin{cases}
	                       1, & \mbox{if }  0 \leq n \leq 1, \\
	                       (2-2^{2-L})^{n-2}, & \mbox{if }  n\geq 2.
		         \end{cases}
\eeq
A stronger claim of $2^L-1 \geq (2-2^{2-L})^L$ for $L > 0$ simply follows from 
$$2^L-1-(2-2^{2-L})^L >2^L-1-2^{L-1}(2-2^{2-L})=1.$$ 
Hence, codeword lengths of the optimal UDOOC code
must satisfy:\,\footnote{
By \eqref{eq:Fkbn} and \eqref{cbound2}, we have that for $i \geq 3$ and $n_i=\ell(\phi_\kb(u_i))$,
\begin{eqnarray*}
i> F_{\kb,n_i-1}=\sum_{t=0}^{n_i-1}c_{\kb,t}\geq 2+\frac{(2-2^{2-L})^{n_i-2}-1}{1-2^{2-L}}
\end{eqnarray*}
which implies
$
\log_{2-2^{2-L}}[(i-2)(1-2^{2-L})+1]+2> n_i=\ell(\phi_\kb(u_i)).
$
Since $(i-2)(1-2^{2-L})+1\leq i$ for $i\geq 2-2^{L-2}$, we obtain
$$\ell(\phi_\kb(u_i))<\log_{2-2^{2-L}}[(i-2)(1-2^{2-L})+1]+2
\leq \log_{2-2^{2-L}}(i)+2.$$
}
\[
\ell(\phi_\kb(u_i)) \leq \log_{2-2^{2-L}}(i)+2 \quad\text{for }i \geq 3. 
\]
Consequently,
\begin{eqnarray}
L_\kb &= & L +  \sum_{i=2}^M p_i \ell(\phi_\kb(u_i)) \nonumber\\
& = & L + p_2 +\sum_{i=3}^M p_i \ell(\phi_\kb(u_i)) \nonumber\\
&\leq & L+p_2 + \sum_{i=3}^M p_i \left(\log_{2-2^{2-L}}(i)+2\right)\nonumber\\
&=& L+2-2p_1-p_2+\sum_{i=3}^M p_i \log_{2-2^{2-L}}(i)\nonumber\\
&\leq & L+2-2p_1-p_2+\sum_{i=3}^M p_i \log_{2-2^{2-L}}\left(\frac{1}{p_i}\right)\nonumber\\
&&\label{e1}\\
&=& L+2-2p_1-p_2\nonumber\\
&&\qquad \qquad +\dfrac{\text{H}({\cal U})+p_1\log_2(p_1)+p_2\log_2(p_2)}{\log_2(2-2^{2-L})}\nonumber
\end{eqnarray}
where \eqref{e1} follows from that $p_1\geq p_2\geq\cdots\geq p_i$ implies
$p_i \leq \frac{1}{i}$ for $1\leq i \leq M$.
\end{proof}

\medskip

We next study the asymptotic compression performance of UDOOCs, i.e., the situation when the source has infinitely many alphabets. Note first that with complete knowledge of the source statistics $\{p_i : i=1, \ldots, M\}$, the upper bound \eqref{secondbound} in Proposition \ref{prop:upper2} can be reformulated using similar arguments as 
\begin{multline}
L_{\kb} \leq L+\dfrac{\text{H}({\cal U})+p_1\log_{2}(p_1)}{\log_2 (g_\kb)}\\
+(1-p_1)(1-\log_{g_\kb}(T_\kb))\label{fourthbound}
\end{multline}
where $T_\kb$ is given by
\beq
T_\kb = \min\left\{g_\kb^{1-n_i}F_{\kb, n_i-1} \, : \, i=2, \cdots, M\right\}, \label{eq:T}
\eeq
and $n_i$ is the smallest integer satisfying $F_{\kb,n_i}\geq \frac{1}{p_i}$. Secondly, we can further extend the above upper bound \eqref{fourthbound} to the case of 
grouping $t$ source symbols (with repetition) to form a new ``grouped" source for UDOOC compression. 
The alphabet set of the new source is therefore ${\cal U}^t$ of size $M^t$.
Let $L_{\kb,t}$ be the per-letter average codeword length of 
UDOOCs for the $t$-grouped source. 
Then, applying \eqref{fourthbound} to the $t$-grouped source yields the following upper bound on $L_{\kb,t}$
\begin{multline}
L_{\kb,t} \leq \frac 1t\bigg(L+\dfrac{\text{H}({\cal U}^t)+q_1\log_{2}(q_1)}{\log_2 (g_\kb)}\\
+(1-q_1)(1-\log_{g_\kb}(T_{\kb,t}))\bigg), \label{eq:Lkbt}
\end{multline}
where 
\beq
T_{\kb,t} = \min\left\{g_\kb^{1-n_{i,t}} F_{\kb, n_{i,t}-1} \, : \, i=2, \cdots, M^t\right\}, \label{eq:Tt}
\eeq
$n_{i,t}$ is the smallest integer satisfying $F_{\kb, n_{i,t}} \geq \frac{1}{q_i}$, 
and $q_i$ is the $i$th largest probability of the grouped source.
For independent and identically distributed (i.i.d.) source, we have
$\text{H}({\cal U}^t)=t\,\text{H}({\cal U})$. Moreover, assuming $M > 1$ and $p_1 < 1$ for the nontrivial i.i.d. sources, we have $q_1 = p_1^t \to 0$ as $t \to \infty$, and 
$T_{\kb,t}$ can be shown to  converge to some finite positive constant
\begin{eqnarray*}
T_{\kb,\infty} &:=& \lim_{t\rightarrow \infty} T_{\kb,t} \\
&=& \left. \dfrac{\underline{x}_\kb^\top \adj \left( \mI - \Am_\kb z \right) \Am_\kb^{L-1} \underline{y}_\kb}{\det\left( \mI - \Am_\kb z \right)}(1-g_\kb z)\right|_{z=g_\kb^{-1}},
\end{eqnarray*}
where the last step follows from 
the conventional expansion theory for power series and also from the fact of
$g_\kb$ being the unique maximal eigenvalue of the adjacency matrix $\Am_\kb$ under $L>2$ (cf.~Proposition~\ref{prop:kUnique}). To elaborate, from the power series expansion, we have that $c_{\kb,n}= \sum_i a_{i,n}\lambda_i^n+c$, where $c$ is some constant, $\{\lambda_i\}$ is the set of nonzero distinct eigenvalues of $\Am_\kb$, and $a_{i,n}$ is the coefficient associated with $\lambda_i$ (which could a polynomial function of $n$ if $\lambda_i$ has algebraic multiplicity larger than one). In particular, assuming $\lambda_1$ is the largest eigenvalue, we can establish
that $T_{\kb,\infty}=a_{1,n}=a_1$, where the second equality emphasizes that $a_{1,n}$ is a constant independent of $n$ since $\lambda_1^{-1}=g_{\kb}^{-1}$ is a simple zero for $h_{\kb}(z)$ when the digraph $G_{\kb}$ is strongly connected.

By taking limits (letting $t \to \infty$) on both sides of \eqref{eq:Lkbt} and by noting that $\lim_{t \to \infty} q_1 = \lim_{t \to \infty} p_1^t = 0$ and $T_{\kb,\infty}$ is some finite positive constant, we summarize the asymptotic compression performance of UDOOCs in the next proposition.

\medskip

\begin{prop}\label{prop:upper5}
Given $g_\kb>1$ and a nontrivial i.i.d.~source, we have
\beq\label{asymbound}
\lim_{t\rightarrow \infty} L_{\kb,t} \leq  \dfrac{\text{H}({\cal U})}{\log_2 (g_\kb)}
\leq \dfrac{\text{H}({\cal U})}{\log_2(2-2^{2-L})}.
\eeq
\end{prop}

\medskip

Two remarks are made based on Proposition \ref{prop:upper5}.
First, the larger asymptotic bound $\text{H}({\cal U})/\log_2(2-2^{2-L})$ in \eqref{asymbound} immediately gives 
$$\lim_{L\rightarrow\infty}\lim_{t\rightarrow \infty} L_{\kb,t}=\text{H}({\cal U}).$$
Hence, if both $t$ and $L$ are sufficiently large, the per-letter average codeword length of UDOOCs can achieve the entropy rate $\text{H}({\cal U})$ of the i.i.d.~source. 
Secondly, the bound of $\text{H}({\cal U})/\log_2(g_\kb)$ in \eqref{asymbound}
is actually achievable by taking the all-zero UW with the source being uniformly distributed.
In other words,
\beq
\lim_{t\rightarrow \infty} L_{\ab,t} = \dfrac{\text{H}({\cal U})}{\log_2 (g_\ab)}= \dfrac{\log_2(M)}{\log_2 (g_\ab)}, \label{eq:tobeproved}
\eeq
where $\ab=0 \ldots 0$. For better readability, we relegate the proof of \eqref{eq:tobeproved} to Appendix \ref{app:d}.

Tables \ref{ubindEN} and \ref{ubAlice} evaluate the bounds for 
the English text source with letter probabilities from \cite{oxfdic} 
and a true text source from \emph{Alice's Adventure in Wonderland} with empirical frequencies directly obtained from the book, respectively. The source alphabet of the English text and that from 
 \emph{Alice's Adventure in Wonderland} is of size $27$, 
where letters of upper and lower cases are regarded the same and
all symbols other than the 26 English letters are treated as one.
It can be observed from Table \ref{ubindEN} that
bound \eqref{firstbound} is always the best among all three bounds but still has a visible gap
to the resultant average codeword length $L_\kb$.
Table \ref{ubAlice} however 
shows that the three bounds may take turn to be on top of the other two.
For example, under $\kb=\ab$, \eqref{firstbound}, \eqref{secondbound} and \eqref{thirdbound}
are the lowest when $(L,t)=(3,1)$, $(L,t)=(5,2)$ and $(L,t)=(6,3)$, respectively.
Table \ref{ubAlice} also indicates that enlarging the value of $t$ may help
improving the per-letter average codeword length as well as the bounds of UDOOCs.
Comparison of the per-letter average codeword length of UDOOCs with
the source entropy will be provided later in the simulation section.

\begin{table}[!h]
\begin{center}
\caption{Upper bounds \eqref{firstbound}, \eqref{secondbound} and \eqref{thirdbound}
on the average codeword length $L_\kb$ of UDOOCs
for English text source with letter probabilities from \cite{oxfdic}. Here, $\ab=0\cdots 0$ and $\bb=0\cdots 01$.}\label{ubindEN}
\begin{tabular}{ c|c|r r r r}\hline
$\kb$&&  $L=3$ & $L=4$ &  $L=5$ &  $L=6$ \\\hline
&$L_\ab$& 6.432& 7.411& 8.411& 9.411 \\ \cline{2-6}
$\ab$& \eqref{firstbound} &   8.330& 9.330& 10.330& 11.330\\ \cline{2-6}
&\eqref{secondbound}  & 9.606& 10.496& 11.488&12.484\\ \hline
&$L_\bb$  & 5.215& 6.185& 7.185& 8.185 \\ \cline{2-6}
$\bb$ &\eqref{firstbound}   & 6.553& 7.553& 8.553& 9.553 \\ \cline{2-6}
 &\eqref{secondbound}& 10.385&10.206& 10.889& 11.769 \\ \hline
--&\eqref{thirdbound}& 10.831& 10.140& 10.652& 11.456 \\ \hline
\end{tabular} 
\end{center}
\end{table}

\begin{table}[!h]
\begin{center}
\caption{Upper bounds \eqref{firstbound}, \eqref{secondbound} and \eqref{thirdbound}
on the per-letter average codeword length $L_{\kb,t}$ of UDOOCs
for English text source from \emph{Alice's Adventure in Wonderland}. Here, $\ab=0\cdots 0$ and $\bb=0\cdots 01$.}\label{ubAlice}
\begin{tabular}{ c|c|c| r r r r } \hline
$\kb$&    &  & $L=3$ & $L=4$ & $L=5$ & $L=6$ \\ \hline 
& & $t=1$ & 5.773 & 6.757 & 7.757& 7.757  \\ \cline{3-7} 
&$L_{\ab,t}$ & $t=2$ & 4.498& 4.920& 5.397& 5.891 \\ \cline{3-7}
 & & $t=3$ & 3.862& 4.089& 4.388& 4.709\\ \cline{2-7}

& & $t=1$ & 7.459 & 8.459 & 9.459 & 10.459 \\ \cline{3-7}
 $\ab$&\eqref{firstbound}& $t=2$ & 6.569& 7.069& 7.569& 7.608\\ \cline{3-7}
 &  & $t=3$ & 5.770& 5.786& 6.119& 6.134\\ \cline{2-7}
 
  && $t=1$ & 8.700 & 9.596 & 10.585 & 11.580  \\ \cline{3-7} 
&\eqref{secondbound}& $t=2$ & 6.548& 6.886& 7.333& 7.813\\ \cline{3-7}
 & & $t=3$ & 5.771& 5.586& 6.120& 6.135\\ \hline 
 
 && $t=1$ & 4.792 & 5.774 & 6.774 & 7.774  \\ \cline{3-7}
 &$L_{\bb,t}$ & $t=2$ & 3.791& 4.134& 4.598& 5.090 \\ \cline{3-7}
 & & $t=3$ & 3.455& 3.532& 3.802& 4.115\\ \cline{2-7}

 & & $t=1$ & 6.716 & 6.973 & 7.973 & 8.973  \\ \cline{3-7} 
$\bb$&\eqref{firstbound} & $t=2$ & 6.108& 6.147&6.647&7.147\\ \cline{3-7}
 &  &$ t=3 $& 5.452& 5.150& 5.483& 5.816\\ \cline{2-7}

 & & $t=1$ & 9.399 & 9.366 & 10.089 & 10.984  \\ \cline{3-7}
 &\eqref{secondbound} & $t=2$ & 7.356& 6.819&7.040& 7.435 \\ \cline{3-7} 
& &$ t=3$ & 5.453& 5.150& 5.483& 5.817\\ \hline 

 & & $t=1$ & 9.676 & 9.221 & 9.801 & 10.632 \\ \cline{3-7}
 --&\eqref{thirdbound} & $t=2$ & 8.106 & 7.035 & 7.399 & 7.816 \\ \cline{3-7}  
   &        & $t=3$ & 6.947 & 5.815 & 5.726 & 5.889 \\ \hline 
\end{tabular} 
\end{center}
\end{table}

\subsection{General Encoding and Decoding Mappings for UDOOCs}

In this subsection, the encoding and decoding mappings for a UDOOC with 
general UW are introduced.

The practice of UDOOC requires the encoding function $\phi_\kb$ to be a bijective mapping 
between the subset of source letters ${\cal U}_\kb(n):=\{u_m : F_{\kb,n-1} < m \leq F_{\kb,n}\}$
and the set of length-$n$ codewords ${\cal C}_\kb(n)$ for all $n$.
Since the resulting average codeword length will be the same for
any such bijective mapping from ${\cal U}_\kb(n)$ to ${\cal C}_\kb(n)$, we are free to devise one that facilities efficient encoding and decoding of message $u_m$.
The bijective encoding mapping $\phi_\kb$ that we propose is described in the following.

We define for any binary stream $\db$ of length $\leq n$,
\beq\label{eq:cdn} 
{\cal C}_\kb(\db,n) := \left\{ \cb \in {\cal C}_\kb(n) \, : \, 
\db\text{ is a prefix of }\cb,\text{ or }\cb=\db
\right\}.
\eeq
Obviously, ${\cal C}_\kb(\db,n)\cap{\cal C}_\kb(\tilde\db,n)=\emptyset$
for every pair of distinct $\db$ and $\tilde\db$ of the same length,
and for any fixed $i$ with $1\leq i\leq n$,
\beq
{\cal C}_\kb(n) \ = \ \bigcup_{\db \in \F^i} {\cal C}_\kb(\db,n).\label{eq:prefixdecom} 
\eeq
Then, given message $u_m \in {\cal U}_\kb(n)$, i.e., the number $n$ is chosen such that $F_{\kb,n-1} < m \leq F_{\kb,n}$, the proposed encoding mapping $\phi_\kb$ produces the codeword $\phi_\kb(u_m)=c_1c_2\cdots c_n$ for source letter $u_m$ 
recursively according to the rule that for $i=1,2,\ldots,n$,
\beq\label{examine}
c_{i} = \left\{\begin{array}{cl}
0, & \text{ if $\rho_{i-1} \leq \abs{{\cal C}_\kb(c_1\cdots c_{i-1}0,n)}$}\\
1, & \text{ if $\rho_{i-1} > \abs{{\cal C}_\kb(c_1\cdots c_{i-1}0,n)}$}
\end{array}
\right.
\eeq
where the progressive metric $\rho_i$ is also maintained recursively as:
\begin{eqnarray}
\rho_{i} & := & \rho_{i-1} - c_i \abs{{\cal C}_\kb(c_1\cdots c_{i-1} 0,n)} \nonumber\\
&=& \left\{\begin{array}{ll}
\rho_{i-1}, & \text{ if $c_i=0$}\\
\rho_{i-1} - \abs{{\cal C}_\kb(c_1\cdots c_{i-1}0,n)}, & \text{ if $c_i=1$}
\end{array}
\right.\nonumber\\
\label{rhoupdate}
\end{eqnarray}
with an initial value $\rho_0 \  = \ m - F_{\kb,n-1}$.
This encoding mapping actually assigns codewords according to their lexicographical ordering.

\begin{ex}
Taking $\kb=010$ as an example, we can see
from Fig. \ref{010pctree} that the seven codewords of length $4$, i.e., 
$0000$, $0011$, $0110$, $0111$, $1100$, $1110$ and $1111$, will be respectively assigned 
to source letters $u_9$, $u_{10}$, $u_{11}$, $u_{12}$, $u_{13}$, $u_{14}$ and $u_{15}$.
The progressive metrics $\rho_0,\rho_1,\rho_2,\rho_3$ for source letter $u_{11}$ are $3,3,1,1$, respectively, 
with $\abs{{\cal C}_\kb(0,4)}=4$, $\abs{{\cal C}_\kb(00,4)}=2$, 
$\abs{{\cal C}_\kb(010,4)}=0$ and $\abs{{\cal C}_\kb(0110,4)}=1$.
\qed
\end{ex}
\medskip

Note again that given $m$ (equivalently, $u_m$), $n$ can be determined
via $F_{\kb,n-1} < m \leq F_{\kb,n}$.
At the end of the $n$th recursion, we must have
\beq
m \ = \ F_{\kb,n-1} + \sum_{i=1}^{n} c_i\abs{{\cal C}_\kb(c_1^{i-1}0,n)} + 1. \label{eq:meq}
\eeq
We emphasize that \eqref{eq:meq} actually gives the corresponding computation-based decoding function $\psi_\kb : {\cal C}_\kb(n) \to {\cal U}_\kb(n)$ for codewords $\cb$ of length $n$. 

One straightforward way to implement $\phi_\kb$ and $\psi_\kb$ is to pre-store the value
of $\abs{{\cal C}_\kb(\db 0,n)}$ for every $\db$ and $n$.
By considering the huge number of all possible prefixes $\db$ for each $n$,
this straightforward approach does not seem to be an attractive one.

Alternatively, we find that $\abs{{\cal C}_\kb(\db 0,n)}$ can be 
obtained through adjacency matrix $\Am_\kb$ introduced in Section \ref{sec:2}.
The advantage of this alternative approach is that there is no need to pre-store or pre-construct 
any part of the codebook ${\cal C}_\kb$, and the value of $\abs{{\cal C}_\kb(\db 0,n)}$ is computed only when it is required during the encoding or decoding processes.
Moreoever, for specific UWs such as $00\ldots0$, $00\ldots01$, and their binary complements,
we can further reduce  the required computations.

In the following subsections, we will first introduce the encoding and decoding algorithms
for specific UWs as they can be straightforwardly understood. Algorithms for general UWs require an additional computation of $\abs{{\cal C}_\kb(\db 0,n)}$ and will be presented in subsequent subsections.  

\subsection{Encoding and Decoding Algorithms for $\text{UW}=11\ldots 1$ }

It has been inferred from Proposition \ref{prop:comp} that
the encoding and decoding of the UDOOC with UW $\kb=00\ldots0$
can be equivalently done through the encoding and decoding 
of the UDOOC with UW $\kb=11\ldots1$
as one can be obtained from the other by binary complementing.
Thus, we only focus on the case of $\kb=11\ldots1$ in this subsection. 

For this specific UW, we observe that
a codeword $\cb = \db0\bb \in {\cal C}_\kb(\db0,n)$ if, and only if,
$0 \bb \in {\cal C}_\kb(n-\ell(\db))$ is a codeword 
of length $n-\ell(\db)$, where $\ell(\db)$ is the length of prefix bitstream $\db$.
We thus obtain
\beq\label{eq:case1p}
\abs{{\cal C}_\kb(\db0,n)} = c_{\kb,n-\ell(\db)}. 
\eeq
It can be shown that the LCCDE for $c_{\kb,n}$ with $\kb=11\ldots 1$ is 
\beq
c_{\kb,n} = \sum _{i=1}^L c_{\kb,n-i}, \ \textnormal{ for all $n > L$,} \label{eq:case1ckbn1}
\eeq
where the initial values are 
\beq
c_{\kb,n} \ = \ \left\{
\begin{array}{ll}
1, & n=0,1,2,\\
2^{n-2}, & n=3,\ldots,L.
\end{array} \right. \label{eq:case1ckbn2}
\eeq
Based on \eqref{eq:case1p}, \eqref{eq:case1ckbn1} and \eqref{eq:case1ckbn2},
the algorithmic encoding and decoding procedures can be described below.

\medskip

\begin{algorithm}[h!] 
\caption{Encoding of UDOOC with $\kb=11\ldots1$} \label{alg:1}
\begin{algorithmic}[1]
\REQUIRE Index $m$ for message $u_m$
\ENSURE Codeword $\phi_\kb(u_m)=c_1 \ldots c_n$
\STATE Compute $c_{\kb,0},c_{\kb,1}, c_{\kb,2},\ldots$ using \eqref{eq:case1ckbn1} and \eqref{eq:case1ckbn2} to determine the smallest $n$ such that $F_{\kb,n}\geq m$. If $n=0$, then $\phi_\kb(u_m)=\text{null}$ and stop the algorithm.
\STATE Initialize $\rho_0 \leftarrow m-F_{\kb,n-1}$
\FOR{$i=1$ to $n$}
\IF{$\rho_{i-1} \leq c_{\kb,n-i+1}$}
\STATE $c_{i} \leftarrow 0$ and $\rho_{i} \leftarrow \rho_{i-1}$
\ELSE
\STATE $c_{i} \leftarrow 1$ and $\rho_{i} \leftarrow \rho_{i-1}-c_{\kb,n-i+1}$ 
\ENDIF
\ENDFOR
\end{algorithmic}
\end{algorithm}

\medskip

\begin{algorithm}[h!] 
\caption{Decoding of UDOOC with $\kb=11\ldots1$} \label{alg:2}
\begin{algorithmic}[1]
\REQUIRE Codeword $\cb=c_1 \ldots c_n$
\ENSURE Index $m$ for message $u_m=\psi_\kb(\cb)$
\STATE Compute $c_{\kb,0},c_{\kb,1}, \ldots, c_{\kb,n}$ using \eqref{eq:case1ckbn1} and \eqref{eq:case1ckbn2} 
\STATE Initialize $m \leftarrow F_{\kb,n-1}+1$
\FOR{$i=1$ to $n$}
\IF{$c_{i}=1$}
\STATE $m\leftarrow m + c_{\kb,n-i+1}$
\ENDIF
\ENDFOR
\end{algorithmic}
\end{algorithm}

\subsection{Encoding and Decoding Algorithms for $11\cdots 10$ }

Again, Proposition \ref{prop:comp} infers
that the encoding and decoding of the UDOOC with UW $\kb=00\ldots 01$
can be equivalently done through the encoding and decoding 
of the UDOOC with UW $\kb=11\ldots 10$.
We simply take $\kb=11\cdots 10$ for illustration.

It can be derived from \eqref{eq:hkbz} that for $\kb=11\cdots 10$,
\beq
c_{\kb,n} = 2c_{\kb,n-1}-c_{\kb,n-L} \label{eq:case2ckbn1}
\eeq
with initial condition
\beq
c_{\kb,n}= \left\{
\begin{array}{ll}
1, & n=0\\
2^n, & n=1, \ldots, L-1\\
2^L-1, & n=L.
\end{array} \right. \label{eq:case2ckbn2}
\eeq
It remains to determine $\abs{{\cal C}_\kb(\db 0,n)}$.
Observe that 
 $\db 0 \bb \in {\cal C}_\kb(\db 0,n)$ if, and only if, $ \bb \in {\cal C}_\kb(n-\ell(\db)-1)$;
 hence, $ \vert {\cal C}_\kb(\db 0,n) \vert = c_{\kb,n-\ell(d)-1}$.  
 We summarize the encoding and decoding algorithms of UDOOCs with $\kb=11\ldots 10$
 in Algorithms \ref{alg:3} and \ref{alg:4}, respectively.
 
\medskip

\begin{algorithm}[h!]
\caption{Encoding of UDOOC with $\kb=11\cdots 10$} \label{alg:3}
\begin{algorithmic}[1]
\REQUIRE Index $m$ for message $u_m$
\ENSURE Codeword $\phi_\kb(u_m)=c_1 \ldots c_n$
\STATE Compute $c_{\kb,0},c_{\kb,1}, c_{\kb,2},\ldots$ using \eqref{eq:case2ckbn1} and \eqref{eq:case2ckbn2} to determine the smallest $n$ such that $F_{\kb,n} \geq m$.
If $n=0$, then $\phi_\kb(u_m)=\text{null}$ and stop the algorithm.
\STATE Initialize $\rho_0 \leftarrow m-F_{\kb,n-1}$
\FOR{$i=1$ to $n$}
\IF{$\rho_{i-1} \leq c_{\kb,n-i}$}
\STATE $c_{i} \leftarrow 0$ and $\rho_{i} \leftarrow \rho_{i-1}$
\ELSE
\STATE $c_{i} \leftarrow 1$ and $\rho_{i} \leftarrow \rho_{i-1}-c_{\kb,n-i}$
\ENDIF
\ENDFOR
\end{algorithmic}
\end{algorithm}

\medskip

\begin{algorithm}[h!]
\caption{Decoding of UDOOC with $\kb=11\cdots 10$} \label{alg:4}
\begin{algorithmic}[1]
\REQUIRE Codeword $\cb=c_1 \ldots c_n$
\ENSURE Index $m$ for message $u_m=\psi_\kb(\cb)$
\STATE Compute $c_{\kb,0},c_{\kb,1}, \ldots, c_{\kb,n}$ using \eqref{eq:case2ckbn1} and \eqref{eq:case2ckbn2} 
\STATE Initialize $m \leftarrow F_{\kb,n-1}+1$
\FOR{$i=1$ to $n$}
\IF{$c_{i}=1$}
\STATE $m\leftarrow m + c_{\kb,n-i}$
\ENDIF
\ENDFOR
\end{algorithmic}
\end{algorithm}

\subsection{Encoding and Decoding Algorithms for General UW $\kb$}

It is clear from the discussions in the previous two subsections as well as from 
\eqref{examine} that to determine $c_i$ in the encoding algorithm, we only need to keep track of the most recent $\rho_{i-1}$, instead of retaining sequentially all of $\rho_0,\ldots,\rho_{i-2}$.
We address the recursion for the update of $\rho_i$ in \eqref{rhoupdate}
only to facilitate our interpretation on the operation of the progressive metric.
The same approach will be followed in the presentation of the general encoding algorithm below,
where a progressive matrix $\Dm_i$ is used in addition to the progressive metric $\rho_i$.

The encoding algorithm for general UWs consists of two phases.
Given the index $m$, we first identify the smallest $n$ 
such that $F_{\kb,n} \geq m$. Note that the computation of $F_{\kb,n}$ requires the knowledge of $c_{\kb,n}$,
which can be recursively obtained using the LCCDE in \eqref{eq:recdkbn}. 
In the second phase, as seen from the two previous subsections,
we need to determine the cardinality of ${\cal C}_\kb(\db, n)$ for any prefix $\db$ with $\ell(\db)\leq n$. Thus, our target in this subsection is to provide an expression for
$\abs{{\cal C}_\kb(\db, n)}$ that holds for general $\kb$ and $\db$.

Define $E_{\kb,0}$ and $E_{\kb,1}$ for digraph $G_\kb=(V,E_\kb)$ as 
\bea
E_{\kb,0} &:=& \left\{ ({\boldsymbol i}, {\boldsymbol j}) \in E_\kb \ : \  j_{L-1} = 0
\right\} \label{eq:e0} \\
E_{\kb,1} &:=& \left\{ ({\boldsymbol i}, {\boldsymbol j}) \in E_\kb \ : \  j_{L-1} = 1
\right\}.\label{eq:e1}
\eea
Literally speaking, $E_{\kb,0}$ (resp. $E_{\kb,1}$) is the set of edges in $E_\kb$,
whose ending vertex has its last bit $j_{L-1}$ equal to $0$ (resp. $1$). 
Let $\Am_{\kb,0}$ and $\Am_{\kb,1}$ 
be the adjacency matrices respectively for digraphs $G_{\kb,0}=(V,E_{\kb,0})$ and $G_{\kb,1}=(V,E_{\kb,1})$. Obviously, $\Am_{\kb}=\Am_{\kb,0}+\Am_{\kb,1}$. 
Based on the two adjacency matrices, we derive
\beq\label{recurc}
\abs{{\cal C}_\kb(\db,n)} \ = \ \xb_\kb^\top \Dm_{i} \Am_\kb^{(n+L-1)-i}  \underline{y}_\kb
\eeq
where for a prefix stream $\db=d_1 \ldots d_i$, 
\beq\label{recurd}
\Dm_i := \prod_{t=1}^i \Am_{\kb,d_t},
\eeq
and
$\xb_\kb$ and $\yb_\kb$ are the initial and ending vectors for digraph $G_\kb$ defined in Section \ref{sec:3a}. 
With \eqref{recurc} and \eqref{recurd},
the general encoding and decoding algorithms
are given in Algorithms \ref{alg:5} and \ref{alg:6}, respectively. Verification of the two algorithms is relegated to Appendix \ref{app:c} for better readability. 

\medskip

\begin{algorithm}[h!]
\caption{Encoding of UDOOC with General $\kb$} \label{alg:5}
\begin{algorithmic}[1]
\REQUIRE Index $m$ for message $u_m$
\ENSURE Codeword $\phi_\kb(u_m)=c_1 \ldots c_n$
\STATE Compute $c_{\kb,0},c_{\kb,1}, c_{\kb,2},\ldots$ using \eqref{eq:recdkbn} and the method in Section \ref{sec:ex} to determine the smallest $n$ such that $F_{\kb,n} \geq m$.
If $n=0$, then $\phi_\kb(u_m)=\text{null}$ and stop the algorithm.
\STATE Initialize $\rho_0 \leftarrow m-F_{\kb,n-1}$ and $\Dm_{0} \leftarrow \mI$
\FOR{$i=1$ to $n$}
\STATE Compute $ \emph{dummy}\leftarrow \xb_\kb^\top \Dm_{i-1} \Am_{\kb,0} \Am_\kb^{(n+L-1)-i} \underline{y}_\kb$ 
\IF{$\rho_{i-1} \leq \emph{dummy}$ }
\STATE $c_{i} \leftarrow 0$, $\rho_{i} \leftarrow \rho_{i-1}$ and $ \Dm_{i} \leftarrow \Dm_{i-1} \Am_{\kb,0}$
\ELSE
\STATE $c_{i} \leftarrow 1$, $\rho_i \leftarrow \rho_{i-1} -\emph{dummy}$ and $\Dm_{i} \leftarrow \Dm_{i-1} \Am_{\kb,1}$
\ENDIF
\ENDFOR
\end{algorithmic}
\end{algorithm}

\begin{algorithm}[h!]
\caption{Decoding of UDOOC with General $\kb$} \label{alg:6}
\begin{algorithmic}[1]
\REQUIRE Codeword $\cb=c_1 \ldots c_n$
\ENSURE Index $m$ for message $u_m=\psi_\kb(\cb)$
\STATE Compute $c_{\kb,0},c_{\kb,1}, \ldots, c_{\kb,n}$ using \eqref{eq:recdkbn} and the method in Section \ref{sec:ex}. 
\STATE Initialize $m \leftarrow F_{\kb,n-1}+1$ and $\Dm_0 \leftarrow \mI$
\FOR{$i=1$ to $n$}
\IF{$c_{i}=1$}
\STATE $m\leftarrow m + \xb_\kb^\top \Dm_{i-1} \Am_{\kb,0}  \Am_\kb^{(n+L-1)-i} \underline{y}_\kb$
\ENDIF
\STATE $\Dm_i \leftarrow \Dm_{i-1} \Am_{\kb,c_i}$
\ENDFOR
\end{algorithmic}
\end{algorithm}

\subsection{Exemplified Realization of the Encoding and Decoding Algorithms 
for General UW $\kb$} \label{sec:ex}

The matrix expressions in \eqref{recurc} and \eqref{recurd}
facilitate the presentation of Algorithms \ref{alg:5} and \ref{alg:6} for general UW;
however, their
implementation involves extensive computation of matrix multiplications.
Since the entries in each row or column of $\Am_\kb$ are all $0$'s except for at most two $1$'s,
the complexity 
of computing 
\beq
c_{\kb,n}=\xb_\kb^\top\Am_\kb^{n+L-1}\yb_\kb
\eeq
 and 
 \beq
 \abs{{\cal C}_\kb(\db0,n)}=\xb_\kb^\top \Dm_{\ell(\db)} \Am_{\kb,0} \Am_\kb^{n+L-\ell(\db)-2}  \underline{y}_\kb \label{eq:computedkbdb0n}
 \eeq
is in fact relatively small. Furthermore, it is much easier to compute $c_{\kb,n}$ than $\abs{{\cal C}_\kb(\db0,n)}$. To see this, note from \eqref{eq:enumerate} that we have the following enumeration for $c_{\kb,n}$
\[
\sum_{n=0}^\infty c_{\kb,n} z^n  \ = \ \frac{\underline{x}_\kb^\top \adj\left( \mI - \Am_\kb z \right) \Am_\kb^{L-1} \yb_\kb}{\det\left( \mI - \Am_\kb z \right) }. 
\]
Thus, simply evaluating the RHS of the above equation gives the values of  $c_{\kb,n}$ for $n=1,2,\ldots,L-1$. The remaining values of $c_{\kb,n}$ for $n \geq L$ can be easily determined through the recursion formula \eqref{eq:recdkbn}. 

Another way to compute the values of $c_{\kb,n}$ can be easily obtained by modifying the algorithm for computing the values of $\abs{{\cal C}_\kb(\db0,n)}$, which we now discuss. The first step to compute $\abs{{\cal C}_\kb(\db0,n)}$ is to break up formula \eqref{eq:computedkbdb0n} into:
\beq
\abs{{\cal C}_\kb(\db0,n)}= \underbrace{(\xb_\kb^\top \Dm_{\ell(\db)})}_{\underline{u}^\top} \Am_{\kb,0} \Am_\kb^{n-\ell(\db)-1}  \underbrace{(\Am_\kb^{L-1}\yb_\kb)}_{\underline{w}_\kb}. \label{eq:break}
\eeq

We note that from the choice of $\db$ in the encoding algorithm \ref{alg:5}, we must have $\abs{{\cal C}_\kb(\db,n)} \geq 1$.\footnote{Given any choice of prefix $\db$, it is possible that $\abs{{\cal C}_\kb(\db,n)}=0$ if $\db \not\in {\cal C}_\kb$, and  in this case we have $\underline{u}= \underline{0}$ in \eqref{eq:break}. However, the prefix $\db$ considered in our encoding algorithm, Algorithm \ref{alg:5}, is always a prefix of some codeword; hence we have 
$\abs{{\cal C}_\kb(\db,n)} > 0$.} Hence, $\underline{u}=\Dm_{\ell(\db)}^\top\xb_\kb$ is actually a zero-one indication vector of length $2^{L-1}$ for the rightmost $(L-1)$ bits of $k_2^L\db$, i.e., 
all components of vector $\underline{u}=[u_1\ u_2\ \cdots u_{2^{L-1}}]^\top$
are $0$'s except the $(j+1)$th component (being $1$'s), where 
$j$ is the integer corresponding to the binary 
representation of the rightmost $(L-1)$ bits of $k_2^L\db$. Hence, $\underline{u}$ can be directly determined without any computation.
In addition, we can pre-compute $\underline{w}_\kb$
 since it is the same for all $n$ and $\db$.
Our task is therefore reduced to computing the value of 
$$\abs{{\cal C}_\kb(\db0,n)}= \underline{u}^\top \Am_{\kb,0} \Am_\kb^{n-\ell(\db)-1}\underline{w}_\kb.$$
Below, we demonstrate how to utilize a finite state machine based on the digraph $G_\kb$ to evaluate $\abs{{\cal C}_\kb(\db0,n)}$ without resorting to matrix operations. 

Notations that are used to describe the finite state machine are addressed first.
Let $\Sf= \lbrace s_{00 \cdots 0},s_{00 \cdots 01},\cdots,s_{11 \cdots 1} \rbrace$ be the set of states indexed by all binary bit-streams of length $L-1$.
We say $s_{\boldsymbol{i}}=s_{i_1\cdots i_{L-1}}$ is a \emph{counting state} if the $(i+1)$th component of $\underline{w}_\kb$ is $1$, where
$i$ is the integer corresponding to binary representation of $\ib=i_1\ldots i_{L-1}$.\footnote{
Here we implicitly use a fact that $\underline{w}_\kb$ is a binary zero-one vector. 
Note that the $(i+1)$th component of $\underline{w}_\kb=[w_1\ \ w_2\ \ \cdots\ \  w_{2^{L-1}}]^\top=\Am_\kb^{L-1} \underline{y}_\kb$ is equal to the number of distinct walks from vertex $i_1\cdots i_{L-1}$ to vertex $k_1\cdots k_{L-1}$ on digraph $G_\kb$. 
This fact follows since there is at most one walk of length $(L-1)$ between the above two vertexes.}
Denote by $\Cf_\kb$
the set of all counting states corresponding to $\kb$. 
Also, for each state $s_\kb \in \Sf$ we define 
\bean
{\cal I}(s_{\ib}) &=& \left\{ s_\jb \ : \ (\jb, \ib) \in E_\kb\right\},\\
{\cal O}(s_{\ib}) &=& \left\{ s_\jb \ : \ (\ib, \jb) \in E_\kb\right\},\\
{\cal I}_b(s_{\ib}) &=& \left\{ s_\jb \ : \ (\jb, \ib) \in E_{\kb,b}\right\},\\
{\cal O}_b(s_{\ib}) &=& \left\{ s_\jb \ : \ (\ib, \jb) \in E_{\kb,b}\right\},
\eean
for $b=0,1$, where the edge-sets $E_{\kb,0}$ and $E_{\kb,1}$ are defined in \eqref{eq:e0} and \eqref{eq:e1}, respectively. Literally speaking, from digraph $G_\kb$, ${\cal I}(s_{\ib})$ is the set of states that link directionally to $s_\ib$, ${\cal O}(s_\ib)$ is the set of states that are linked directionally by $s_\ib$, and ${\cal I}_0 (s_\ib)$ is the set of states that link to $s_\ib$ via a so-called $0$-edge in $E_{\kb,0}$.
The sets ${\cal I}_1(s_\ib)$, ${\cal O}_0(s_\ib)$ and ${\cal O}_1(s_\ib)$ all have in a similar meaning. 

In our state machine, we associate each state $s_\ib$ with an integer.
Without ambiguity, we use $s_\ib$ to also denote the integer associated with it.
Define an operator $\Xi_\kb: \Sf \to \Sf$, which updates the value associated with each  state according to:
$$\Xi_\kb \ : \ s_\ib \leftarrow \sum_{s_\jb \in {\cal I}(s_\ib)} s_\jb\quad\text{for all }s_\ib \in \Sf.$$
It should be noted that the operator $\Xi_\kb$ updates all states in $\Sf$ in a  parallel fashion. Also, if ${\cal I}(s_\ib)$ is an empty set,
operator $\Xi_\kb$ would set $s_\ib\leftarrow 0$.
We similarly define operators $\Xi_{\kb,0}$ and $\Xi_{\kb,1}$ respectively as
$$\Xi_{\kb,0} \ : \ s_\ib \leftarrow \sum_{s_\jb \in {\cal I}_0(s_\ib)} s_\jb
\text{ and }\ \Xi_{\kb,1} \ : \ s_\ib \leftarrow \sum_{s_\jb \in {\cal I}_1(s_\ib)} s_\jb.$$

An example is provided below to help clarify these notations.

\medskip

\begin{ex} For UW $\kb=000$ of length $L=3$, there are four possible states in $\Sf = \lbrace s_{00},s_{01},s_{10},s_{11}\rbrace$.
Because $\Am_\kb^{L-1} \underline{y}_\kb =\begin{bmatrix} 0&1&0&1 \end{bmatrix}^\top$, we have $\Cf_\kb= \lbrace s_{01},s_{11} \rbrace$. 
Create the $0$-edges and $1$-edges of the digraph in Fig. \ref{AlgEx}. Table \ref{ex3} 
then  shows  ${\cal I}(s_\ib)$, ${\cal O}(s_\ib)$, ${\cal I}_0(s_\ib)$, ${\cal I}_1(s_\ib)$, ${\cal O}_0(s_\ib)$ and ${\cal O}_1(s_\ib)$ for each $s_\ib$. 

	\begin{figure}
		\begin{center}
		\includegraphics[width=0.5\textwidth]{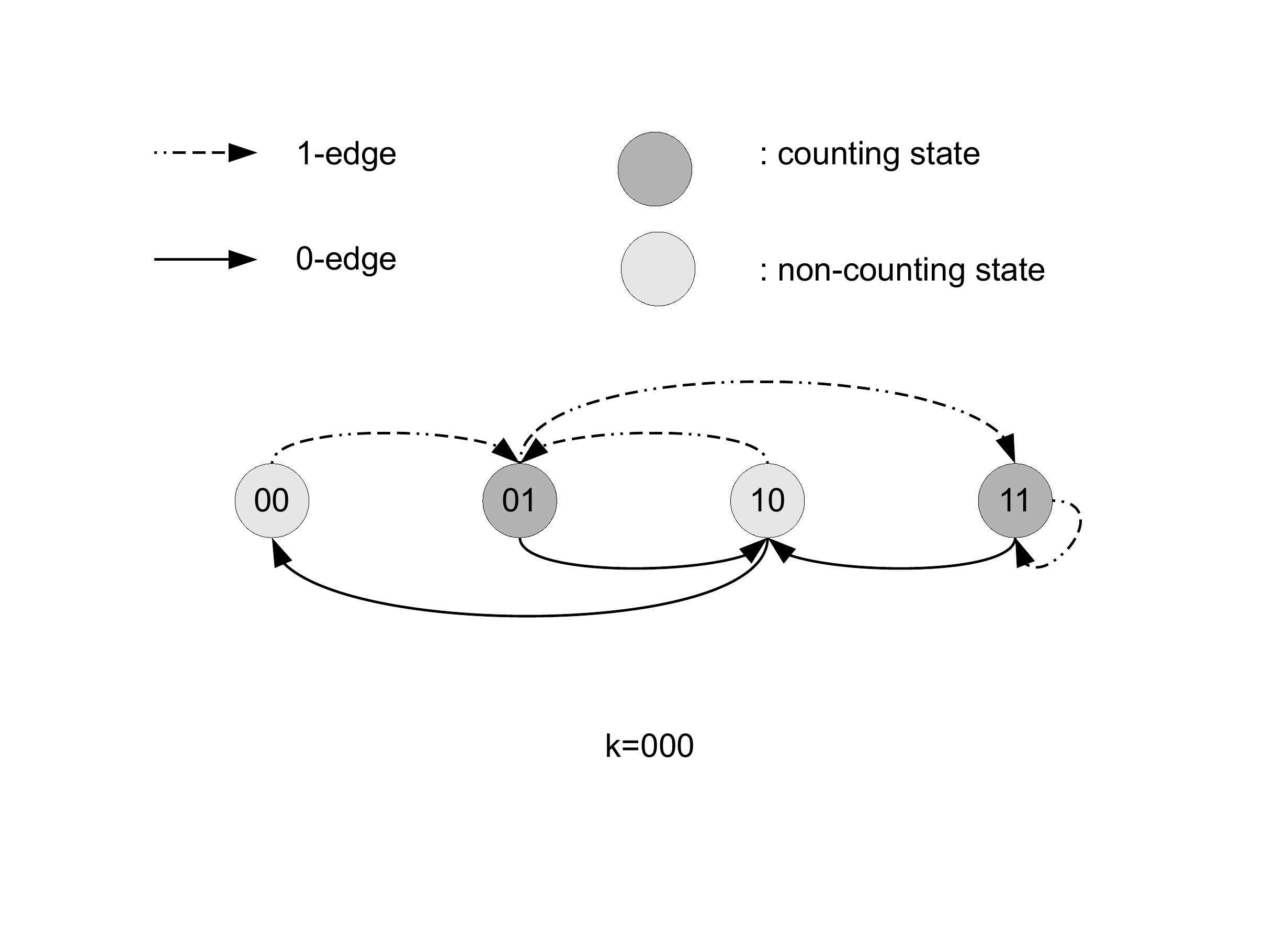}
		\caption{Digraph $G_{000}$ for UW $\kb=000$}\label{AlgEx}
		\end{center}
	\end{figure}

\begin{table}[!h]
\begin{center}
\caption{Various state sets for UW $\kb=000$}\label{ex3}
\scriptsize
\begin{tabular}{|c|c|c|c|c|}
\hline 
$s_\ib $ & $s_{00}$ & $s_{01}$ & $s_{10}$ & $s_{11}$ \\ 
\hline 
${\cal I}(s_\ib)$& $\lbrace s_{10}\rbrace$ & $\lbrace s_{00},s_{10} \rbrace$ & $\lbrace s_{01},s_{11} \rbrace$ & $\lbrace s_{01},s_{11} \rbrace$ \\ 
\hline 
${\cal O}(s_\ib)$& $\lbrace s_{01} \rbrace$ & $\lbrace s_{10},s_{11} \rbrace$ & $\lbrace s_{00},s_{01} \rbrace$ & $\lbrace s_{10},s_{11}\rbrace$ \\ 
\hline 
${\cal I}_0(s_\ib)$& $\lbrace s_{10} \rbrace$ & $\lbrace \rbrace$ & $\lbrace s_{01},s_{11}\rbrace$ & $\lbrace \rbrace$ \\ 
\hline 
${\cal I}_1(s_\ib)$& $\lbrace \rbrace$ & $\lbrace s_{00},s_{10}\rbrace$ & $\lbrace \rbrace$ & $\lbrace s_{01},s_{11}\rbrace$ \\ 
\hline 
${\cal O}_0(s_\ib)$& $\lbrace \rbrace$ & $\lbrace s_{10} \rbrace$ & $\lbrace s_{00}\rbrace$ & $\lbrace s_{10}\rbrace$ \\ 
\hline 
${\cal O}_1(s_\ib)$& $\lbrace s_{01} \rbrace$ & $\lbrace s_{11} \rbrace$ & $\lbrace s_{01}\rbrace$ & $\lbrace s_{11}\rbrace$ \\ 
\hline 
\end{tabular} 
\end{center}
\end{table}

According to the first row in Table \ref{ex3}, the operator $\Xi_\kb$ simultaneously updates all states in $\Sf$ according to
$$\Xi_\kb \ : \ \begin{array}{l}
s_{00}  \leftarrow  s_{10}\\
s_{01}  \leftarrow  s_{00} + s_{10}\\
s_{10}  \leftarrow  s_{01}+s_{11} \\
s_{11}  \leftarrow  s_{01}+s_{11}.
\end{array}$$
Likewise, the operators $\Xi_{\kb,0}$ and $\Xi_{\kb,1}$ simultaneously update all states in $\Sf$ according to
$$
\Xi_{\kb,0} \ : \ \begin{array}{l}
s_{00}  \leftarrow  s_{10}\\
s_{01}  \leftarrow  0 \\
s_{10}  \leftarrow  s_{01}+s_{11}\\
s_{11}  \leftarrow  0
\end{array}
\text{ and } \quad
\Xi_{\kb,1} \ : \ \begin{array}{l}
s_{00}  \leftarrow  0\\
s_{01}  \leftarrow  s_{00} + s_{10}\\
s_{10}  \leftarrow  0\\
s_{11}  \leftarrow  s_{01}+s_{11}.
\end{array}$$ 
\qed
\end{ex}

\medskip

With the above, we now demonstrate how to compute $c_{\kb,i}$ and $\abs{{\cal C}_\kb(\db0,n)}$ using the finite state machine. Note $c_{\kb,n}=\xb_\kb^\top \Am_\kb^{n} \underline{w}_\kb$; hence to compute $c_{\kb,n}$, the states are initialized such that $s_{k_2\ldots k_L}=1$ and $s_\ib=0$ for all remaining $\ib\neq k_2^L$.
Note that these initial values correspond exactly to the component values of  vector $\xb_\kb$. Next we apply $n$ times the operator $\Xi_\kb$ to update the states in $\Sf$. It can be seen that the resulting values of the states correspond exactly to the contents of the row vector $\xb_\kb^\top \Am_\kb^n$. Thus, the value of $c_{\kb,n}$ can be obtained by summing the values of the counting states. Again, we remark that 
we only need the finite state machine for computing the values of $c_{\kb,n}$ for $n=1,2,\ldots,L-1$, as the values of $c_{\kb,n}$ for $n \geq L$ can be easily determined by the recursion formula \eqref{eq:recdkbn}. 

On the other hand, to compute 
\[
 \abs{{\cal C}_\kb(\db 0,n)}=\underline{u}^\top \Am_{\kb,0} \Am_\kb^{n-\ell(\db)-1} \underline{w}_\kb
\] 
for a given prefix $\db$, 
we initialize 
the values associated with all states to be zero except $s_{u_{m-L+2} \cdots u_{m}}=1$,
where $u_{m-L+2} \cdots u_{m}$ is the rightmost $(L-1)$ elements in $\ub=u_1\ldots u_m=k_2^L\db$. 
Apply the operator $\Xi_{\kb,0}$ to all states in $\Sf$ once, followed by
updating all the states $(n-\ell(\db)-1)$ times via operator $\Xi_\kb$.
Then, the sum of the values of all counting states equals $\abs{{\cal C}_\kb(\db 0,n)}$.

\section{Practice and performance of UDOOCs}\label{sec:5}

In Fig.~\ref{result2345}, we compare the numbers of length-$n$ codewords
for all UWs of lengths $L=2$, $3$, $4$ and $5$.
These  numbers 
are plotted in logarithmic scale and are normalized against the number of length-$n$ codewords for the all-zero UW $\kb=0\ldots 00$ to facilitate 
their comparison.
By the equivalence relation defined in Definition \ref{defn:eqrln},
only one UW in each equivalence class needs to be illustrated.
We have the following observations.

\ben
\item The logarithmic ratio $\log_2(c_{\kb,n}/c_{\ab,n})$, where $\ab=0\ldots 00$, exhibits some transient fluctuation for $n\leq L$ but becomes a steady straight line of negative slope after $n>L$.
This hints that $c_{\kb,n}$ has a steady exponential growth when $n$ is beyond $L$.

\item The number $c_{\bb,n}$, where $\bb=00\ldots 01$, is 
always the largest among all $c_{\kb,n}$ when $n$ is small.
However, this number has an apparent trend to be overtaken by those of other UWs as $n$ grows and will be eventually smaller than the number of length-$n$ codewords for the all-zero UW. This result matches the statement of 
Theorem \ref{thm:2}. 

\item As a contrary, the number $c_{\ab,n}$ for the all-zero UW $\ab=0\ldots 00$
is the smallest among all $c_{\kb,n}$ for UWs of the same length when $n$ is small.
Although Theorem \ref{thm:2} indicates that this number will eventually be  the largest,
Fig.~\ref{result2345} shows that such would happen only when 
$n$ is very large. 

\item As a result of the two previous observations, UW $\bb=00\ldots 01$ 
perhaps remains a better choice in the compression of sources with practical number of source letters even though it is asymptotically the worst. We will confirm this inference by the later practice
of UDOOCs on a real text source from the book \emph{Alice's Adventure in Wonderland}.
\een
	
\begin{figure*}[!ht]
\centering
\begin{tabular}{cc}
\includegraphics[width=0.48\textwidth] {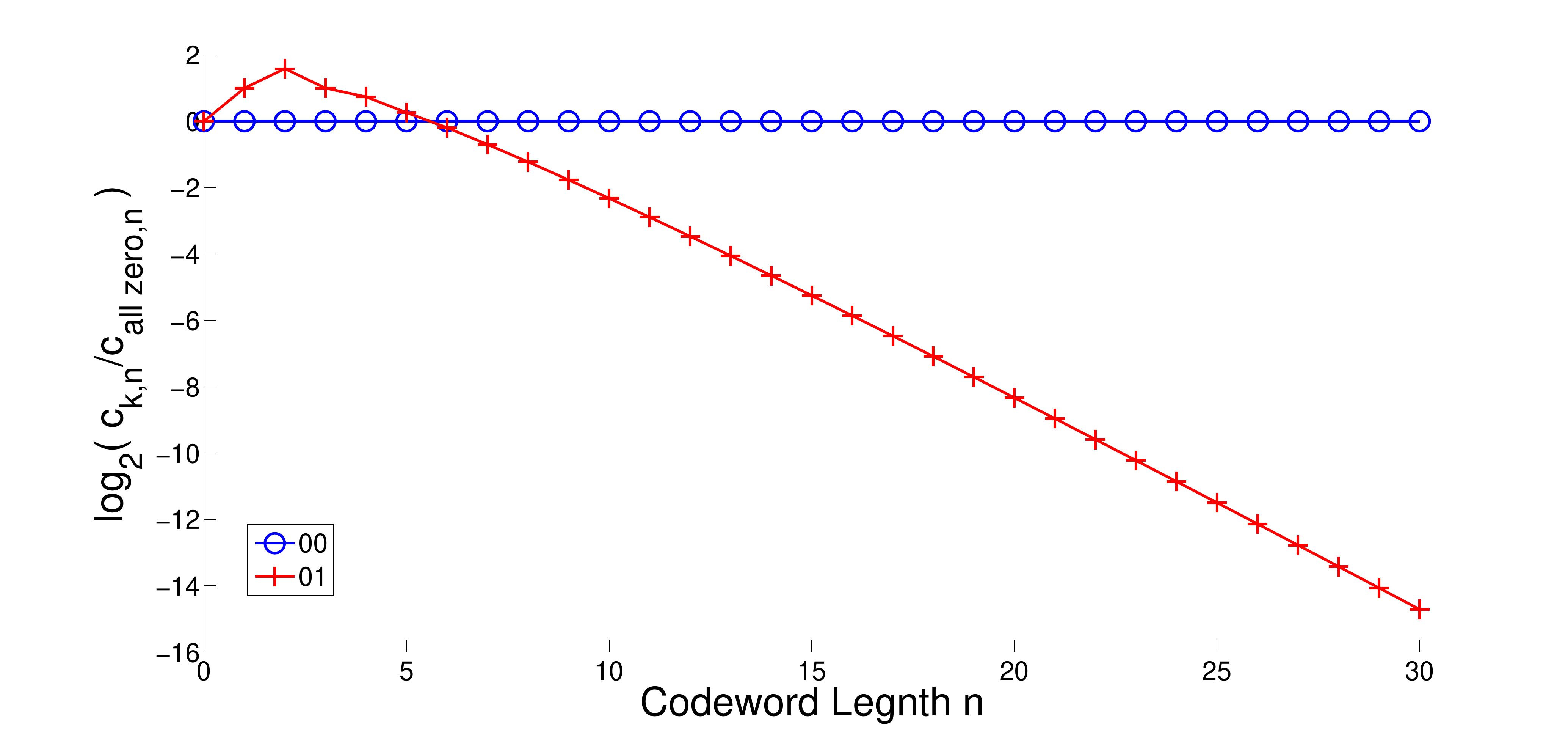}
&\includegraphics[width=0.48\textwidth]{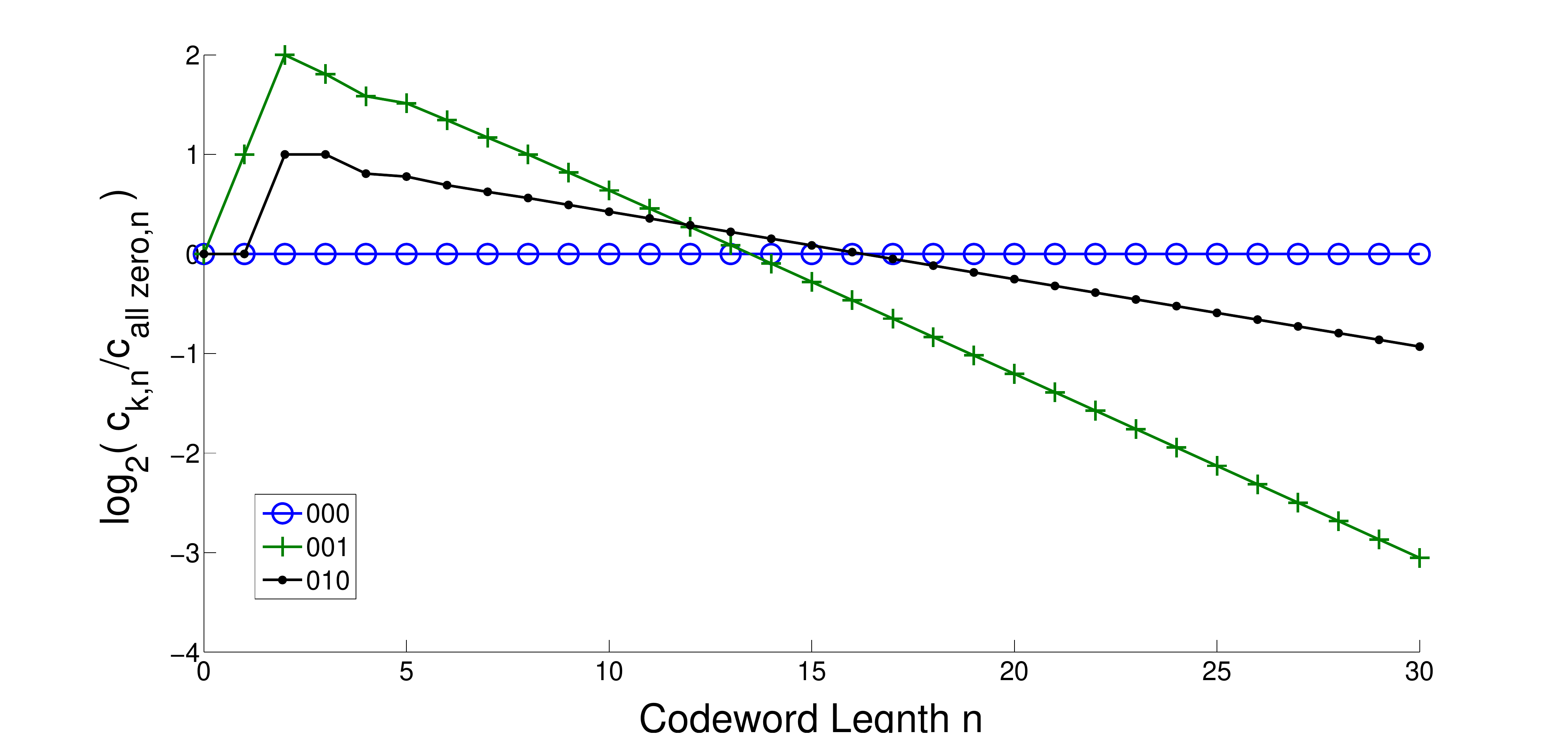}\\
{\footnotesize (a) $L=2$}&{\footnotesize  (b) $L=3$}\\[2mm]
\includegraphics[width=0.48\textwidth]{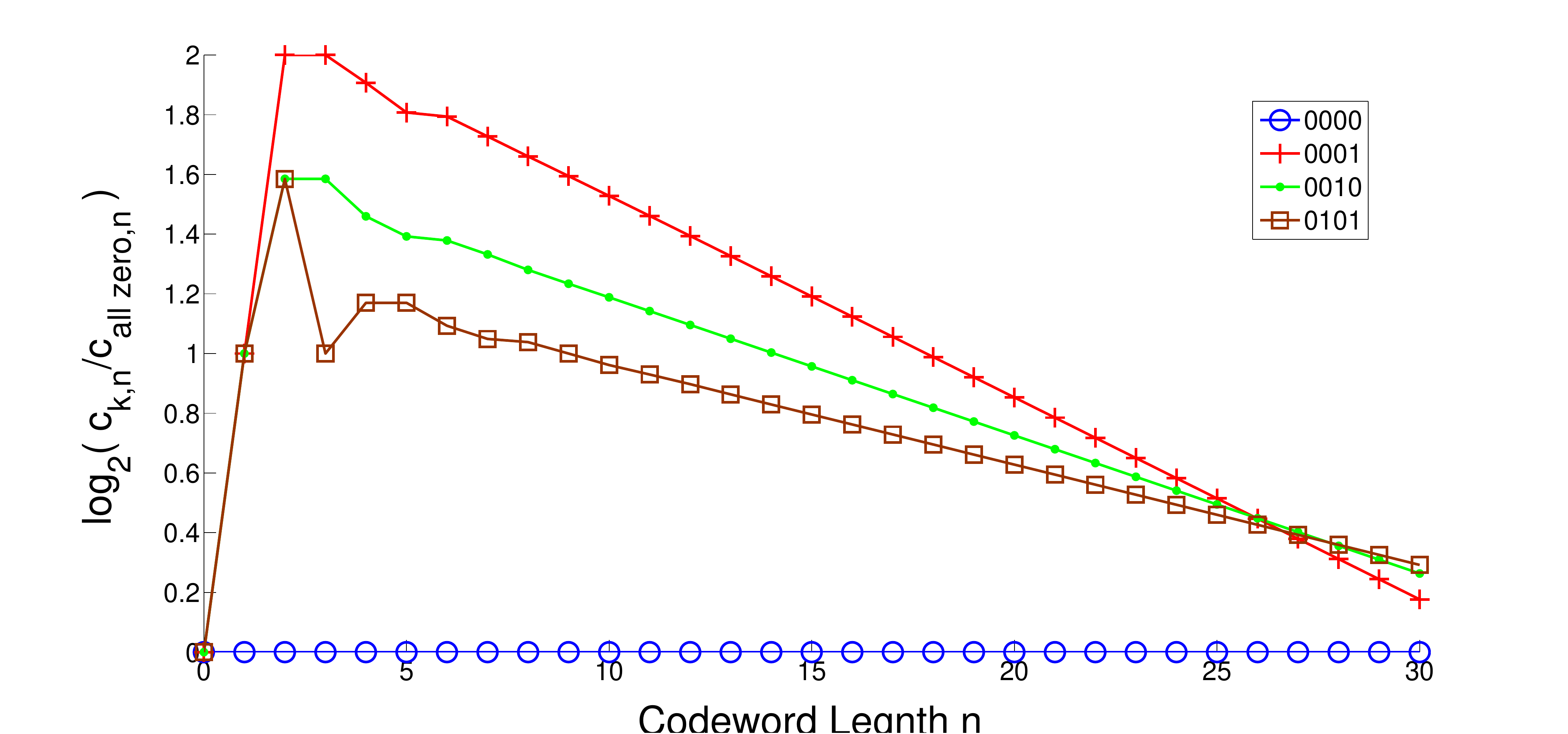}
&\includegraphics[width=0.48\textwidth]{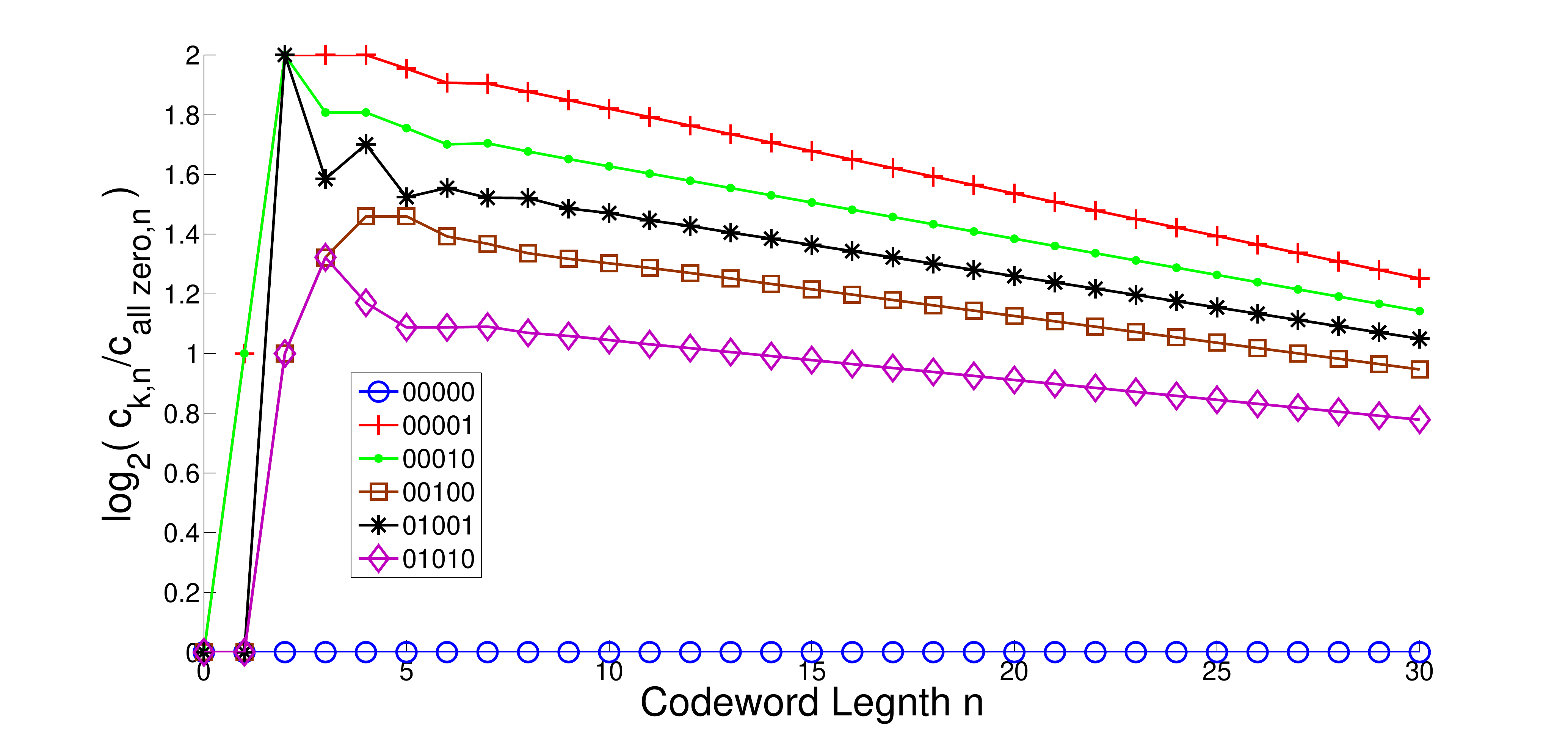}\\
{\footnotesize (c) $L=4$}&{\footnotesize (d) $L=5$}\\[2mm]
\end{tabular}
\caption{Normalized numbers of length-$n$ codewords for UWs of lengths $L=2,3,4,5$}\label{result2345}
\end{figure*}

We next investigate the compression rates of UDOOCs
and compare them with those 
of the Huffman and Lempel-Ziv (specifically, LZ77 and LZ78) codes.
In this experiment, the standard Huffman code in the communication toolbox of Matlab
is used instead of the adaptive Huffman code. The LZ77 executable is obtained from the basic compression library in \cite{LZ77},
while the LZ78 is self-implemented using C++ programing language. As a convention, 
the data is binary ASCII encoded before 
it is fed into the two Lempel-Ziv compression algorithms.
The sliding window for the LZ77 is set as $10,000$ bits,
and the tree-structured LZ78 is implemented without any windowing.

Three different English text sources are used, in which the uppercase and lowercase of each English letter are treated as the same symbol. The first English text source is distributed uniformly over the 26 symbols.
The second English text source is assumed independent and identically distributed (i.i.d.)
with marginal statistics from \cite{oxfdic}.
The third one is a realistic English text source from \emph{Alice's Adventure in 
Wonderland}, in which any symbols other than the 26 English alphabets are regarded as a ``space." In addition, the effect of grouping $t$ symbols as a grouped source for compression
is studied, which will be termed $t$-grouper in  remarks below.
The results are summarized in Tables \ref{avgtable1} and
\ref{avgtable2}, in which the average codeword length of UDOOCs has already taken into account the length of UWs.
We remark on the experimental results as follows.

\ben
\item First of all, it can be observed from Table \ref{avgtable1} that
the length-2 UW $\kb=01$ gives a good per-letter average codeword length
only when $t=1$. When the size of source alphabet increases by grouping $t=2$ or $t=3$ letters as one symbol for UDOOC compression, the per-letter average codeword length dramatically grows.
Note that $\kb=01$ is the only UW, whose number of length-$n$ codewords
has a linear growth with respect to $n$, i.e., we have $c_{01,n}=n+1$.
Since the size of source alphabets increases exponentially in $t$
when $t$-grouper is employed,
the resulting per-letter average codeword length also increases
exponentially as $t$ grows. Therefore, when $\text{UW}=01$, $t$-grouper 
will result in an extremely poor performance
for moderately large $t$.

\item By independently generating $10^6$ letters according to the statistics in \cite{oxfdic}
for compression, we record the per-letter average codeword in the second row 
of Table \ref{avgtable1}.
As expected, the Huffman coding scheme
gives the smallest per-letter average codeword length of $4.253$ bits per letter, when 
$3$-grouper is used.
The gap of per-letter average codeword lengths
between the 3-grouper Huffman and the 3-grouper UDOOC
however can be made as small as $4.795-4.253=0.542$ bits per source letter if
$\text{UW}=0001$. This is in contrast to the gap of $1.007$ bits 
when uniform independent English text source is the one to be compressed (cf.~the first row in Table \ref{avgtable1}).
We would like to point out that
the error propagation of UDOOCs
is limited firmly by at most two codewords, while
that of the Huffman code may be statistically beyond this range.
In comparison with the LZ77 and LZ78,
the UDOOC clearly performs better in compression rate
for usual independent English text source.

\item When the compression of a source with memory
such as the book titled \emph{Alice's Adventures in Wonderland} \cite{alice}
is concerned, the third row in Table \ref{avgtable1}
shows that the gap of per-letter average codeword lengths between the optimal $3$-grouper Huffman and the $3$-grouper
UDOOC with $\text{UW}=0001$ is narrowed down to $0.305$ bits per letter.
The $3$-grouper UDOOC with the all-zero UW also performs well for this source. 
Note that part of the per-letter average codeword length of UDOOCs 
is contributed by the UW, i.e., $L/t$; hence, in a sense, a larger $t$ and a smaller $L$ are favored (except for $L=2$). As can be seen from Table \ref{avgtable1}, the best compression performance is given by $t=3$, $L=4$, and $\text{UW}=0001$. 

\item For the third English text source,
the LZ77 performs better than all of the $1$-grouper UDOOC compression schemes
but one. We then compare the running time of both algorithms. We reduce the window size of LZ77 so that it has a similar running time to the $1$-grouper UDOOC scheme. The compression performance of LZ77 degrades down
to $5.234$ bits per letter, which is larger than that of the $1$-grouper UDOOC.
Note that we only compare their running time in encoding in Table \ref{avgtable2} as
the decoding efficiency of UDOOCs is seemingly better than that of the LZ77.
Considering also the low memory consumption of UDOOCs when a specific 
UW is pre-given in addition to its simplicity in implementation, the UDOOC can be regarded as a 
cost-effective compression scheme for practical applications.

\een

\begin{table*}[!t]
\caption{Average codeword lengths in bits per source symbol for the compression of three different sources. The best one among $1$-grouper, $2$-grouper and $3$-grouper
of the same compression scheme is boldfaced. 
}\label{avgtable1}
\scriptsize

\begin{center}
\begin{tabular}{ c | c c c | c | c | c c c | c | c c c | c c c }
\hline \hline
 Type & \multicolumn{3}{c|}{\bf Entropy} & \bf LZ77 & \bf LZ78 & \multicolumn{3}{c|}{\bf Huffman}&  & \multicolumn{3}{c|}{\bf UW = $00\cdots 0$} & \multicolumn{3}{c}{\bf UW = $00\cdots 01$}  \\ 
 \hline\hline
  & $t=1$ & $t=2$ & $t=3$& &  &$t=1$ & $t=2$ & $t=3$ & & $t=1$ & $t=2$ & $t=3$ & $t=1$ & $t=2$ & $t=3$\\ \cline{2-16} 
 Independent& & & & & & & & & $L=2$ & 6.961 & 6.820 & {\bf 6.790} & {\bf 5.846} & 12.76 &  41.99\\
 English Letter with& 4.700 & 4.700 & 4.700 & 7.992 & 7.178 & 4.768 & 4.738 & {\bf 4.702} & $L=4$ & 8.576 & 6.831 & {\bf 6.213} &  7.000 & 5.899 & {\bf 5.709}\\ 
 Uniform distribution & & & & & & & & & $L=6$ & 10.58 & 7.748 & {\bf 6.746} & 9.000 & 6.768 & {\bf 6.104}\\ \hline \hline
 
  & $t=1$ & $t=2$ & $t=3$& & &$t=1$ & $t=2$ & $t=3$ &  & $t=1$ & $t=2$ & $t=3$ & $t=1$ & $t=2$ & $t=3$\\ \cline{2-16}
 Independent& & & & & & & & & $L=2$ & 5.591 & {\bf 5.550} & 5.637 & {\bf 4.557} & 7.771 &  20.907\\
 English Letter with & 4.246 & 4.246 & 4.246 & 7.925 &6.626 & 4.274 & 4.261 & {\bf 4.253} & $L=4$ & 7.411 & 5.872 & {\bf 5.351} &  6.185 & 4.970 & {\bf 4.795}\\ 
 Usual distribution & & & & & & & & &$L=6$ &  9.411 & 6.818 & {\bf 5.924} & 8.185 & 5.882 & {\bf 5.274}\\ \hline \hline

& $t=1$ & $t=2$ & $t=3$ & & &  $t=1$ & $t=2$ & $t=3$ & & $t=1$ & $t=2$ & $t=3$ & $t=1$ & $t=2$ & $t=3$\\  \cline{2-16}
Alice's & & & & & & & & & $L=2$ & 4.887 & 4.340 & {\bf 3.958} & {\bf 4.068} & 4.975 & 7.573\\ 
Adventures & 3.914 & 3.570 & 3.215 &4.661 &6.028 & 3.940 & 3.585 & {\bf 3.226}& $L=4$ & 6.757 & 4.920 & {\bf 4.089} & 5.774 & 4.133 & {\bf 3.531}\\ 
in Wonderland & & & & & & & & &$L=6$ & 8.757 & 5.890 & {\bf 4.709} & 7.774 & 5.089 & {\bf 4.115}\\ \hline \hline

\end{tabular} 
\end{center}
\end{table*}

\begin{table}[!h]
\caption{Average codeword lengths in bits per source symbol and running time in seconds for the UDOOC encoding and the LZ77 encoding on \emph{Alice's Adventures in Wonderland}.
The programs are implemented using C++, and are executed in a Microsoft Windows-based desktop with intel-Core7 2.4G CUP and 8G memory.
}\label{avgtable2}
\scriptsize
\begin{center}

\begin{tabular}{c|c|c|c}
\hline\hline
\multicolumn{2}{c|}{Type} & Average Codewrod length & Running Time \\ \hline
\multirow{2}{*}{UDOOC} & UW $\kb=00$ & 4.887 & 0.0162 sec \\ \cline{2-4}
      & UW $\kb=01$ & 4.068 & 0.0158 sec\\ \hline
\multirow{2}{*}{LZ77}  & Window Size = $10^4$ bits & 4.661 & 0.0328 sec\\ \cline{2-4}
      & Window Size $= 3000$ bits & 5.234 & 0.01607 sec\\ 
\hline\hline
\end{tabular}
\end{center}
\end{table}

\section{Conclusion}\label{sec:6}

	In this paper, we have provided a general construction of UDOOCs with arbitrary UW. Combinatorial properties of UDOOCs are subsequently investigated. 
Based on our studies, the appropriate UW for the UDOOC compression of a given source
can be chosen. Various encoding and decoding algorithms for general UDOOCs,
as well as their efficient counterparts for specific UWs like $\kb=00\ldots0$, $00\ldots01$,
are also provided.
Performances of UDOOCs are then compared with the Huffman and Lempel-Ziv codes.
Our experimental results show that the UDOOC 
can be a good practical candidate for lossless data compression
when a cost-efficient solution is desired.

\appendices

\section{Proof of Theorem \ref{thm:1}} \label{app:a}

In this section, we will prove \eqref{eq:skbz}, the enumeration of $s_{\kb,n}$  in Theorem \ref{thm:1}. Our proof technique is similar to that in \cite{DIFF99}. 

Let ${\F^\infty} := \bigcup_{n \geq 0} \F^n$ be the set of all binary sequences. For a word $\wb=w_1 \ldots w_n \in {\F^\infty}$ of length $n$, let ${\cal F}_\kb(\wb)$ be the set of index pairs  indicating the places that $\wb$ contains $\kb$ as a subword, i.e., 
\[
{\cal F}_\kb(\wb) \ = \ \left\{ (i,j) \ : \ \kb = w_i^j \right\}.
\]
Further denote by $\ell(\wb)$ the length of word $\wb$.  Then 
\bean
f(z) &=& \sum_{n \geq 0} s_{\kb,n} z^n \\
&\stackrel{\text{(i)}}{=}& \sum_{\wb \in {\F^\infty}} z^{\ell(\wb)} 0^{\abs{{\cal F}_\kb(\wb)}}\\
&=& \sum_{\wb \in {\F^\infty}} z^{\ell(\wb)} \prod_{a \in {\cal F}_\kb(\wb)} (1+(-1))\\
&\stackrel{\text{(ii)}}{=} & \sum_{\wb \in {\F^\infty}}  z^{\ell(\wb)} \sum_{A \subseteq {\cal F}_\kb(\wb)} (-1)^{\abs{A}} \yes \label{eq:f1}
\eean
where  
in (i) we have adopted the convention of $0^0=1$, and (ii) follows from the inclusion-exclusion principle. In light of \eqref{eq:f1}, we will regard the pair $(\wb,A)$ with $A \subseteq {\cal F}_\kb(\wb)$ as a {\em marked word}. The set of all marked words is thus defined as
\[
{\cal M}_\kb \ := \ \left\{ (\wb, A) \ : \ \wb \in {\F^\infty}\text{ and } A \subseteq {\cal F}_\kb(\wb)\right\}.
\]
Define the following weight function for elements in ${\cal M}_\kb$
\beq
\pi(\wb,A) \ := \ z^{\ell(\wb)} (-1)^{\abs{A}};
\eeq
then \eqref{eq:f1} can be rewritten as
\beq
f(z) \ = \ \sum_{(\wb,A) \in {\cal M}_\kb} \pi(\wb,A). \label{eq:f2}
\eeq
To determine $f(z)$, below we introduce the concept of a {\em cluster}. 

\medskip

\begin{defn}[Cluster] \label{defn:cluster}
We say the marked word $(\wb,A)$ is a {\em cluster} if, and only if, 
\[
\bigcup_{(i_t,j_t) \in A}\,  [i_t,j_t] = [1,\ell(\wb)]
\]
where by $[a,b]$ we mean the closed interval $\{ x \in \R : a \leq x \leq b\}$ on the  real line. The set of all clusters is thus 
\[
{\cal T}_\kb  \ = \ \left\{ (\wb,A) \in {\cal M}_\kb \ : \ (\wb,A) \text{ is a cluster}\right\}.
\]
\end{defn}

\medskip

\begin{defn}[Concatenation of sets of marked words]
For any two sets of marked words ${\cal A}_\kb$ and ${\cal B}_\kb$, we define the concatenation of ${\cal A}_\kb$ and ${\cal B}_\kb$ as
\[
{\cal A}_\kb \vee {\cal B}_\kb  :=  \left\{ (\ab \bb, A \cup {\frak J}(B,\ell(\ab))  :  (\ab,A) \in {\cal A}_\kb, (\bb,B) \in {\cal B}_\kb \right\}
\]
where by $\ab\bb$ we meant the usual concatenation of strings $\ab$ and  $\bb$, and  the function ${\frak J}(B,\ell(\ab))$ is 
\[
{\frak J}(B,\ell(\ab)) \ := \ \left\{ (i_t + \ell(\ab), j_t + \ell(\ab)) : (i_t,j_t) \in B \right\}.
\]
\end{defn}

\medskip

Having defined the concatenation operation $\vee$
for sets of marked words, we next claim the following decomposition for the set ${\cal M}_\kb$
\beq
{\cal M}_\kb \ = \ \{(\text{null},\emptyset)\} \cup \left( {\cal M}_\kb \vee {\cal F} \right) \cup \left( {\cal M}_\kb \vee {\cal T}_\kb \right), \label{eq:decomp}
\eeq
where ${\cal F} := \{(b,\emptyset) : b \in \F\}$. 

To show \eqref{eq:decomp}, for any $(\wb,A) \in {\cal M}_\kb$ we distinguish the following three disjoint cases:
\ben
\item If $\ell(\wb)=0$, it is obvious that $\wb$ is a null word and $A = \emptyset$ from the definition of ${\cal F}_\kb(\wb)$.
\item For $\ell(\wb) \geq 1$, appending an arbitrary binary word to  $\wb$ results in another marked word $(\wb b, A)$, which cannot be a cluster since
\[
\bigcup_{(i_t,j_t) \in A} [i_t, j_t] \subset [1,\ell(\wb)+1].
\]
Conversely, take any marked word $(\wb,A)$ from ${\cal M}_\kb$ with $\ell(\wb)=n$. If $j_t < \ell(\wb)=n$ for all $(i_t,j_t) \in A$, then we can delete the rightmost bit from $\wb$, and the resulting pair $(w_1^{n-1},A)$ is still a marked word. Summarizing the above gives the following equalities between two sets of marked words
\bea
 \IEEEeqnarraymulticol{3}{l}{ \left\{ (\wb,A) \in {\cal M}_\kb: j_t < \ell(\wb) \text{ for all $(i_t,j_t) \in A)$}\right\}}\no\\
\quad &=&  \left\{ (\wb b, A) : (\wb,A)\in {\cal M}_\kb, b \in \F \right\} \nonumber\\
&=& {\cal M}_\kb \vee {\cal F}, \label{eq:case2}
\eea
where the last equality follows from the definition of concatenation operation $\vee$. 
\item The last case concerns the situation when $(\wb,A)$ 
satisfies $\ell(\wb)=n\geq 1$, $A=\{(i_1, j_1), \ldots, (i_m, j_m)\}$ and $i_1 < \cdots < i_m < j_m=n$. In other words, this is the case when $\max \{ j_t : (i_t,j_t)\in A\} = \ell(\wb)$,
which is disjoint from the second case. For this, let $u$ be the smallest index such that $[i_{u+t}, j_{u+t}] \cap [i_{u+t+1}, j_{u+t+1}] \neq \emptyset$ for all $t=0,1,\ldots,m-u+1$. Then obviously we have the following de-concatenation of $(\wb,A)$
\bean
 \IEEEeqnarraymulticol{3}{l}{
(\wb,A) \ = \ (w_1^{i_u-1}, \{ (i_t, j_t): t=1,\ldots,u-1\})}\\
\quad && \vee\,  (w_{i_u}^{n}, \{ (i_t - i_u+1, j_t - i_u+1): t=u, \ldots, m\}).
\eean
Clearly, the first marked word $(w_1^{i_u-1}, \{ (i_t, j_t): t=1,\ldots,u-1\}) \in {\cal M}_\kb$. The second marked word $(w_{i_u}^{n}, \{ (i_t - i_u+1, j_t - i_u+1): t=u, \ldots, m\})$ is a cluster since 
\[
\bigcup_{t=u}^m [i_t - i_u+1, j_t - i_u+1] = [1,n-i_u+1]
\]
by the choice of $u$. Hence we arrive at the following equality between two sets of marked words
\bea
\IEEEeqnarraymulticol{3}{l}{
\left\{ (\wb, A)  \in {\cal M}_\kb \ : \ \max \{ j_t : (i_t,j_t)\in A\} = \ell(\wb) \right\}}\no\\
\quad &=&  {\cal M}_\kb \vee {\cal T}_\kb. \label{eq:case3}
\eea
\een
Combining the case of null word and equations \eqref{eq:case2} and \eqref{eq:case3} proves the desired claim of \eqref{eq:decomp}. 

Using the decomposition in \eqref{eq:decomp}, we can rewrite \eqref{eq:f2} in terms of the three sets, i.e., the set for null word, ${\cal M}_\kb \vee {\cal F}$, and ${\cal M}_\kb \vee {\cal T}_\kb$. In particular, we have 
\bean
 \IEEEeqnarraymulticol{3}{l}{ \sum_{(\wb,W) \in {\cal M}_\kb \vee {\cal T}_\kb} \pi(\wb,W)}\\
\quad &=& \sum_{(\ab,A) \in {\cal M}_\kb} \sum_{ (\bb,B) \in {\cal T}_\kb} z^{\ell(\ab \bb)} (-1)^{\abs{A \cup {\frak J}(B,\ell(\ab))}}\\
&=&  \sum_{(\ab,A) \in {\cal M}_\kb} \sum_{ (\bb,B) \in {\cal T}_\kb} z^{\ell(\ab) + \ell( \bb)} (-1)^{\abs{A}+\abs{B}}\\
&=& \left( \sum_{(\ab,A) \in {\cal M}_\kb} \pi(\ab,A) \right) \left( \sum_{(\bb,B) \in {\cal T}_\kb} \pi(\bb,B) \right). \yes \label{eq:f32}
\eean
Similarly, one can show that
\beq
\sum_{(\wb,A) \in {\cal M}_\kb \vee {\cal F}_2} \pi(\wb,A) \ = \ 2z \sum_{(\wb,A) \in {\cal M}_\kb} \pi(\wb,A). \label{eq:f33}
\eeq
Substituting \eqref{eq:f32} and \eqref{eq:f33} into \eqref{eq:f2} gives 
\[
f(z) = \sum_{(\wb,A) \in {\cal M}_\kb} \pi(\wb,A) = 1 + 2z f(z) + f(z) T(z),
\]
or equivalently,
\beq
f(z) \ = \ \frac{1}{1 - 2 z - T(z)} \label{eq:f4},
\eeq
where $T(z)$ is the weight enumerator of elements in ${\cal T}_\kb$ given by 
\beq
T(z) \ := \ \sum_{(\bb,B) \in {\cal T}_\kb} \pi(\bb,B) . \label{eq:Tz}
\eeq

Determining $T(z)$ is now relatively easy. Recall that the overlap function $r_\kb(i)={\bf 1}(k_1^{L-i} =  k_{i+1}^L)$, where ${\bf 1}(\cdot)$ is the usual indicator function,  shows exactly whether the length-$(L-i)$ prefix of $\kb$ is also a suffix of $\kb$. Let ${\cal R}_\kb = \left\{ i : 1 \leq i \leq L-1, r_\kb(i)=1\right\}$. For any cluster $(\bb,B) \in {\cal T}_\kb$ with $\bb=b_1 \ldots b_n$, we must have $b_{n-L+1}^n = \kb$ by Definition \ref{defn:cluster}. So for any $i \in {\cal R}_\kb$, i.e., $r_\kb(i)=1$, we have $b_{n-L+i+1}^n=k_{i+1}^L = k_1^{L-i}$. 
Hence the pair
\[
\left( \bb k_{L-i+1} \ldots k_L, B \cup \left\{ (n+i-L+1,n+i)\right\}\right) 
\]
is a cluster in ${\cal T}_\kb$. It implies that for $i \in {\cal R}_\kb$, the set
\beq
{\cal T}_{\kb,i} := \left\{ 
\begin{array}{l}
\left(\bb k_{L-i+1}^L, B \cup \left\{ (n+i-L+1,n+i)\right\}\right) : \\
\hspace{1.3in} (\bb,B)\in {\cal T}_\kb, n=\ell(\bb) 
\end{array}\right\} \label{eq:Tkbi}
\eeq
is a subset of ${\cal T}_\kb$.

On the other hand, take any $(\bb,B) \in {\cal T}_\kb$ with $\ell(\bb)=n$ and $B=\{(i_t, j_t): t=1, \ldots, m\}$, where $1=i_1 < i_2 < \cdots < i_m < j_m=n$ and $i_m=n-L+1$. If $m=1$, then $\bb=\kb$ and $B=\{(1,L)\}$. Hence we consider the case when $m > 1$. As $(\bb,B)$ is a cluster, $[i_{m-1},j_{m-1}] \cap [i_m, j_m] \neq \emptyset$ and $b_{i_{m-1}}^{j_{m-1}} = b_{i_m}^{j_m}=\kb$. Therefore, we must have $b_{i_m}^{j_{m-1}} = k_1^{v}=k_{L-v+1}^L$, where $v=j_{m-1}-i_m+1$. Thus, $r_\kb(L-v)=1$ and $(\bb,B) \in {\cal T}_{\kb,L-v}$. The above discussion then gives the following decomposition for ${\cal T}_\kb$
\beq
{\cal T}_\kb \ = \  \{ \left(\kb,\{(1,L)\} \right) \} \cup \left( \bigcup_{i \in {\cal R}_\kb} {\cal T}_{\kb,i} \right). \label{eq:decompT}
\eeq
For enumerating the weights of elements in ${\cal T}_\kb$, we further claim that ${\cal T}_{\kb,i} \cap {\cal T}_{\kb,j} = \emptyset$ for all $i \neq j$. This simply follows from the definition of ${\cal T}_{\kb,i}$ in \eqref{eq:Tkbi} that for any $(\bb,B) \in {\cal T}_{\kb,i}$ and $(\bb',B') \in {\cal T}_{\kb,j}$, say $B=\{(i_t,j_t): t=1, \ldots, m\}$ and $B'=\{(i_t,j_t): t=1, \ldots, m'\}$,  where the pairs $(i_t,j_t)$ are arranged in ascending order,
we have that $j_m - j_{m-1}=i$ for $B$ and $j_{m'}-j_{m'-1}=j$ for $B'$. This proves our claim. Finally, using \eqref{eq:decompT} and the fact that the sets $\{{\cal T}_{\kb,i}\}$ are disjoint, we obtain 
\bean
T(z) &=&  \pi(\kb,\{(1,L)\}) + \sum_{i=1}^{L-1} r_\kb(i)  \sum_{(\bb,B) \in {\cal T}_\kb,i} \pi(\bb,B)\\
&=& z^{\ell(\kb)} (-1) + \sum_{i=1}^{L-1} r_\kb(i)  \sum_{(\bb,B) \in {\cal T}_\kb} z^{\ell(\bb)+i} \left( -1 \right)^{\abs{B}+1}\\
&=& - z^L - \sum_{i=1}^{L-1} r_\kb(i) z^i T(z).
\eean
Hence
\[
T(z) \ = \ -\frac{z^L}{1+ \sum_{i=1}^{L-1} r_\kb(i) z^i}.
\]
Substituting the above into \eqref{eq:f4} proves \eqref{eq:skbz} of Theorem \ref{thm:1}. 

\section{Degree of $ \det\left( \mI - \Am_\kb z \right)$} \label{app:b}

In this section, we will determine the degree of  polynomial  $\det\left( \mI - \Am_\kb z \right)$ that is required in the proof of Theorem \ref{thm:1}.

\medskip

\begin{prop} \label{prop:deg}
Let $\Am_\kb$ be the adjacency matrix for the digraph $G_\kb$ associated with UW $\kb$ defined in Section \ref{sec:3}. Then
\beq
\deg \det\left( \mI - \Am_\kb z \right) \ = \ L.
\eeq
\end{prop}
\begin{proof}
First, from \eqref{eq:skbnz1} and \eqref{eq:skbz}, the two equivalent
formulas for the enumeration of $s_{\kb,n}$, we see $\det\left( \mI - \Am_\kb z \right)$ is divisible by $h_\kb(z)=(1-2z)(1+\sum_{i=1}^{L-1} r_\kb(i)z^i)+z^L $. It follows that 
\[
\deg \det\left( \mI - \Am_\kb z \right) \geq L.
\]
To establish the converse of the above inequality, i.e., $\deg\det\left( \mI - \Am_\kb z \right) \leq L$, it suffices to show that $\rank(\Am_\kb ^{L-1}) \leq L$, which in turns implies $\rank(\Am_\kb ^L) \leq L$. As a result, the algebraic multiplicity of eigenvalue $0$ for $\Am_\kb$ is at least $2^{L-1}-L$. Hence, the degree of $\det(\mI-\Am_\kb z)$ is at most $L$. 

To prove the claim, given the UW $\kb=k_1 \ldots k_L$ of length $L$ and the corresponding adjacency matrix $\Am_\kb$ for digraph $G_\kb$, let 
\[
\Hm = \Am_\kb + \eb_{\kb_1} \eb_{\kb_2}^\top
\]
where $\kb_1 = k_1^{L-1}$ and $\kb_2 = k_2^L$, and where by $\eb_\db \in \F^{2^{L-1}}$ with $\db =d_1 \ldots d_{L-1} \in \F^{L-1}$ we mean $(\eb_\db)_{j+1}=1$ if $j$ has the binary representation $\db$, and $(\eb_\db)_{j+1}=0$, otherwise. 

Apparently, $\Hm$ is the adjacency matrix for the digraph without UW forbidden constraint and is therefore independent of the choice of $\kb$. As an example, if $L=3$, then 
\[
\Hm = \begin{bmatrix}
1 & 1 & 0 & 0 \\ 
0 & 0 & 1 & 1 \\ 
1 & 1 & 0 & 0 \\ 
0 & 0 & 1 & 1
\end{bmatrix}. 
\] 
Furthermore, it can be easily verified that $\Hm^{L-1}=\underline{\bf 1}\, \underline{\bf 1}^\top$ is the all-one matrix. Armed with the above, we now have 
\begin{eqnarray}
\Am_\kb ^{L-1} &=& \left(\Hm-\eb_{\kb_1} \eb_{\kb_2}^\top\right) ^{L-1}\nonumber\\
&=& \Hm^{L-1} - \sum_{i=0}^{L-2} \Hm^{L-2-i} \left( \eb_{\kb_1} \eb_{\kb_2}^\top\right) \Am_\kb^i, \label{eq:AkLdeg}
\end{eqnarray}
where the last equality is due to the following identity for square matrices $A$ and $B$:
\[
(A-B)^{L-1} \ = \ A^{L-1} - \sum_{i=0}^{L-2} A^{L-2-i} B (A-B)^i.
\]
Applying the standard rank inequality of $\rank(A+B) \leq \rank(A) + \rank(B)$ \cite{HJ} to \eqref{eq:AkLdeg} yields 
\bean
 \IEEEeqnarraymulticol{3}{l}{\rank\left(\Am_\kb ^{L-1}\right)}\\
\quad & \leq&  \rank\left( \Hm^{L-1} \right) + \sum_{i=0}^{L-2} \rank \left( \Hm^{L-2-i} \left( \eb_{\kb_1} \eb_{\kb_2}^\top\right) \Am_\kb^i  \right)\\
&=& 1 + \sum_{i=0}^{L-2} 1 = L,
\eean
and the proof is completed. 
\end{proof}

\section{Verification of Algorithms \ref{alg:5} and \ref{alg:6}} \label{app:c}

For completeness, we verify Algorithms \ref{alg:5} and \ref{alg:6} in this section.

For message $u_1$, i.e., the most likely message, we have from line 1 in Algorithm \ref{alg:5} that $m=1$ and $n=0$ since $F_{\kb,0}=c_{\kb,0}=1$. This results in the encoding output of the null codeword. In parallel, when receiving the null codeword, we have $n=0$. Algorithm \ref{alg:6} then sets $m=1$ at line 2 as $F_{\kb,-1}=0$.
This verifies the correctness of Algorithms \ref{alg:5} and \ref{alg:6} for message $u_1$.

For $m\geq 2$, we shall show that for each $n \geq 1$,  the encoding function $\phi_\kb$ is a bijection
between ${\cal U}_\kb(n)=\{u_m : F_{\kb,n-1} < m \leq F_{\kb,n}\}$
and ${\cal C}_\kb(n)$, and the decoding function $\psi_\kb$ is the functional inverse of $\phi_\kb$. Equivalently, it suffices to show that 
\begin{enumerate}
\item $\psi_\kb$ is a bijection
between ${\cal C}_\kb(n)$ and ${\cal U}_\kb(n)$ for each $n \geq 1$, and 
\item $\phi_\kb$ is the functional inverse of $\psi_\kb$
\end{enumerate}
We will proceed with this approach. 

Prior to establishing the claims, we first introduce below a well-ordering of binary sequences. This is in fact a key concept embedded in Algorithms \ref{alg:5} and \ref{alg:6}. 

\begin{defn}[Lexicographical ordering] For any two binary sequences $\ab=a_1 \ldots a_i$  and $\bb=b_1 \ldots b_j$, we say $\ab \succ \bb$ if $i > j$, or if $i=j$ and there exists a smallest integer $s$, $1 \leq s \leq i$, such that $a_u=b_u$ for $u=1,\ldots,s-1$, $a_s=1$, and $b_s=0$. 
\end{defn}

\medskip

Obviously, such ordering is a total-ordering of binary sequences. How the lexicographical ordering of binary sequences plays a key role in the encoding and decoding of UDOOCs is due to the following lemma. 

\begin{lemma} \label{lem:1}
For any two length-$n$ codewords $\ab,\bb \in {\cal C}_\kb(n)$, we have $\ab \succ \bb$ if, and only if,
\beq
\sum_{i=1}^{n} a_i \abs{{\cal C}_\kb(a_1^{i-1} 0, n)} > \sum_{i=1}^{n} b_i \abs{{\cal C}_\kb(b_1^{i-1} 0, n)}. \label{eq:lem1}
\eeq
\end{lemma}
\begin{proof}
As $\ell(\ab) = \ell(\bb)$ and $\ab \succ \bb$, there exists a smallest integer $s$, $1 \leq s \leq n$, such that $a_u=b_u$ for $u=1,\ldots,s-1$, $a_s=1$, and $b_s=0$. Thus,
\bean
 \IEEEeqnarraymulticol{3}{l}{\sum_{i=1}^{n} a_i \abs{{\cal C}_\kb(a_1^{i-1} 0, n)}} \\
\quad &\geq& \sum_{i=1}^{s-1} a_i \abs{{\cal C}_\kb(a_1^{i-1} 0, n)} +  \abs{{\cal C}_\kb(a_1^{s-1} 0, n)} \\
& > & \sum_{i=1}^{s-1} a_i \abs{{\cal C}_\kb(a_1^{i-1} 0, n)} + \sum_{i=s+1}^{n} b_i \abs{{\cal C}_\kb(a_1^{s-1} 0 b_{s+1}^{i-1}0, n)} \\
&=&   \sum_{i=1}^{n} b_i \abs{{\cal C}_\kb(b_1^{i-1} 0, n)},
\eean
where the second inequality follows from the fact that the sets ${\cal C}_\kb(a_1^{s-1} 0 b_{s+1}^{i-1}0, n)$, where $i=s+1, \ldots, n$ and $b_i=1$, are disjoint proper subsets of ${\cal C}_\kb(a_1^{s-1} 0, n)$. 
\end{proof}

\medskip

With the above lemma, given a codeword $\cb=c_1 \ldots c_n$, Algorithm \ref{alg:6} 
outputs $\psi_\kb(\cb)=m$ with 
\bea
m &=& \sum_{i=1}^{n} c_i \xb_\kb^\top \left( \prod_{j=1}^{i-1}\Am_{\kb,c_j} \right)\Am_{\kb,0} \Am_\kb^{(n+L-1)-i} \underline{y}_\kb + F_{\kb,n-1}+1 \no\\
&=& \sum_{i=1}^{n} c_i \abs{{\cal C}_\kb(c_1^{i-1} 0, n)} + F_{\kb,n-1}+1. \label{eq:appceq2}
\eea
We remark that the first term in the above, i.e., $\sum_{i=1}^{n} c_i \abs{{\cal C}_\kb(c_1^{i-1} 0, n)}$, is the only term dependent on $\cb$, and it also appears in \eqref{eq:lem1}. It means that the encoding and decoding algorithms of UDOOC given in Algorithms \ref{alg:5} and \ref{alg:6} are indeed based on the lexicographical ordering of length-$n$ codewords in ${\cal C}_\kb(n)$. 
Using Lemma \ref{lem:1} we can establish the range of $\psi_\kb$ when restricted to ${\cal C}_\kb(n)$. 

\begin{cor}
The range of $\psi_\kb$ when restricted to ${\cal C}_\kb(n)$ is the set ${\cal U}_\kb(n)=\{u_m : F_{\kb,n-1} < m \leq F_{\kb,n}\}$. Therefore, $\psi_\kb$ is a bijection between ${\cal C}_\kb(n)$ and ${\cal U}_\kb(n)$ for all $n \geq 1$. 
\end{cor}
\begin{proof}
Given ${\cal C}_\kb(n)$, let $\bb$ be the smallest member and $\db$ be the largest member according to the lexicographical ordering, i.e. $\bb \preceq \cb \preceq \db$ for all $\cb \in {\cal C}_\kb(n)$. It then follows from Lemma \ref{lem:1} that 
\[
\min_{\cb \in {\cal C}_\kb(n)} \psi_\kb(\cb) \ = \ \psi_\kb(\bb) \quad \text{ and } \quad \max_{\cb \in {\cal C}_\kb(n)} \psi_\kb(\cb) \ = \ \psi_\kb(\db) .
\]
For the minimum, from \eqref{eq:appceq2} we have 
\[
\psi_\kb(\bb) \ = \  \sum_{i=1}^{n} b_i \abs{{\cal C}_\kb(b_1^{i-1} 0, n)} + F_{\kb,n-1}+1.
\]
Since $\bb$ is the smallest member, it follows that for all $i$, $i=1,\ldots,n$, $ \abs{{\cal C}_\kb(b_1^{i-1} 0, n)}=0$ if $b_i=1$. Hence 
\[
\min_{\cb \in {\cal C}_\kb(n)} \psi_\kb(\cb) \ = \ \psi_\kb(\bb)  \ = \ F_{\kb,n-1}+1.
\]
To see the maximum, again from \eqref{eq:appceq2} 
\[
\psi_\kb(\db) \ = \  \sum_{i=1}^{n} d_i \abs{{\cal C}_\kb(d_1^{i-1} 0, n)} + F_{\kb,n-1}+1.
\]
Since $\db$ is the largest member in ${\cal C}_\kb(n)$, the sets  ${\cal C}_\kb(d_1^{i-1} 0, n)$, where $i=1, \ldots, n$ and $d_i=1$, are disjoint and proper subsets of ${\cal C}_\kb(n)$. Moreover, for any $\cb \in {\cal C}_\kb(n)$ and $\cb \prec \db$, there exists a smallest integer $s$, $1 \leq s \leq n$, such that $d_u=c_u$ for $u=1,\ldots,s-1$, $d_s=1$, and $c_s=0$. This in turn implies $\cb \in {\cal C}_\kb(d_1^{s-1} 0, n)$. Therefore,
\color{black}
\[
\bigcup_{i=1\atop d_i=1}^n {\cal C}_\kb(d_1^{i-1} 0, n) \ = \ {\cal C}_\kb(n) \setminus \{ \db\}
\]
and
\[
\psi_\kb(\db) \ = \  c_{\kb,n}-1 + F_{\kb,n-1}+1 = F_{\kb,n}.
\]
Finally, noting that $\abs{{\cal C}_\kb(n)} = \abs{{\cal U}_\kb(n)}$ and that $\psi_\kb$ is injective by Lemma \ref{lem:1}, we conclude that $\psi_\kb$ is bijective.
\end{proof}

\medskip

So far we have established the first claim that $\psi_\kb$ is a bijection
between ${\cal C}_\kb(n)$ and ${\cal U}_\kb(n)$. To prove the second claim that $\phi_\kb$ is the functional inverse of $\psi_\kb$, given a codeword $\cb=c_1 \ldots c_n$, Algorithm \ref{alg:6} outputs 
\[
m = \psi_\kb(\cb)=  \sum_{i=1}^{n} c_i \abs{{\cal C}_\kb(c_1^{i-1} 0, n)} + F_{\kb,n-1}+1
\]
and $F_{\kb,n-1} < m \leq F_{\kb,n}$. 
Line 2 of Algorithm \ref{alg:5} would produce the correct $n$ for $m$. Then, from line 3 of Algorithm \ref{alg:5}, we get 
\[
\rho_0 = \sum_{i=1}^{n} c_i \abs{{\cal C}_\kb(c_1^{i-1} 0, n)} +1.
\]
For the loop of lines 3-10 of Algorithm \ref{alg:5}, when $i=1$, \emph{dummy} has value 
\[
\emph{dummy} \ = \ \underline{x}_\kb^\top \Am_{\kb,0} \Am_{\kb}^{n+L-2} \underline{y}_\kb = \abs{{\cal C}_\kb(0, n)}. 
\]
We distinguish two cases:
\begin{enumerate}
\item if $c_1=0$, then we must have 
\[
\rho_0 = \sum_{i=2}^{n} c_i \abs{{\cal C}_\kb(0c_2^{i-1} 0, n)} +1 \leq \emph{dummy}
\]
since $\sum_{i=2}^{n} c_i \abs{{\cal C}_\kb(0c_2^{i-1} 0, n)}$ is the sum of the cardinalities of certain disjoint subsets (with different prefixes) of ${\cal C}_\kb(0, n)$. Hence lines 5-9 of Algorithm \ref{alg:5} output $c_1=0$ as desired. 
\item if $c_1=1$, then
\[
\rho_0 = \abs{{\cal C}_\kb(0, n)} + \sum_{i=2}^{n} c_i \abs{{\cal C}_\kb(1c_2^{i-1} 0, n)} +1 > \emph{dummy}
\]
and lines 5-9 of Algorithm \ref{alg:5} gives the correct $c_1=1$. 
\end{enumerate}
Furthermore, it can be seen that at the end of line 9, we have 
\[
\rho_1 = \sum_{i=2}^{n} c_i \abs{{\cal C}_\kb(c_1^{i-1} 0, n)} +1 
\]
for the next iteration. Now suppose we are at the $t$th iteration of Algorithm \ref{alg:5} for some integer $t$ with $1 < t < n$. We have already determined $c_1, c_2, \ldots, c_{t-1}$, and have 
\[
\rho_{t-1} = \sum_{i=t}^{n} c_i \abs{{\cal C}_\kb(c_1^{i-1} 0, n)} +1. 
\]
Line 4 of Algorithm \ref{alg:5} then gives
\bean
\emph{dummy} &=& \underline{x}_\kb^\top \left( \prod_{i=1}^{t-1} \Am_{\kb,c_i} \right)\Am_{\kb,0} \Am_{\kb}^{(n+L-1)-t} \underline{y}_\kb\\
& =& \abs{{\cal C}_\kb(c_1^{t-1}0, n)}. 
\eean
Using the same reasoning as the above it can be easily shown that lines 5-9 of Algorithm \ref{alg:5} always produce the correct value for $c_t$. Finally at the $n$th iteration we have 
\[
\rho_{n-1} = c_n \abs{{\cal C}_\kb(c_1^{n-1} 0, n)} +1
\]
and
\[
\emph{dummy} = \underline{x}_\kb^\top \left( \prod_{i=1}^{n-1} \Am_{\kb,c_i} \right)\Am_{\kb,0} \Am_{\kb}^{L-1} \underline{y}_\kb= \abs{{\cal C}_\kb(c_1^{n-1}0, n)}. 
\]
It should be noted that $c_1^{n-1}0$ is a length-$n$ word, hence $\emph{dummy} =0$ or $1$. We distinguish the following cases:
\begin{enumerate}
\item If $\emph{dummy} =0$, then $c_1^{n-1}0$ cannot be a valid codeword for UDOOC. Lines 5-9 of Algorithm \ref{alg:5} achieve exactly the above, since we have 
\[
\rho_{n-1} = c_n \cdot \emph{dummy} + 1 = 1 > \emph{dummy} =0
\]
and the algorithm always outputs $c_n=1$. 
\item If $\emph{dummy} =1$, then $\rho_{n-1}  = c_n + 1$. The same reasoning as the above shows that lines 5-9 of Algorithm \ref{alg:5} always produce the correct value for $c_n$. 
\end{enumerate}

We therefore complete the proof that $\phi_\kb$ is the functional inverse of $\psi_\kb$. 

\section{$\lim_{t\rightarrow\infty}L_{\kb,t}$ for all-zero UW
and uniform i.i.d.~source} \label{app:d}

Let $\ab=00\ldots0$ be the all-zero UW of length $L$. From \eqref{ha}, \eqref{eq:case1ckbn1}, and \eqref{eq:case1ckbn2}, it can be easily verified that 
\beq
\sum_{n=0}^\infty c_{\ab,n}z^n=1+\dfrac{z}{1-\sum_{i=1}^L z^i}. \label{eq:z0}
\eeq
Furthermore, from \eqref{ha} we have $g(z) = (1-z) h_\ab(z) = 1-2z+z^{L+1}$. It is straightforward to show that the two polynomials $g(z)$ and $\frac{d}{dz}g(z)$ are co-prime to each other; hence there are no repeating zeros in $h_\ab(z)$. It then implies that all the nonzero eigenvalues of $\Am_{\ab}$ are simple. 

Denote by $\lambda_1 \cdots \lambda_L$ the nonzero eigenvalues of $\Am_{\ab}$,
and assume without loss of generality that $|\lambda_1|>|\lambda_2|\geq \cdots\geq |\lambda_L|.$
Then
\[
c_{\ab.n}=\delta_n+\sum_{i=1}^L a_i(\lambda_i)^n
\]
where $a_1 \cdots a_L$ are constants such that \eqref{eq:z0} holds. We can also obtain the closed-form expression for $F_{\ab,n}$ as
\[
F_{\ab,n}=1+\sum_{i=1}^L a_i \dfrac{\lambda_i^{n+1}-1}{\lambda_i-1}, \quad \text{for all $n \geq 0$}
\]

Consider a uniform i.i.d.~source $\cal U$ of alphabet size $M$ with $M>1$. Let ${\cal U}^t$ be the grouped source obtained by grouping any $t$ source symbols (with repetition) in $\cal U$. It is clear that ${\cal U}^t$ is also a uniform i.i.d.~source. 
The per-letter average codeword length is given by
\bean
L_{\ab,t}&=&\frac{1}{t}\left( L +\sum_{i=2}^{M^t} p_i \ell(\phi(u_i))\right)\\
&=&\frac{1}{t}\left( L +   \dfrac{1}{M^t}\sum_{i=2}^{M^t}\ell(\phi(u_i))\right).
\eean
Let $N$ be the smallest integer such that $F_{\ab,N} \geq M^t > F_{\ab,N-1}$.
Then
\bean
L_{\ab,t} &\geq& \frac{1}{t}\left( L +   \dfrac{1}{F_{\ab,N}}\sum_{i=2}^{M^t}\ell(\phi(u_i))   \right)\\
&\geq&\frac{1}{t}\left( L +   \dfrac{1}{F_{\ab,N}}\sum_{i=2}^{N-1} i \cdot c_{\ab,i}  \right)\\
&\geq&\frac{1}{\log_M (F_{\ab,N})}\left( L +   \dfrac{1}{F_{\ab,N}}\sum_{i=2}^{N-1} i \cdot c_{\ab,i}  \right).
\eean
Consequently,
\bean
 \IEEEeqnarraymulticol{3}{l}{\lim_{t\rightarrow \infty} L_{\ab,t}}\\
\ & \geq& \lim_{N \rightarrow \infty} \frac{1}{\log_M (F_{\ab,N})}\left( L + \frac{1}{F_{\ab,N}}\sum_{i=2}^{N-1} i \cdot c_{\ab,i}\right)\\
&=&
\lim_{N \rightarrow \infty}\dfrac{\sum_{i=2}^{N-1} i\cdot c_{\ab,i}}{F_{\ab,N}\log_M (F_{\ab,n})}\\
&=& \lim_{N \rightarrow \infty}\! \dfrac{\displaystyle\sum_{i=1}^L a_i\dfrac{\lambda_i[(N-1)\lambda_i^{N}-N\lambda_i^{N-1}+1]}{(\lambda_i-1)^2}}{\left(\! 1\!+\!\displaystyle\sum_{i=1}^L a_i \dfrac{\lambda_i^{N+1}-1}{\lambda_i-1}\right)\!\log_M \!\!\left(\! 1\!+\!\displaystyle\sum_{i=1}^L a_i \dfrac{\lambda_i^{N+1}-1}{\lambda_i-1}\right)}\\
&=&\frac{1}{\log_M (\lambda_1)}=\frac{1}{\log_M (g_\ab)}.
\eean
This implies
\[
\lim_{t\rightarrow\infty}L_{\ab,t} \geq \frac{\log_2 (M)}{\log_2 (g_\ab)}
=\frac{\text{H}({\cal U})}{\log_2(g_\ab)}.
\]


\begin{thebibliography}{1}

\bibitem{LEN94} N.~Alon and A.~Orlitsky,
``A lower bound on the expected length of one-to-one codes,"
\emph{IEEE Trans.~Inf.~Theory}, vol.~40, no.~5, pp.~1670-1672, September 1994.

\bibitem{DIGRAPH} J.~Bang-Jensen and G.~Z.~Gutin, \emph{Theory, Algorithms and Applications}, Springer Monographs in Mathematics, 2009.

\bibitem{std} Information Technology-Telecommunications And Information Exchange Between Systems-Local and Metropolitan Area Networks-Specific Requirements-Part 11: Wireless LAN Medium Access Control (MAC) and Physical Layer (PHY) Specifications, IEEE Standard 802.11-1999.

\bibitem{GRAPH} N.~ Biggs, \emph{Algebraic Graph Theory}, Cambridge Mathematical Library, 1994.

\bibitem{LEN96} C.~Blundo and R.~D.~Prisco,
``New bounds on the expected length of one-to-one codes,"
\emph{IEEE Trans.~Inf.~Theory}, vol.~42, no.~1, pp.~246-250, January 1996.

\bibitem{LEN07} J.~Cheng, T.-K.~Huang and C.~Weidmann,
``New bounds on the expected length of optimal one-to-one codes,"
\emph{IEEE Trans.~Inf.~Theory}, vol.~53, no.~5, pp.~1884-1895, May 2007.

\bibitem{Cover} T.~M.~Cover and J.~A.~Thomas, \emph{Elements of Information Theory},
New York, NY: John Wiley \& Sons, 1991. 

\bibitem{MATH293}
R.~Doroslova\v{c}ki, ``The set of all the words of length $n$ over any alphabet with a forbidden good subword," 
\emph{Univ.~u Novom Sadu, Zb.~Rad.~Prirod.-Mat.~Fak.~Ser.~Mat.}, 23:2, pp.~239-244, 1993.

\bibitem{MATH195}
R.~Doroslova\v{c}ki, ``The set of all the words of length $n$ over alphabet $\{0,1\}$ with any forbidden subword of length three," 
\emph{Univ.~u Novom Sadu, Zb.~Rad.~Prirod.-Mat.~Fak.~Ser.~Mat.}, 25: 2, pp.~111-115, 1995.

\bibitem{MATH98}
R.~Doroslova\v{c}ki, ``Binary $n$-Words without the subword $1010\cdots 10$," 
\emph{Novi Sad J.~Math.}, vol.~28, no.~2, pp.~127-133, 1998.

\bibitem{MATH99}
R.~Doroslova\v{c}ki, ``On binary $n$-words with forbidden 4-subwords," 
\emph{Novi Sad J.~Math.}, vol.~29, no.~1, pp.~27-32, 1999.

\bibitem{MATH00}
R.~Doroslova\v{c}ki, ``$n$-words over any alphabet with forbidden any $3$-subwords," \emph{Novi Sad J.~Math.}, vol.~30, no.~2, pp.~159-163, 2000.

\bibitem{LEN80} J.~G.~Dunham,
``Optimal noiseless coding of random variables,"
\emph{IEEE Trans.~Inf.~Theory}, vol.~IT-26, no.~3, p.~345, May 1980.

\bibitem{AMM59} Sam E.~Ganis, \emph{Notes on the Fibonacci Sequence},
Amer.~Math.~Monthly, 1959, pp.~129-130. 

\bibitem{GR} C.~Godsil and G.~F.~Royle, \emph{Algebraic Graph Theory},
Springer, 2001.

\bibitem{LNMATH} I.~Goulden and D.~M.~Jackson, ``An inversion theorem for cluster decompositions of sequences with distinguished subsequences,"
\emph{J.~London Math.~Soc}, pp.~567-576, 1979.

\bibitem{MATH81}
L.~J.~Guibas and A.~M.~Odlyzko, ``Periods in strings,"
\emph{J.~Combinatorial Theory}, series A 30, pp.~19-42, 1981.

\bibitem{HK} K.~M.~Hoffman and R.~Kunze, \emph{Linear Algebra},
2nd edition, Pearson,1971.

\bibitem{HJ} R.~A.~Horn and C.~R.~Johnson, \emph{Matrix Analysis}, 2nd edition,
Cambridge University Press, 2012.

\bibitem{DIFF10}
E.~J.~Kupin and D.~S.~Yuster, ``Generalizations of the Goulden-Jackson cluster method," 
\emph{J.~Difference Eq.~Appl.}, 16:12, pp.~1463-1480, 2010.

\bibitem{LEN78}S.~K.~Leung-Yan-Cheong and T.~M.~Cover,
``Some equivalences between Shannon entropy and Kolmogorov complexity,"
\emph{IEEE Trans.~on Information theory}, vol.~IT-24, no.~3, pp.~331-338, May 1978.

\bibitem{SP98}
J.~Noonan, ``New upper bounds for the connective constants of self-avoiding walks," 
\emph{J.~Statistical Physics}, vol.~91, nos.~5/6, 1998.

\bibitem{DIFF99}
J.~Noonan and D.~Zeilberger, ``The Goulden-Jackson cluster method: extensions, applications, and implementations," 
\emph{J.~Difference Eq.~Appl.}, 5: pp.~355-377, 1999.

\bibitem{LEN82} J.~Rissanen, 
``Tight lower bounds for optimum code length," 
\emph{IEEE Trans.~Inf.~Theory}, vol.~IT-28, no.~2, pp.~348-349, March 1982.

\bibitem{MATH03}
E.~Rivals  and S.~Rahmann, ``Combinatorics of periods in strings,"
\emph{J.~Combinatorial Theory}, series A 104, pp.~95-113, 2003.

\bibitem{LEN08J} S.~A.~Savari,
``On one-to-one codes for memoryless cost channels,"
\emph{IEEE Trans.~Inf.~Theory}, vol.~54, no.~1, pp.~367-379, January 2008.

\bibitem{LEN04} S.~A.~Savari and A.Naheta,
``Bounds on the expected cost of one-to-one codes,"
\emph{IEEE International Symposium on Information Theory}, June 2004.

\bibitem{EC} R.~Stanley, \emph{Enumerative Combinatorics}, vols.~1 and 2, Cambridge Studies in Advanced Mathematics, 2011.

\bibitem{LEN08O} W.~Szpankowski,
``A one-to-one code and its anti-redundancy,"
\emph{IEEE Trans.~Inf.~Theory}, vol.~54, no.~10, pp.~4762-4766, October 2008.

\bibitem{LEN09} W.~Szpankowskia and S.~Verd\,{u},
``Minimum expected length of fixed-to-variable lossless compression of memoryless sources,"
\emph{IEEE International Symposium on Information Theory}, Seoul, Korea, July 2009.

\bibitem{LEN86} E.~I.~Verriest, 
``An achievable bound for optimal noiseless coding of a random variable," \emph{IEEE Trans.~Inf.~Theory}, vol.~IT-32, no.~4, pp.~592-594, July 1986.

\bibitem{DIFF06}
X.~Wen, ``The symbolic Goulden-Jackson cluster method,"
\emph{J.~Difference Eq.~Appl.}, 11:2, pp.~173-179, 2006.

\bibitem{LEN72} A.~D.~Wyner,
``An upper bound on the entropy series,"
\emph{Inf.~Control}, vol.~20, 30: pp.~176-181, 1972.

\bibitem{LZ77}  http://bcl.comli.eu/download-en.html

\bibitem{alice} http://corpus.canterbury.ac.nz/descriptions/

\bibitem{oxfdic} http://oxforddictionaries.com/words/what-is-the-frequency-of-the-letters-of-the-alphabet-in-english
\end{thebibliography}
\end{document}